\documentclass[twoside,11pt]{article}

\usepackage{jmlr2e}

\usepackage{amssymb}
\usepackage{amsfonts}
\usepackage{mathtools}
\usepackage{algorithm}
\usepackage{algcompatible}
\usepackage{thmtools,thm-restate}
\usepackage{comment}
\usepackage{ifpdf}
\usepackage{bm}
\usepackage{color}
\usepackage[displaymath,mathlines]{lineno}
\usepackage{scalerel}
\usepackage{stackengine,wasysym}
\usepackage{url}
\usepackage{hyperref}
\urldef\myhomepage\url{https://math.la.asu.edu/~slan}

\usepackage[stable]{footmisc}
\usepackage{enumerate}
\usepackage{enumitem}
\usepackage{tikz-cd}
\usetikzlibrary{arrows,positioning}
\tikzset{
  shift left/.style ={commutative diagrams/shift left={#1}},
  shift right/.style={commutative diagrams/shift right={#1}}
}
\usepackage{subcaption}
\usepackage{enumitem}
\usepackage{booktabs}
\usepackage{pifont}
\newcommand{\cmark}{\ding{51}}%
\newcommand{\xmark}{\ding{55}}%

\newtheorem{lem}{Lemma}
\newtheorem{prop}{Proposition}

\newtheorem{dfn}{Definition}
\newtheorem{asump}{Assumption}

\newtheorem{rk}{Remark}

\newcommand{\vect}[1]{\boldsymbol{#1}}
\newcommand{\tp}[1]{{#1}^{\mathsf T}}
\renewcommand{\bar}{\overline}
\newcommand{\eps}{\epsilon}
\newcommand{\pa}{\partial}

\renewcommand{\eps}{\varepsilon}
\renewcommand{\epsilon}{\varepsilon}
\renewcommand{\Sigma}{\varSigma}

\newcommand{\bbP}{\mathbb P}
\newcommand{\bbE}{\mathbb E}
\newcommand{\E}{\mathrm E}
\newcommand{\V}{\mathrm{Var}}
\newcommand{\Cov}{\mathrm{Cov}}

\newcommand{\tr}{\mathrm{tr}}
\newcommand{\VEC}{\mathrm{vec}}

\DeclareMathAlphabet\mathbfcal{OMS}{cmsy}{b}{n}
\DeclareMathOperator{\diag}{diag}
\newcommand{\half}{\frac12}

\newcommand{\bzero}{{\bf 0}}
\newcommand{\by}{{\bf y}}
\newcommand{\bY}{{\bf Y}}
\newcommand{\bx}{{\bf x}}
\newcommand{\bX}{{\bf X}}
\newcommand{\bt}{{\bf t}}
\newcommand{\bz}{{\bf z}}
\newcommand{\bZ}{{\bf Z}}
\newcommand{\bdm}{{\bf m}}
\newcommand{\bM}{{\bf M}}
\newcommand{\bC}{{\bf C}}
\newcommand{\bI}{{\bf I}}
\newcommand{\GP}{\mathcal{GP}}
\newcommand{\mC}{{\mathcal C}}
\newcommand{\mI}{{\mathcal I}}
\newcommand{\mL}{{\mathcal L}}
\newcommand{\mN}{{\mathcal N}}
\newcommand{\mMN}{\mathcal{MN}}
\newcommand{\mX}{{\mathcal X}}
\newcommand{\mT}{{\mathcal T}}
\newcommand{\mZ}{{\mathcal Z}}
\newcommand{\mD}{{\mathcal D}}

\newcommand{\mbB}{{\mathbb B}}
\newcommand{\mbH}{{\mathbb H}}
\newcommand{\bU}{{\bf U}}
\newcommand{\qprod}{{\,\dot\otimes\,}}
\newcommand{\qsum}{{\,\dot\oplus\,}}

\ifpdf
   \graphicspath{{./figure/PNG/}{./figure/PDF/}{./figure/}}
\else
   \graphicspath{{./figure/EPS/}{./figure/}}
\fi

\usepackage{lastpage}
\jmlrheading{23}{2022}{1-\pageref{LastPage}}{8/21; Revised
6/22}{8/22}{21-0952}{Shiwei Lan}

\ShortHeadings{Learning Temporal Evolution of Spatial Dependence}{Lan}
\firstpageno{1}

\begin{document}

\title{Learning Temporal Evolution of Spatial Dependence with Generalized Spatiotemporal Gaussian Process Models}

\author{\name Shiwei Lan \email slan@asu.edu \\
       \addr School of Mathematical and Statistical Sciences\\
       Arizona State University\\
       Tempe, AZ 85287, USA}

\editor{Matthias Seeger}

\maketitle

\begin{abstract}
A large number of scientific studies involve high-dimensional spatiotemporal data with complicated relationships. 
In this paper, we focus on a type of space-time interaction named \emph{temporal evolution of spatial dependence (TESD)}, which is a zero time-lag spatiotemporal covariance.
For this purpose, we propose a novel Bayesian nonparametric method based on non-stationary spatiotemporal Gaussian process (STGP).
The classic STGP has a covariance kernel separable in space and time, failed to characterize TESD.
More recent works on non-separable STGP treat location and time together as a joint variable, which is unnecessarily inefficient.
We generalize STGP (gSTGP) to introduce time-dependence to the spatial kernel by varying its eigenvalues over time in the Mercer's representation.
The resulting non-stationary non-separable covariance model bares a quasi Kronecker sum structure. 
Finally, a hierarchical Bayesian model for the joint covariance is proposed to allow for full flexibility in learning TESD.
A simulation study and a longitudinal neuroimaging analysis on Alzheimer's patients demonstrate that the proposed methodology is (statistically) effective and (computationally) efficient in characterizing TESD.
Theoretic properties of gSTGP including posterior contraction (for covariance) are also studied.
\end{abstract}


\section{Introduction}
Spatiotemporal data are ubiquitous in our daily life. For example, the climate data manifest a trend of global warming, and the traffic data feature a network structure in space and a periodic pattern in time. 
There is usually intricate interaction between space and time in such spatiotemporal process $y(\bx, t)$. In particular, the dependence among spatial locations, $\Cov(y(\bx, \cdot), y(\bx', \cdot))$, may change over time. For example, in the study of dynamic brain connectivity in neuroscience \citep{cribben12,fiecas16,lan_2019}, the spatial dependence among multi-site brain signals varies along certain cognitive processes. In the longitudinal analysis of brain images \citep{hyun2016}, different brain regions also have changing connection in the progression of diseases such as Alzheimer. 
We formally define such \emph{temporal evolution of spatial dependence (TESD)} for process $y(\bx, t)$ as a zero time-lag ($t=t'$) covariance, i.e. $\Cov(y(\bx, t), y(\bx', t))$.
In general, TESD is an important subject to understand complex relationships, to predict their progress, and to extrapolate them to unknown territory. 
In this work, we propose a novel fully Bayesian nonparametric model based on a generalized spatiotemporal Gaussian process (STGP) to characterize TESD in spatiotemporal data. 

STGP is a special type of Gaussian process that can model both spatial and temporal information simultaneously. 
There is a rich literature on STGP, of which a large body \citep{Cressie_1999, Gneiting_2002, sarkka2012,sarkka2013,niu2016} imposes stationarity condition on the covariance $\mC(\bx-\bx', t-t')$.
However, stationary STGP would not work for learning TESD ($t=t'$) because temporal stationarity directly induces TESD constant in time, i.e. $\mC(\bx-\bx', 0)$.
Consequently, ``independent and identically distributed (i.i.d.)" \citep{marco2015} is not an appropriate assumption for spatiotemporal observations with changing TESD.
From a multivariate time series point of view, these observations can be seen as a realization of a vector process $\by(t)=(y(\bx_1,t), \cdots, y(\bx_I, t))$ that has time-invariant covariance (TESD).

Classical STGP often assumes a separable structure for the kernel such as $\mC_\bx\otimes \mC_t$, e.g. \cite{Paciorek_2003,Paciorek_2006,hartikainen2011,sarkka2012,sarkka2013,niu2016,kuzin2018,TODESCATO2020}.
Separable covariance in this setting ($t=t'$) reduces to a scalar factor by a spatial kernel, i.e. $\mC_\bx(\bx,\bx') \cdot \sigma^2_t$, thus loses full freedom in describing TESD (See more details in Section \ref{sec:STGP}).
To better reflect the space-time interactions, a large class of non-separable STGP models have been proposed by parametric construction \citep{Cressie_1999, Gneiting_2002}, by spectral representation \citep{Fuentes_2008}, and by kernel convolution \citep{marco2015,Wang_2020,das2020} or mixing \citep{Fonseca_2011}, etc.
Some works \citep{datta2016,hyun2016} treat spatial and temporal variables with no difference in the joint kernel, leading to a high-dimensional dense covariance matrix which is inefficient for learning TESD (See Section \ref{sec:comparison}).

Even among non-stationary non-separable covariance models \citep{Fuentes_2008,singh2010,cressie2011,luttinen2012,marco2015,datta2016,hyun2016,senanayake2016,Zhang_2020,Wang_2020,das2020},
most of them focus on modeling and predicting mean functions but are not designed for efficient and flexible learning of TESD.

In the majority of existing works, the joint covariance kernel is built from parametric covariance functions such as exponential, Mat\'ern, or their derivations. For example, \cite{Cressie_1999, Gneiting_2002} give conditions on admissible functions for such parametric construction of stationary non-separable models and \cite{Fuentes_2008} generalize to the non-stationary non-separable case.
Although regarded as nonparametric models for mean functions, they are limited in their capability of learning TESD (See Section \ref{sec:simulation}).
Even in the works of more recently proposed deep GP \citep{Damianou_2013,Salimbeni_2017,Dunlop_2018,Zhao_2021}, covariance kernels bare very flexible structures but are parameterized by neural networks.
In general, there is a lack of flexible and efficient Bayesian non-parametric models particularly for \emph{spatiotemporal covariance learning} (e.g. TESD). This paper aims to fill the blank in the literature.

To learn TESD, we propose a time-dependent spatial kernel, $\mC_{\bx|t}$, to generalize classical separable STGP. Based on the Mercer's representation of the spatial kernel, $\mC_\bx$, we introduce the time-dependence by varying the eigenvalues of $\mC_\bx$ in time. 
Moreover, we endow an independent GP (hyper-)prior on these dynamic eigenvalues and obtain a nonparametric model to flexibly capture complex TESD (See Section \ref{sec:Td-Skernel} and Section \ref{sec:simulation}).
To respect the time-changing nature of TESD, we construct a novel non-separable kernel with a quasi Kronecker sum structure (See model II in Section \ref{sec:gSTGP}). This structure results in a highly sparse joint covariance matrix, making it amenable for efficient inference. 
As a reference, we also consider a quasi Kronecker product generalization more suitable for i.i.d. observations (See model I in Section \ref{sec:gSTGP}). This is done to justify the careful design of covariance kernel, very essential in learning TESD with STGP. 

The Mercer's representation and the Karhunen-Lo\'eve expansion were considered in the literature of spatiotemporal modeling \citep{West_1997,Wikle_1999,Wikle_2002,Fontanella_2003,cressie2011}.
Similar ideas of mixture of basis functions have recently been explored in \citep[coregionalization,][]{Banerjee_2015}, \citep[Dirichlet mixture,][]{Gelfand_2005,das2020}.
The novelty of our construction lies in:
1) the quasi Kronecker sum formulation with a balanced structure designed for efficient covariance learning (Section \ref{sec:comparison}),
2) random construction of dynamic spatial kernel flexible for modeling TESD (Section \ref{sec:rand_construction}), and
3) theoretic guarantee for the Bayesian learning of TESD (Section \ref{sec:post_contr}).
The proposed methodology also generalizes the semi-parametric approaches \citep{wilson11,fox15,lan_2019}.

Focusing on TESD in the spatiotemporal covariance learning, this work makes multi-fold contributions:
\begin{enumerate}[nosep]
\item This is a novel non-separable, non-stationary, and fully nonparametric covariance model dedicated to learning TESD in spatiotemporal analysis;
\item The separable STGP is generalized by introducing the time-dependence to the spatial kernel via the Mercer's representation;
\item It provides a systematic comparison among multiple model structures of STGP based on (quasi) Kronecker product and sum both theoretically and numerically.
\end{enumerate}
The proposed methodology sheds light on brain degradation of Alzheimer's patients in a longitudinal neuroimaging analysis.
The potential utility of this work can be found in multiple areas including the genome-wide association study (GWAS), climate change, and investment management etc.

The rest of the paper is organized as follows. In Section \ref{sec:gen_STGP}, the classical separable STGP (baseline model 0) is reviewed and generalized by introducing the time-dependence to the spatial kernel. 
Two model structures, quasi Kronecker product (referenced model I) and quasi Kronecker sum (proposed model II), in the generalized STGP (gSTGP) are proposed and compared for learning TESD. 
Then we systematically investigate various theoretic properties of gSTGP in Section \ref{sec:theory} and briefly discuss the posterior inference and the prediction of TESD in Section \ref{sec:infpred}. 
In Section \ref{sec:numerics}, we conduct a simulation study (Section \ref{sec:simulation}) and apply the proposed methodology to analyze a series of positron emission tomography (PET) brain images of Alzheimer's patients (Section \ref{sec:ADPET}) obtained from the Alzheimer's Disease Neuroimaging Initiative \citep{ADNI}.
By comparing with various non-stationary/non-separable models, we
illustrate the effectiveness and efficiency of the proposed gSTGP in modeling and predicting TESD of spatiotemporal processes. Finally we conclude in Section \ref{sec:conclusion} with a few comments on the methodology and some discussions of future directions.

\section{Generalizing Spatiotemporal Gaussian Processes}\label{sec:gen_STGP}
In this section, we first define the (separable) STGP using the matrix normal distribution and explain why it fails to characterize TESD.
This motivates the generalization of STGP to introduce the time-dependent spatial kernel. 
Two model structures based on quasi Kronecker product and sum are constructed and compared but the latter is found to be more theoretically effective and computationally efficient in learning TESD.

Let $\mX\subset \mathbb R^d$ be a bounded spatial domain and let $\mT\subset \mathbb R_+$ be a bounded temporal domain. Denote $\mZ :=\mX\times \mT$ as the joint domain and $\bz:=(\bx,t)$ as the joint variable.
The spatiotemporal data $\{y_{ij}\,| i=1,\cdots,I;\, j=1,\cdots,J\}$ are taken on a grid of points $\{\bz_{ij}=(\bx_i,t_j)\, |\bx_i\in\mX, t_j\in\mT\}$ with the spatial discrete size $I$ and the temporal discrete size $J$.
A (centered) STGP is uniquely determined by its covariance kernel $\mC_\bz : \mZ\times\mZ \rightarrow \mathbb R$, a bilinear symmetric positive-definite function.
$\mC_\bz$ could be defined by 
exploring structures in space and time, to be detailed below.

\subsection{Spatiotemporal Gaussian Process}\label{sec:STGP}
The spatiotemporal data $\{y_{ij}\}$ have noise usually assumed i.i.d. and are often modeled using the standard STGP model:
\begin{equation}\label{eq:stgpm}
\begin{aligned}
y_{ij} &= f(\bx_i,t_j) + \eps_{ij}, \quad \eps_{ij} \overset{iid}{\sim} \mN(0,\sigma^2_\eps) \\
f(\bz) &\sim \GP(0, \mC_\bz)
\end{aligned}
\end{equation}
where the joint spatiotemporal kernel $\mC_\bz$ 
is in the form of a Kronecker product of spatial kernel $\mC_\bx$ and temporal kernel $\mC_t$:
\begin{equation}\label{eq:sepkern}
\textrm{model 0 (separable)}: \qquad \mC_\bz = \mC_\bx \otimes \mC_t : \mZ\times\mZ \rightarrow \mathbb R,\quad (\bz, \bz') \mapsto \mC_\bx(\bx, \bx') \cdot \mC_t(t, t')
\end{equation}

With such separable kernel, we define the classical STGP through the matrix normal distribution as follows.
\begin{dfn}[STGP]
A stochastic process $f(\bx,t)$ is called (separable) spatiotemporal Gaussian process with a mean function $m(\bx,t)$, a spatial kernel $\mC_\bx$ and a temporal kernel $\mC_t$ 
if for any finite collection of locations $\bX=\{\bx_i\}_{i=1}^I$ and times $\bt=\{t_j\}_{j=1}^J$,
\begin{equation}
{\bf F} = f(\bX,\bt) = [f(\bx_i, t_j)]_{I\times J} \sim \mMN_{I\times J}( \bM, \bC_\bx, \bC_t)
\end{equation}
where $\bM = m(\bX, \bt) = [m(\bx_i, t_j)]_{I\times J}$, $\bC_\bx = \mC_\bx(\bX,\bX) = [\mC(\bx_i, \bx_{i'})]_{I\times I}$,  $\bC_t = \mC_t(\bt,\bt) = [\mC(t_j, t_{j'})]_{J\times J}$, and the matrix normal distribution $\mMN$ is interpreted through the vectorization as in the following remark.
We denote $f\sim \GP(m,\mC_\bx,\mC_t)$.
\end{dfn}
\begin{rk}
If we vectorize the matrix ${\bf F}_{I \times J}$, then we have
\begin{equation*}
\VEC({\bf F}) \sim \mN(\VEC(\bM), \bC_t\otimes\bC_\bx),\quad
\VEC(\tp{\bf F}) \sim \mN(\VEC(\tp\bM), \bC_\bx\otimes\bC_t)
\end{equation*}
where $\otimes$ is the regular Kronecker product for matrices.
For this reason, we also denote a (centered) STGP as $f\sim \GP(0,\mC_\bx \otimes \mC_t)$.
\end{rk}

Denote $\sigma^2_t := \mC_t(t,t)$.
For any fixed time $t\in \mT$, the covariance of the separable STGP $f(\bx, t)$ in the space domain is reduced to
\begin{equation}
\Cov[f(\bx,t), f(\bx',t)] = \mC_\bx(\bx,\bx') \cdot \sigma^2_t, \quad \forall t\in \mT
\end{equation}
which at best changes by a scalar factor in time (without full freedom in describing the time change of spatial dependence), or at worst is constant in time (when a stationary $\mC_t$ is adopted).
Such drawback of the separable kernel makes the corresponding STGP unable to characterize TESD, defined as follows.
\begin{dfn}[TESD]
The temporal evolution of spatial dependence (TESD) of a spatiotemporal process $y(\bx, t)$ is the spatial covariance conditioned on a common time $t=t'\in \mT$, denoted as $\mC_{y|t}$:
\begin{equation}\label{eq:TESD}
\mC_{y|t}(\bx, \bx') :=\Cov[y(\bx, t), y(\bx', t)], \quad for \; \bx, \bx' \in \mX
\end{equation}
\end{dfn}
\begin{rk}
TESD can be viewed as a zero time-lag covariance in the spatiotemporal analysis. Therefore, the (weak) stationarity (in time) is an inappropriate assumption as it induces constant (hence trivial) TESD regardless of the kernel structure.
In the following we always assume non-stationarity for the process of interest unless specified otherwise.
\end{rk}
Because of the failure of separable kernel for learning TESD, we are motivated to remedy the classical STGP with time-dependence in the spatial covariance.

\subsection{Generalized Spatiotemporal Gaussian Process}\label{sec:gSTGP}
In this subsection, we generalize the classical STGP by introducing time-dependence to the spatial kernel $\mC_\bx$
while keeping the desirable structure in the joint kernel $\mC_\bz$.
Two structures are constructed based on quasi Kronecker product (model I) and sum (model II) respectively.
We promote model II as an effective and efficient tool for learning TESD and use model I as a reference to emphasize the importance of kernel design in STGP.

First, it is intuitive to replace the spatial $\mC_\bx$ kernel with a time-dependent analogy $\mC_{\bx|t}$ in the separable kernel \eqref{eq:sepkern} to derive the following non-separable joint kernel
\begin{equation}
\mC_\bz = \mC_{\bx|t} \qprod \mC_t
\end{equation}
where $\mC_{\bx|t}$ and the quasi Kronecker product $\qprod$ will be defined in Section \ref{sec:Td-Skernel}.
If we view $\{y_{ij}\}$ as observations with i.i.d. noise taken from a spatiotemporal process $y(\bz)$,
then the marginal covariance for $y$ in \eqref{eq:stgpm} becomes
\begin{equation}\label{eq:margcov1}
\textrm{model I (qKron-prod)}: \qquad \mC_y^\text{I} = \underbrace{\mC_{\bx|t} \qprod \mC_t}_{prior} \; + \; \underbrace{\sigma^2_\eps \mI_\bx \otimes \mI_t}_{likelihood}
\end{equation}
where $\mI_\bx(\bx,\bx')=\delta(\bx=\bx')$, and $\mI_t(t,t')=\delta(t=t')$ with $\delta(\cdot)$ being the Dirac function.

Alternatively, \cite{lan_2019} take a perspective from multivariate time series and model the noise in $\{y_{ij}\}$ as independent but not identically distributed (i.n.i.d.) (a more reasonable assumption for nonconstant TESD):
\begin{equation}\label{eq:vGPR}
\begin{aligned}
\by(t_j) &:= \{y(\bx_i, t_j)\}_{i=1}^I = {\bf f}(t_j) + \vect\eps_j, \quad \vect\eps_j \overset{inid}{\sim} \mN({\bf 0},\vect\Sigma_{t_j}) \\
f_i(t) &\overset{iid}{\sim} \GP(0, \mC_t), \quad i=1,\cdots, I.
\end{aligned}
\end{equation}
where the parametric covariance matrix $\vect\Sigma_t$ encodes the spatial dependence among multiple time series changing with time $t$.
If we replace $\vect\Sigma_t$ with a time-dependent spatial kernel $\mC_{\bx|t}$, then \eqref{eq:vGPR} can be generalized to a fully non-parametric model:
\begin{equation}\label{eq:fullBayesCOV}
\begin{aligned}
y(\bx, t)| m,\mC_{\bx|t} &\sim \GP_\bx(m, \mC_{\bx|t} \qprod \mI_t) \\
m(\bx, t) &\sim \GP_t(0, \mI_\bx \otimes \mC_t)
\end{aligned}
\end{equation}
which has the marginal covariance for $y$ in a form of quasi Kronecker sum (see Section \ref{sec:Td-Skernel}):
\begin{equation}\label{eq:margcov2}
\textrm{model II (qKron-sum)}: \qquad \mC_y^\text{II} = \underbrace{\mI_\bx \otimes \mC_t}_{prior} \; + \; \underbrace{\mC_{\bx|t} \qprod \mI_t}_{likelihood} =: \mC_{\bx|t} \qsum \mC_t
\end{equation}


\begin{figure}[t] 
   \centering
   \includegraphics[width=1\textwidth,height=.35\textwidth]{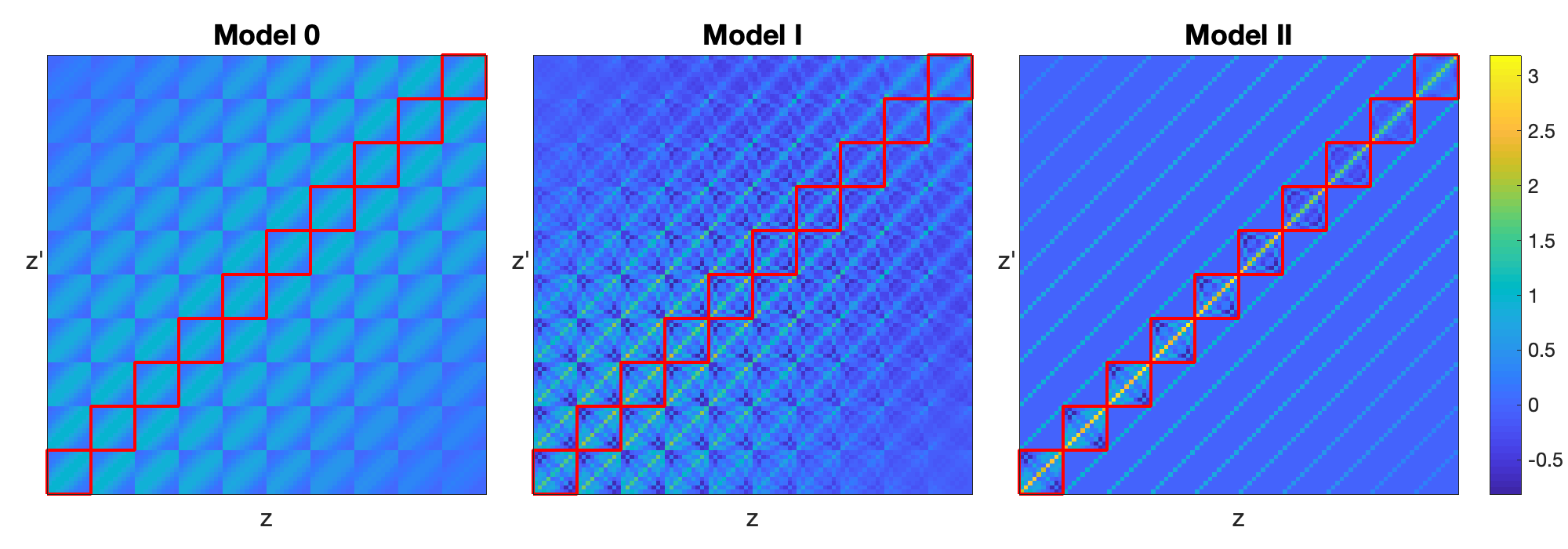} 
   \caption{Joint kernels $\mC_y$ specified by model 0 with separable structure $\mC_\bx\otimes \mC_t$ (left), model I with quasi Kronecker product structure $\mC_{\bx|t}\qprod\mC_t$ (middle) and model II with quasi Kronecker sum structure $\mC_{\bx|t}\qsum\mC_t$ (right).
   Blocks in red frame illustrate the temporal evolution of spatial kernel $\mC_{y|t}$ (TESD). Note $\bz=(\bx,t)$.}
   \label{fig:jtkern}
\end{figure}
Figure \ref{fig:jtkern} illustrates the structures of $\mC_y$ in the three models. 
If we arrange the joint state as $\bZ=\bX\otimes{\bf t}$, each of the blocks in red frame indicates a covariance matrix in the space $\bX$ at a time $t_j$, i.e. $\Cov[y(\bX,t_j), y(\bX,t_j)]_{I\times I}$, and together they describe how $\mC_{y|t}$ evolves in time (TESD).
From the time horizon, TESD can be viewed as a kernel (matrix) valued function (blocks in red frame) of time $t$.
Model 0 has constant TESD $\mC_{y|t}$ (same blocks). On the other hand, both models I and II incorporate the time-dependence for the spatial kernel through $\mC_{\bx|t}$ (changing blocks).
However, by allowing for a complex structure in the blocks off the main diagonal (not in red frame), model I would become overly complicated and expensive to fit. By contrast, model II provides enough flexibility in the main diagonal blocks (in red frame) for TESD, and yet results in a highly sparse joint kernel. See more comparison in Section \ref{sec:comparison}.
The following proposition states that the time-dependent spatial kernel $\mC_{\bx|t}$ is the \emph{essence} of TESD $\mC_{y|t}$ in these models.

\begin{prop}\label{prop:essence}
If $y$ is a spatiotemporal process according to one of the three models 
\eqref{eq:sepkern} \eqref{eq:margcov1} \eqref{eq:margcov2},
then we have the following TESD's 
\begin{equation}
\mC_{y|t}^\mathrm{0} \equiv\underbrace{\mC_\bx\sigma^2_t}_{prior}+\underbrace{\sigma^2_\eps\mI_\bx}_{likelihood}, \quad 
\mC_{y|t}^\mathrm{I} =\underbrace{\mC_{\bx|t}\sigma^2_t}_{prior}+\underbrace{\sigma^2_\eps\mI_\bx}_{likelihood}, \quad 
\mC_{y|t}^\mathrm{II} =\underbrace{\mI_\bx \sigma^2_t}_{prior} + \underbrace{\mC_{\bx|t}}_{likelihood}\\
\end{equation}
\end{prop}

\subsection{Construction of Time-Dependent Spatial Kernel}\label{sec:Td-Skernel}
In both models I \eqref{eq:margcov1} and II \eqref{eq:margcov2}, the time-dependent spatial kernel $\mC_{\bx|t}$ is the key to model TESD.
In this subsection, we construct $\mC_{\bx|t}$ by dynamically varying eigenvalues of the spatial kernel $\mC_\bx$ in the Mercer's theorem.

First, the centered (spatial) GP $\GP(0,\mC_\bx)$ is determined by its covariance kernel $\mC_\bx$
which defines a Hilbert-Schmidt integral operator on $L^2(\mX)$ as follows:
\begin{equation}\label{eq:intop}
T_{\mC_\bx}: L^2(\mX) \rightarrow L^2(\mX), \quad \phi(\cdot) \mapsto \int \mC_\bx(\cdot, \bx') \phi(\bx') d\bx'
\end{equation}
Denote $\{\lambda_\ell^2,\phi_\ell(\bx)\}$ as the eigen-pairs of $T_{\mC_\bx}$ such that $T_{\mC_\bx} \phi_\ell(\bx)=\lambda_\ell^2 \phi_\ell(\bx)$.
Then $\{\phi_\ell(\bx)\}$ serves as an orthonormal basis for $L^2(\mX)$. 
By the Mercer's theorem, we have the following representation of the spatial kernel $\mC_\bx$:
\begin{equation}\label{eq:mercer}
\mC_\bx(\bx, \bx') = \sum_{\ell=1}^\infty \lambda_\ell^2 \phi_\ell(\bx) \phi_\ell(\bx') 
\end{equation}
where the series converges in $L^2(\mX)$ norm.

\subsubsection{Deterministic Construction}
To introduce the time-dependence to the spatial kernel, thus denoted as $\mC_{\bx|t}$, we let the eigenvalues $\{\lambda_\ell^2\}$ change with time and denote them as $\{\lambda_\ell^2(t)\}$. 
That is, we define
\begin{equation}
\lambda_\ell^2(t):=\langle \phi_\ell(\bx), \mC_{\bx|t} \phi_\ell(\bx) \rangle, \quad T_{\mC_{\bx|t}} \phi_\ell(\bx)=\lambda_\ell^2(t) \phi_\ell(\bx)
\end{equation}
Let $\lambda(t):=\{\lambda_\ell(t)\}_{\ell=1}^\infty$ for $\forall t\in \mT$. 
To ensure the well-definedness of the generalization, we make the following assumption 
which essentially requires $\mC_{\bx|t}$ to be a trace ($\tr(\mC_{\bx|t})=\sum_{\ell=1}^\infty \lambda_\ell^2(\cdot)$) class operator in $L^1(\mT)$.
\begin{asump}\label{asmp:dyntrace}
We assume $\lambda_\ell\in L^2(\mT)$ for each $\ell\in \mathbb N$ and the infinite sequence $\lambda$ satisfy
\begin{equation}\label{eq:dyntrace}
\lambda \in \ell^2(L^2(\mT)), \quad i.e. \quad \Vert\lambda\Vert_{2,2}^2:= \sum_{\ell=1}^\infty \Vert\lambda_\ell(\cdot)\Vert_2^2<+\infty
\end{equation}
\end{asump}

Under Assumption \ref{asmp:dyntrace} we can have the following time-dependent spatial kernel $\mC_{\bx|t}$ well-defined through the series representation as in the Mercer's theorem (See Theorem \ref{thm:wellpose}):
\begin{equation}\label{eq:spatkern_t1}
\mC_{\bx|t}(\bx, \bx') = \sum_{\ell=1}^\infty \lambda_\ell^2(t) \cdot \phi_\ell(\bx) \phi_\ell(\bx') 
\end{equation}
%
where the eigen-basis $\{\phi_\ell(\bx)\}$ can be chosen as the eigen-functions of the spatial operator $\mC_\bx$ as in \eqref{eq:mercer} or other (e.g. Fourier) basis functions.
Such construction has some similarity to the ``coregionalization" model \citep{Banerjee_2015} for which the spatial process $\bY(\bx)={\bf A} {\bf w}(\bx)$ has a temporal covariance ${\bf T}={\bf A}\tp{\bf A}$ and independent component processes $w_\ell(\cdot) \sim \GP(0,\rho_\ell)$.
Denote ${\bf T}_\ell={\bf a}_\ell \tp{\bf a}_j$ with ${\bf a}_\ell$ as the $\ell$-th column of ${\bf A}$. Thus ${\bf T}=\sum_{\ell=1}^L {\bf T}_\ell$ is finite-dimensional. The resulting covariance $\Cov(\bY(\bx),\bY(\bx'))=\sum_{\ell=1}^L {\bf T}_\ell \cdot \rho_\ell(\bx-\bx')$ is analogous to a finite truncation of \eqref{eq:spatkern_t1}.

With $\mC_{\bx|t}$ in \eqref{eq:spatkern_t1} we can define the quasi Kronecker product for the prior kernel $\mC_m=\mC_{\bx|t} \qprod \mC_t$ in model I and the likelihood kernel $\mC_{y|m}=\mC_{\bx|t} \qprod \mI_t$ in model II repsectively:
\begin{align}
\mC_m^\text{I}(\bz,\bz') &= \mC_{\bx|t}^\half \mC_{\bx|t'}^\half \qprod \mC_t (\bz,\bz') = \sum_{\ell=1}^\infty \lambda_\ell(t) \mC_t(t,t') \lambda_\ell(t') \phi_\ell(\bx) \phi_\ell(\bx') \label{eq:jtkern} \\
\mC_{y|m}^\text{II}(\bz,\bz') &= \mC_{\bx|t} \qprod \mI_t (\bz,\bz') =\sum_{\ell=1}^\infty \lambda_\ell(t) \mI_t(t,t') \lambda_\ell(t')  \phi_\ell(\bx) \phi_\ell(\bx') \label{eq:likern}
\end{align}

\subsubsection{Random Construction}\label{sec:rand_construction}
\begin{figure}[t] 
   \centering
   \includegraphics[width=1\textwidth,height=.6\textwidth]{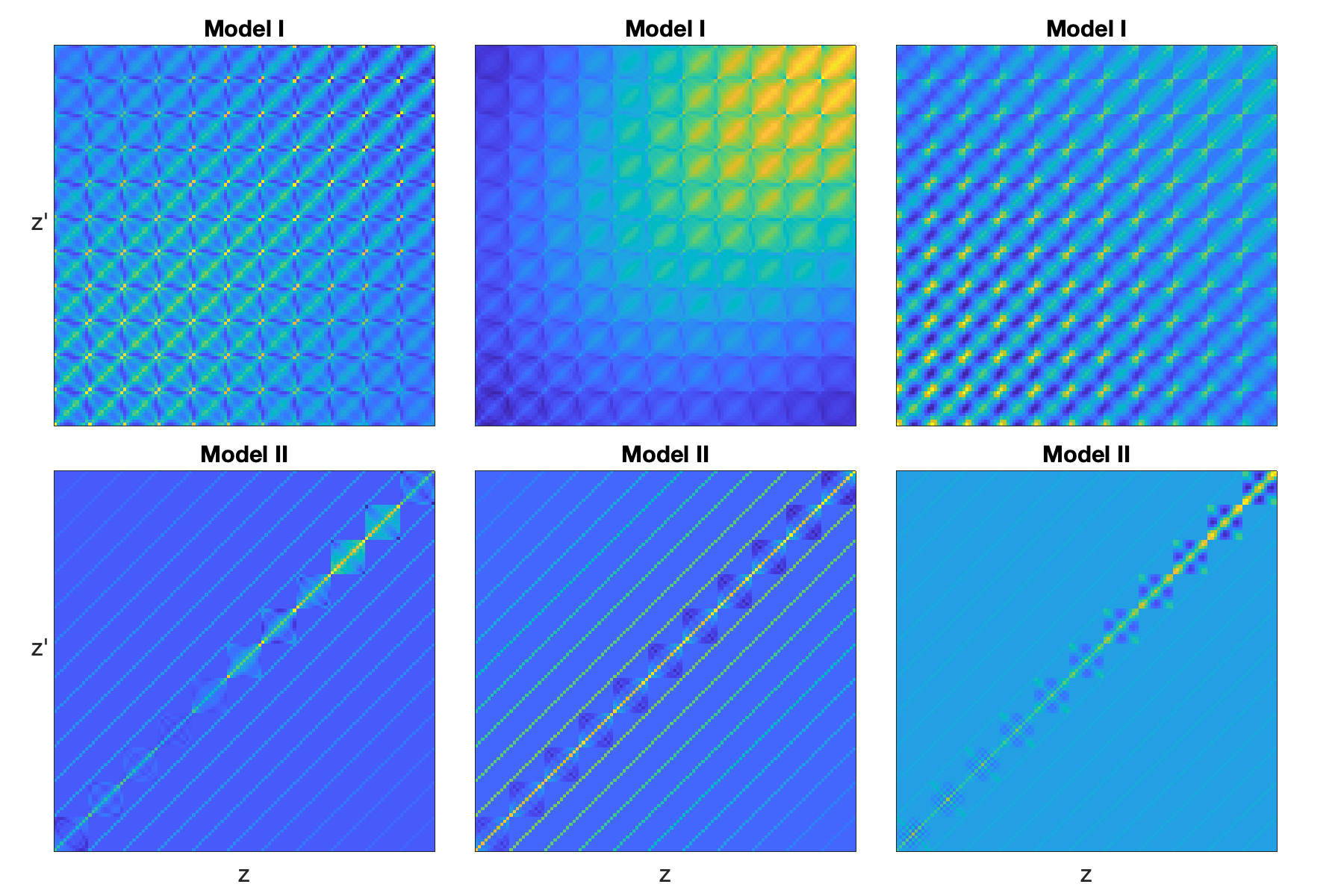} 
   \caption{Random draws of the joint covariance $\mC_y$ by model I (upper row) and model II (lower row) respectively.}
   \label{fig:randcov}
\end{figure}
To avoid parametric construction of $\lambda_\ell$, we adopt a full Bayesian approach and endow a (hyper) GP prior $\lambda_\ell(\cdot) \sim \GP(0, \mC_{\lambda,\ell})$ for each $\ell\in \mathbb N$, similarly as what is done for the mean function $m$ in standard GP regression.
Now we consider random $\lambda$ in the probability space $(\Omega, \mathcal B(\Omega), \bbP)$ with $\Omega=\ell^2(L^2(\mT))$, $\sigma$-algebra $\mathcal B(\Omega)$ and probability measure $\bbP$
defined as finite product of Gaussian measures $\{\GP(0, \mC_{\lambda,\ell})\}$ consistently extended to infinite product by Kolmogorov extension theorem \citep[c.f. Theorem 29 in section A.2.1 of][]{Dashti_2017}.
To reduce the modeling parameters, we let $\{\mC_{\lambda,\ell}\}$ share a common factor $\mC_u$ and model them with decaying magnitude, as detailed in the following assumption.
\begin{asump}\label{asmp:randeig}
Let the random infinite sequence $\lambda\in (\ell^2(L^2(\mT)), \mathcal B, \bbP)$ satisfy
\begin{equation}
\lambda_\ell(\cdot) \sim \GP(0, \mC_{\lambda,\ell}), \quad \mC_{\lambda,\ell} = \gamma_\ell^2 \mC_u, \qquad \sum_{\ell=1}^\infty \gamma_\ell^2 < \infty
\end{equation}
In Practice, we could set for each $\ell \in \mathbb N$,
\begin{equation}\label{eq:randcoeff}
\lambda_\ell(t) = \gamma_\ell u_\ell(t), \quad u_\ell(\cdot) \overset{iid}{\sim} \GP(0, \mC_u), \quad \gamma_\ell=\ell^{-\kappa/2} \;\textrm{for some}\; \kappa>1
\end{equation}
\end{asump}
Under Assumption \ref{asmp:randeig}, Assumption \ref{asmp:dyntrace} holds in $\bbP$ because $\bbE[\Vert \lambda\Vert_{2,2}^2] = \sum_{\ell=1}^\infty \gamma_\ell^2 \bbE[\Vert u_\ell\Vert_2^2] = \tr(\mC_u) \sum_{\ell=1}^\infty \gamma_\ell^2 <\infty$. A few random joint kernels $\mC_y$ by models I (upper row) and II (lower row) in Figure \ref{fig:randcov} illustrate the flexible prior candidates for the covariance $\mC_y$ produced by gSTGP to capture complicated TESD (main diagonal blocks).

\begin{figure}[t]
\centering
\resizebox{.9\textwidth}{0.45\textwidth}{
\begin{tikzpicture}
\begin{scope}
\node[shape=circle,draw,fill=gray!60] (Y) at (0,0) {$y_{ijk}$};
\node[shape=circle,draw,scale=1.3] (m) at (-1,1.5) {$m$};
\node[shape=circle,draw,scale=.8] (Cym) at (1,1.5) {$\mC_{y|m}^0$};
\node[shape=circle,draw] (Cm) at (-1,3) {$\mC_m^0$};
\path[-angle 90,thick]
(m) edge (Y)
(Cym) edge (Y)
(Cm) edge (m);
\draw (-1.4,-.75) rectangle (1.4,.75) node[pos=.85] {\small{IJK}};
\draw (0, -1.2) node {Model 0};
\end{scope}
\begin{scope}[xshift=2in]
\node[shape=circle,draw,fill=gray!60] (Y) at (0,0) {$y_{ijk}$};
\node[shape=circle,draw,scale=1.3] (m) at (-1,1.5) {$m$};
\node[shape=circle,draw,scale=.8] (Cym) at (1,1.5) {$\mC_{y|m}^\textrm{I}$};
\node[shape=circle,draw] (Cm) at (-1,3) {$\mC_m^\textrm{I}$};
\node[shape=circle,draw,scale=1.2] (u) at (-1,4.5) {$u_\ell$};
\node[shape=circle,draw,scale=1.1] (Cu) at (-1,6) {$\mC_u$};
\path[-angle 90,thick]
(m) edge (Y)
(Cym) edge (Y)
(Cm) edge (m)
(u) edge (Cm)
(Cu) edge (u);
\draw (-1.4,-.75) rectangle (1.4,.75) node[pos=.85] {\small{IJK}};
\draw (-1.7,3.9) rectangle ++(1.4,1.25) node[pos=.85] {\small{L}};
\draw (0, -1.2) node {Model I};
\end{scope}
\begin{scope}[xshift=4in]
\node[shape=circle,draw,,scale=1.1,fill=gray!60] (Y) at (0,0) {$\bY_k$};
\node[shape=circle,draw,scale=1.3] (m) at (-1,1.5) {$m$};
\node[shape=circle,draw,scale=.8] (Cym) at (1,1.5) {$\mC_{y|m}^\textrm{II}$};
\node[shape=circle,draw] (Cm) at (-1,3) {$\mC_m^\textrm{II}$};
\node[shape=circle,draw,scale=1.2] (u) at (1,3) {$u_\ell$};
\node[shape=circle,draw,scale=1.1] (Cu) at (1,4.5) {$\mC_u$};
\path[-angle 90,thick]
(m) edge (Y)
(Cym) edge (Y)
(Cm) edge (m)
(u) edge (Cym)
(Cu) edge (u);
\draw (-1.4,-.75) rectangle (1.4,.75) node[pos=.85] {\small{K}};
\draw (.3,2.4) rectangle ++(1.4,1.25) node[pos=.85] {\small{L}};
\draw (0, -1.2) node {Model II};
\end{scope}
\end{tikzpicture}}
\caption{Graphical structures of STGP Model 0 (left), I (middle) and II (right) in \eqref{eq:gstgpm}.}
\label{fig:graph_stgp}
\end{figure}

With all these definitions, we summarize the foregoing STGP models in the following unified form and illustrate their graphic structures in Figure \ref{fig:graph_stgp}.
\begin{equation}\label{eq:gstgpm}
\begin{aligned}
&& y(\bz)| m, \mC_{y|m} &\sim \GP( m,  \mC_{y|m}) && \\
&& m(\bz) &\sim \GP(0, \mC_m) && \\
\textrm{model 0 (baseline)}:&& \mC_{y|m}^0 &= \sigma^2_\eps \mI_\bx \otimes \mI_t, \quad &\mC_m^0 &= \mC_\bx \otimes \mC_t  \\
\textrm{model I (reference)}:&& \mC_{y|m}^\textrm{I} &= \sigma^2_\eps \mI_\bx \otimes \mI_t, \quad &\mC_m^\textrm{I} &= \mC_{\bx|t}(\lambda) \qprod \mC_t  \\
\textrm{model II (proposed)}:&& \mC_{y|m}^\textrm{II} &= \mC_{\bx|t}(\lambda) \qprod \mI_t, \quad &\mC_m^\textrm{II} &= \mI_\bx \otimes \mC_t \\
&& \lambda_\ell(t) &= \gamma_\ell u_\ell(t), \quad &u_\ell(\cdot) &\overset{iid}{\sim} \GP(0, \mC_u) \; \textrm{for} \; \ell \in \mathbb N
\end{aligned}
\end{equation}

\subsection{Comparison of Two Models}\label{sec:comparison}
Before concluding this section, we compare the generalized STGP models I \eqref{eq:margcov1} and II \eqref{eq:margcov2} and explain the superiority of model II compared with model I for learning TESD.

From the modeling perspective, model I puts much more weight on mean regression than on covariance learning (See Figure \ref{fig:graph_stgp}).
The STGP model \eqref{eq:gstgpm} can be equivalently viewed as decomposing the process into mean and residual: $y(\bz) = m(\bz) + \eps(\bz)$ with $m(\bz)\sim \GP(0, \mC_m)$ and $\eps(\bz)\sim \GP(0, \mC_{y|m})$.
The residual process $\eps(\bz)$ could possibly be inter-correlated in space over time but this is not correctly reflected with the likelihood kernel $\mC_{y|m}^\textrm{I}=\sigma^2_\eps \mI_\bx\otimes \mI_t$ in model I under i.i.d. assumption.
To make it worse, with more and more data, the posterior of the covariance function may not contract to the true value as it becomes dominated by the likelihood not sufficiently modeled. 
On the other hand, model II with a balanced structure (Figure \ref{fig:graph_stgp}) between the prior kernel $\mC_m^\textrm{II}$ and the likelihood kernel $\mC_{y|m}^\textrm{II}$ achieves a good trade-off.
Moreover, the posterior concentrates on the likelihood with the kernel structure $\mC_{y|m}^\textrm{II}=\mC_{\bx|t} \qprod \mI_t$ that could correctly capture TESD given enough data (See Figure \ref{fig:jtkern} and more numerical evidence in Section \ref{sec:sim_fit}).

From the computing perspective, model II has significantly lower computational complexity than model I.
As illustrated in Figure \ref{fig:jtkern}, discretizing $\mC_y^\text{II}$ leads to a highly sparse covariance matrix formed by a block-diagonal matrix $\bC_{\bx|t}$ and a sparse matrix $\bC_t\otimes \bI_\bx$.
On the contrary, the discretized $\mC_y^\text{I}$ is in general a dense matrix that involves intensive computation.
Therefore, model II is more computationally efficient for learning TESD.
See more details in Appendix \ref{apx:compadv} and Section \ref{sec:sim_fit}.

\section{Theories}\label{sec:theory}
Now we study the theoretic properties of the time-dependent spatial kernel $\mC_{\bx|t}$ and the generalized STGP.
The key elements are the spatial basis $\{\phi_\ell(\bx)\}$ and the dynamic eigenvalues $\{\lambda_\ell(t)\}$ (could be random) for which we will make a few more assumptions.
Readers can skip or postpone reading this section without interruption.

\subsection{Well-definedness}
First we prove the well-definedness of kernels \eqref{eq:jtkern} and \eqref{eq:likern} in the Mercer's representation in the following theorem.
\begin{restatable}{thm}{wellpose}
[Mercer's Kernels]
\label{thm:wellpose}
Under Assumption \ref{asmp:dyntrace}, both $\mC_m^\text{I}=\mC_{\bx|t}^\half \mC_{\bx|t'}^\half \qprod \mC_t$ and $\mC_{y|m}^\text{II}=\mC_{\bx|t} \qprod \mI_t$ are well defined non-negative definite kernels on $\mZ$.
\end{restatable}
\begin{proof}
See Appendix \ref{apx:wellpose}.
\end{proof}

With the Mercer's kernels \eqref{eq:jtkern} and \eqref{eq:likern}, we could represent STGP in the spatial basis with random time-varying coefficients similarly as in the Karhunen-Lo\'eve theorem \citep{FUKUNAGA1990,Dong_2006}.
\begin{restatable}{thm}{KLexpan}
[Karhunen-Lo\'eve Expansion]
\label{thm:KLexpan}
Under Assumption \ref{asmp:dyntrace}, STGP $f(\bx, t)\sim \GP(0, \mC_\bz)$ has the following representation of series expansion:
\begin{equation}\label{eq:KLexpan}
f(\bx, t) = \sum_{\ell=1}^\infty f_\ell(t) \phi_\ell(\bx), \quad f_\ell(t) = \int_\mX f(\bx, t) \phi_\ell(\bx) d\bx
\end{equation}
where $\{f_\ell\}_{\ell=1}^\infty$ are random processes with 
mean functions $\bbE[f_\ell(t)] = 0$ and covariance functions as follows
\begin{itemize}
\item if $\mC_\bz=\mC_{\bx|t}^\half \mC_{\bx|t'}^\half \qprod \mC_t$, then $\bbE[f_\ell(t) f_{\ell'}(t')] = \lambda_\ell(t)\mC_t(t,t') \lambda_\ell(t') \delta_{\ell\ell'}$.
\item if $\mC_\bz=\mC_{\bx|t} \qprod \mI_t$, then $\bbE[f_\ell(t) f_{\ell'}(t')] = \lambda_\ell^2(t) \delta(t=t') \delta_{\ell\ell'}$.
\end{itemize}
\end{restatable}
\begin{proof}
See Appendix \ref{apx:KLexpan}.
\end{proof}

\subsection{Regularity of Random Functions}
For the convenience of discussion, we introduce the following general $(k,s,p)$-norm \footnote{When $k=2$, this is related to Sobolev norm in the frequency domain and Hilbert scales.} to the infinite-sequence functions $\lambda=\{\lambda_\ell\}_{\ell=1}^\infty$ for $k,s>0$ and $0<p\leq \infty$.
\begin{equation}\label{eq:kspnorm}
\Vert \lambda \Vert_{k,s,p} = \left( \sum_{\ell=1}^\infty \ell^{ks} \Vert \lambda_\ell \Vert_p^k \right)^{\frac1k}
\end{equation}
And we denote the space $\ell^{k,s}(L^p(\mT)):= \{\lambda | \Vert\lambda\Vert_{k,s,p}<+\infty\}$.
Note, the norm in Assumption \ref{asmp:dyntrace} corresponds to the special case $k=2, s=0, p=2$.
For a given spatial basis $\{\phi_\ell(\bx)\}_{\ell=1}^\infty$, there is one-one correspondence $f(\bx, t)\leftrightarrow \{f_\ell(t)\}_{\ell=1}^\infty$ in \eqref{eq:KLexpan}.
Therefore, we could also define $(k,s,p)$-norm \eqref{eq:kspnorm} for $f\in \ell^{k,s}(L^p(\mT))$.
Note, when $p=2$, $f\in \ell^{k,s}(L^2(\mT))$ with a fixed spatial basis $\{\phi_\ell(\bx)\}_{\ell=1}^\infty$ also implies $f\in \ell^{k,s}(L^2(\mZ))$ regardless of spatial basis (normalized in $L^2(\mX)$) because
$\left( \sum_{\ell=1}^\infty \ell^{ks} \Vert f_\ell(t)\phi_\ell(\bx) \Vert_2^k \right)^{\frac1k}=\left( \sum_{\ell=1}^\infty \ell^{ks} \Vert f_\ell(t) \Vert_2^k \Vert \phi_\ell(\bx) \Vert_2^k \right)^{\frac1k} = \Vert f \Vert_{k,s,2}$.
For the rest of this section, we consider the case $k=p=2$.
In the following, notation $\lesssim$ ($\gtrsim$) means ``smaller (greater) than or equal to a universal constant times".

If the dynamic eigenvalues $\lambda$ decay in order $\kappa>1$, the following proposition states that they fall in a subset of $\ell^2(L^2(\mT))$.
\begin{prop}\label{prop:regdyneig}
Under Assumption \ref{asmp:randeig}-\eqref{eq:randcoeff}, 
$\lambda\in \ell^{2,s}(L^2(\mT))$ in $\bbP$ for $s<(\kappa-1)/2$.
\end{prop}
\begin{proof}
It is straightforward to verify that
\begin{equation*}
\bbE[\Vert\lambda\Vert_{2,s,2}^2] = \sum_{\ell=1}^\infty \ell^{2s} \bbE[\Vert \lambda_\ell\Vert_2^2]
= \sum_{\ell=1}^\infty \ell^{2s} \gamma_\ell^2 \bbE[\Vert u_\ell\Vert_2^2] \lesssim \tr(\mC_u) \sum_{\ell=1}^\infty \ell^{2s-\kappa} <\infty
\end{equation*}
if $2s-\kappa<-1$, i.e. $s<(\kappa-1)/2$.
\end{proof}

To discuss the regularity of random functions drawn from STGP, we need the following assumptions on the spatial basis $\{\phi_\ell(\bx)\}_{\ell=1}^\infty$ 
and the dynamic eigenvalues $\lambda$.
\begin{asump}\label{asmp:lipscond}
We assume the spatial basis $\{\phi_\ell(\bx)\}_{\ell=1}^\infty$ are bounded in $L^\infty(\mX)$ and are Lipschitz with controlled growth rate in the Lipschitz constants $\mathrm{Lip}(\phi_\ell)$:
\begin{equation}\label{eq:lipspat}
\sup_{\ell\in\mathbb N} \Vert \phi_\ell\Vert_\infty + \ell^{-1} \mathrm{Lip}(\phi_\ell) \leq C, \quad for \; some \; C>0
\end{equation}
Define $Q_{\lambda, \mC}(t,t'):=\lambda^2(t) \mC(t,t)-2\lambda(t)\mC(t,t')\lambda(t')+\lambda^2(t)\mC(t,t)$.
We need the following additional assumption on $\lambda$ for the regularity of the full function $f(\bx, t)$
\begin{equation}\label{eq:liptemp}
\lambda \in \ell^{2,s}(L^\infty(\mT)), \qquad \sup_{\ell\in \mathbb N} \ell^{-2}\sup_{t,t'\in\mT}\frac{Q_{\lambda_\ell, \mC}(t,t')}{\Vert\lambda\Vert_\infty^2|t-t'|^2} \leq C, \quad for \; some \; C>0
\end{equation}
\end{asump}
The following theorem 
states that the regularity of random function $f$ in \eqref{eq:KLexpan} depends on the decay rate $s$ of dynamic eigenvalues $\lambda$.
\begin{restatable}{thm}{regularity}
[Regularity of Random Functions]
\label{thm:regularity}
Assume $\lambda\in \ell^{2,s}(L^2(\mT))$. 
If $f(\bx, t)\sim \GP(0, \mC_\bz)$ as in Theorem \ref{thm:KLexpan}, then 
$f=\sum_{\ell=1}^\infty f_\ell(t) \phi_\ell(\bx) \in \ell^{2,s}(L^2(\mZ))$ in probability.

Moreover, under Assumption \ref{asmp:lipscond}-\eqref{eq:lipspat}, there is a version \footnote{A version/modification of stochastic process $\tilde f(\bx)$ of $f(\bx)$ means $\bbP[\tilde f(\bx)=f(\bx)]=1$ for $\forall \bx\in \mX$.}  $\tilde f(\bx)$ of $f(\bx):=\int_\mT f(\bx,t)dt$ in $C^{0,s'}(\mX)$ for $s'<s$.
If further $\mC_\bz=\mC_{\bx|t}^\half \mC_{\bx|t'}^\half \qprod \mC_t$ and $\{Q_{\lambda_\ell, \mC_t}\}$ satisfies Assumption \ref{asmp:lipscond}-\eqref{eq:liptemp},
then there is a version $\tilde f(\bz)$ of $f(\bz)$ in $C^{0,s'}(\mZ)$ for $s'<s$.
\end{restatable}
\begin{proof}
See Appendix \ref{apx:regularity}.
\end{proof}
\begin{rk}
For random $\lambda$ satisfying Assumption \ref{asmp:randeig}-\eqref{eq:randcoeff},
the above results still hold for $s<(\kappa-1)/2$ by Proposition \ref{prop:regdyneig}.
\end{rk}

\begin{restatable}{cor}{mgGP}
\label{cor:mgGP}
If $f(\bx, t)\sim \GP(0, \mC_\bz)$ has a continuous version, then $\{f_\ell\}_{\ell=1}^\infty$ as in Theorem \ref{thm:KLexpan} are GP's defined on $\mT$.
\end{restatable}
\begin{proof}
See Appendix \ref{apx:mgGP}.
\end{proof}

\subsection{Posterior Contraction}\label{sec:post_contr}
Now we consider the posterior properties.
On the separable Banach space $(\mbB=\ell^2(L^2(\mT)), \Vert\cdot\Vert_{2,2})$,
we consider a Gaussian random element $\lambda$ satisfying Assumption \ref{asmp:randeig}-\eqref{eq:randcoeff} and denote its associated reproducing kernel Hilbert space (RKHS) as $(\mbH,\Vert\cdot\Vert_\mbH)$.
We assume $\tr(\mC_u)=1$ by rescaling in \eqref{eq:randcoeff}. Then RKHS is $\mbH=\ell^{2,\kappa/2}(L^2(\mT))$ with the following inner product and norm
\begin{equation}
\langle h, h' \rangle_\mbH = \sum_{\ell=1}^\infty \langle \gamma_\ell^{-1} h_\ell, \gamma_\ell^{-1} h'_\ell \rangle, \, \forall h, h'\in \mbH, \qquad \Vert \cdot \Vert_\mbH = \langle \cdot, \cdot\rangle_\mbH^\half
\end{equation}
Define the contraction rate of $\lambda$ at $\lambda_0$ as follows
\begin{equation}\label{eq:contrate1}
\varphi_{\lambda_0}(\eps) = \inf_{h\in\mbH:\Vert h-\lambda_0\Vert_{2,2}\leq \eps} \half \Vert h\Vert_\mbH^2 - \log \Pi(\Vert \lambda\Vert_{2,2} <\eps)
\end{equation}
where $\Pi$ is the prior measure on $\lambda$.
Let $p$ be a centered (assume $m\equiv0$ for simplicity) Gaussian model, which is uniquely determined by its covariance $\mC_{\bx|t}=\sum_{\ell=1}^\infty \lambda_\ell^2(t)\phi_\ell \otimes \phi_\ell$.
For a fixed spatial basis $\{\phi_\ell\}$, the model density $p$ is parametrized by $\lambda$, hence denoted as $p_\lambda$.
Let $n=I\wedge J$.
Denote $P_\lambda^{(n)}:=\bigotimes_{j=1}^n P_{\lambda,j}$ as the product measure on $\bigotimes_{j=1}^n(\mX_j,\mathcal B_j,\mu_j)$.
Each $P_{\lambda,j}$ has a density $p_{\lambda_j}$ with respect to the $\sigma$-finite measure $\mu_j$.
Define the average Hellinger distance as $d_{n,H}^2(\lambda,\lambda')=\frac1n\sum_{j=1}^n\int (\sqrt{p_{\lambda,j}}-\sqrt{p_{\lambda',j}})^2 d\mu_j$.
To bound the Hellinger distance between the modeling parameter $\lambda$ and its true value $\lambda_0$, we make the following assumption.
\begin{asump}\label{asmp:eigbound}
Let $\lambda\in \ell^{1,s}(L^\infty(\mT))$ 
with some $s>0$. Assume $\lambda$ satisfy the following bounds
\begin{equation}
c_\ell:= \inf_{t\in\mT} |\lambda_\ell(t)| \gtrsim \ell^{-s/2}, \quad C:= \sup_{\ell\in\mathbb N} \Vert\lambda_\ell\Vert_\infty <+\infty
\end{equation}
\end{asump}
Denote the observations $Y^{(n)}=\{Y_j\}_{j=1}^n$ with $Y_j=y(\bX,t_j)$. Note they are i.i.d. in model I and i.n.i.d. in model II conditioned on the mean.
Now we prove the following posterior contraction about $\mC_{\bx|t}$ in model II, which generalizes Theorem 2.2 of \cite{lan_2019}.
\begin{restatable}{thm}{postcontrCII}
[Posterior Contraction of $\mC_{\bx|t}$ in model II]
\label{thm:postcontrCII}
Let $\lambda$ be a Borel measurable, zero-mean, tight Gaussian random element in $\Theta=\ell^2(L^2(\mT))$ satisfying Assumption \ref{asmp:eigbound} and $P_\lambda^{(n)}=\bigotimes_{j=1}^n P_{\lambda,j}$ be the product measure of $Y^{(n)}$ parametrized by $\lambda$.
If the true value $\lambda_0\in \Theta$ is in the support of $\lambda$, and $\eps_n$ satisfies the rate equation $\varphi_{\lambda_0}(\eps_n)\leq n \eps_n^2$ with $\eps_n\geq n^{-\half}$,
then there exists $\Theta_n\subset\Theta$ such that $\Pi_n(\lambda\in \Theta_n:d_{n,H}(\lambda,\lambda_{n,0})>M_n\eps_n|Y^{(n)})\rightarrow 0$ in $P_{\lambda_{n,0}}^{(n)}$-probability for every $M_n\rightarrow \infty$.
\end{restatable}
\begin{proof}
See Appendix \ref{apx:postcontrCII}.
\end{proof}

\begin{rk}
Note the above posterior contraction theorem is for the covariance operator $\mC_{\bx|t}$ (characterized through the dynamic eigen-values $\lambda$), not for the mean function (which has already been extensively studied in \cite{Ghosal_2017}). Therefore, the non-parametric model II of \eqref{eq:gstgpm} for $\mC_{\bx|t}$ is not in the conjugate setting: because of the position of $\lambda$ inside $\mC_{\bx|t}$, the likelihood for $\lambda$ is not Gaussian.
\end{rk}

\begin{rk}
The observations $Y^{(n)}$ in model I are conditional iid and the likelihood model is determined by $\mC_{y|m}=\sigma^2_\eps \mI_\bx \otimes \mI_t$ which does not contain $\lambda$,
thus the posterior of $\lambda$ in model I cannot contract to the correct value.
\end{rk}

Although the above theorem dictates that the posterior of $\lambda$ contracts to the true value $\lambda_0$ at certain rate $\eps_n$, it does not provide the details of $\eps_n$.
The following theorem specifies the contraction rate, which depends on the regularity of both the truth and the prior used.
\begin{restatable}{thm}{contrateCII}
[Posterior Contraction Rate of $\mC_{\bx|t}$ in model II]
\label{thm:contrateCII}
Let $\lambda$ be a Gaussian random element satisfying Assumption \ref{asmp:randeig}-\eqref{eq:randcoeff} and $\tr(\mC_u)=1$. 
The rest settings are the same as in Theorem \ref{thm:postcontrCII}. 
If the true value $\lambda_0\in \ell^{2,s}(L^2(\mT))$, then we have the rate of posterior contraction $\eps_n = \Theta(n^{-(\frac{\kappa-1}{2}\wedge s)/\kappa})$.
\end{restatable}
\begin{proof}
See Appendix \ref{apx:contrateCII}.
\end{proof}

\section{Posterior Inference and Predictions}\label{sec:infpred}
Suppose we are given the spatiotemporal data $\mD:=\{\bZ,\bY\}$, and there are $K$ independent trials in the data $\bY:=\{\bY_k\}_{k=1}^K$
with each $\bY_k=y(\bZ)=[y(\bx_i, t_j)]_{I\times J}$.
Denote $\bM_{I\times J}:=m(\bZ)$, $\bC_\bM := \mC_m(\bZ,\bZ)$ and $\bC_{\bY|\bM} := \mC_{y|m}(\bZ, \bZ)$.
We truncate the Mercer's representation of $\mC_{\bx|t}$ for $L$ terms.
Consider the model \eqref{eq:gstgpm} with hyper-parameters which are in turn given priors respectively, summarized as follows
\begin{equation}\label{eq:gstgp_full}
\begin{aligned}
\VEC(\bY_k)| \bM, \bC_{\bY|\bM} &\sim \mN( \VEC(\bM), \bC_{\bY|\bM}), \quad \text{i.i.d.\, for}\; k=1, \cdots, K \\
m(\bz) &\sim \GP(0, \mC_m), \quad \mC_{\bx|t}(\bz,\bz') = \sum_{\ell=1}^L \lambda_\ell(t) \lambda_\ell(t') \phi_\ell(\bx) \phi_\ell(\bx') \\
\lambda_\ell(t) &= \gamma_\ell u_\ell(t), \quad u_\ell(\cdot) \overset{iid}{\sim} \GP(0, \mC_u) \; \textrm{for} \; \ell=1, \cdots, L \\
\mC_* &= \sigma^2_* \exp(-0.5\Vert *-*'\Vert^s/\rho_*^s) \\
\sigma^2_* &\sim \Gamma^{-1}(a_*,b_*), \quad \log\rho_* \sim \mN(m_*,V_*), \quad * = \bx, t, \,\textrm{or}\, u
\end{aligned}
\end{equation}
where the likelihood kernel $\mC_{y|m}$ and the prior kernel $\mC_m$ are specified in \eqref{eq:gstgpm}. 
We adopt the Metropolis-Within-Gibbs scheme and use the slice samplers \citep{neal03,murray10} for the posterior inference. More details can be found in Appendix \ref{apx:postinf}.

With the posteriors we can consider the following various prediction problems at new data points $(\bx_*, t_*)$, $(\bx, t_*)$ or $(\bx_*, t)$:
\begin{equation}
m(\bx_*, t_*)|\mD, \quad m(\bx_*, t)|\mD, \quad m(\bx, t_*)|\mD, \quad \mC_{\bx|t_*}(\bx, \bx')|\mD, \quad \mC_{\bx|t}(\bx, \bx_*)|\mD
\end{equation}

\subsection{Prediction of Mean}
We only consider the prediction $m(\bx_*, t_*)|\mD$ because the other two predictions $m(\bx_*, t)|\mD$, $m(\bx, t_*)|\mD$ are sub-problems of it.
The prediction of mean has been well studied in the literature.
The following proposition gives the predictive distribution of the mean function $m(\bx, t)$ in our set-up.
\begin{restatable}{prop}{predmean}
\label{prop:predmean}
Fit the spatiotemporal data $\mD=\{\bZ,\bY\}$ with the model \eqref{eq:gstgp_full}.
Then given a new point $\bz_*=(\bx_*, t_*)$ we have
\begin{align*}
m(\bz_*)|\mD &\sim \mN(m', C') \\ 
m' &= \tp c_* (\bC_\bM + K^{-1} \bC_{\bY|\bM})^{-1} \bar{\bY},  \quad
C' = C_{m_*} - \tp c_* (\bC_\bM + K^{-1} \bC_{\bY|\bM})^{-1} c_*
\end{align*}
where we denote
\begin{equation*}
\bar{\bY}:= \frac1K \sum_{k=1}^K \VEC(\bY_k), \quad C_{m_*} := \mC_m(\bz_*,\bz_*), \quad c_* := \mC_m(\bZ,\bz_*), \quad \tp c_* := \mC_m(\bz_*,\bZ)
\end{equation*}
\end{restatable}
\begin{proof}
See Appendix \ref{apx:predmean}.
\end{proof}

\subsection{Prediction of Covariances}
Now we consider two types of prediction for covariances of particular interest.
The first one, $\mC_{\bx|t_*}(\bx, \bx')|\mD$, evolves TESD among existing locations to other (future) time(s) $t_*$.
The second one, $\mC_{\bx|t}(\bx, \bx_*)|\mD$, extends TESD to new (neighboring) location(s) $\bx_*$.
Both predictions have practical meaning and useful applications. For example, the former could predict how the brain connection evolves during some memory process, or in the progression of brain degradation of Alzheimer's disease. With the latter we could extend our knowledge of climate change from observed regions to unobserved territories.

Note, the prediction of TESD to new locations is exclusive to the proposed fully nonparametric model. The semi-parametric methods of dynamic covariance modeling \citep{wilson11,fox15,lan_2019} with a covariance matrix (instead of a kernel) for the spatial dependence do not have this feature because the discrete spatial size (the size of covariance matrix) has been fixed.

\subsubsection{Evolve Spatial Dependence to Future Time}\label{sec:TESD2future}
Note from the definition \eqref{eq:spatkern_t1}, we know that $\mC_{\bx|t}$ is a function of dynamic eigenvalues $\{\lambda_\ell(t)\}$ with fixed spatial basis $\{\phi_\ell(\bx)\}$.
Therefore, the prediction of the kernel $\mC_{\bx|t}$ in the time direction can be reduced to predicting $\lambda_\ell(t_*)|\mD$ as follows. Denote $\vect\lambda_\ell:=\lambda_\ell(\bt)$.
\begin{equation*}
p(\lambda_\ell(t_*)| \mD) = \int p(\lambda_\ell(t_*), \vect\lambda_\ell| \mD) d \vect\lambda_\ell = \int p(\lambda_\ell(t_*)| \vect\lambda_\ell) p(\vect\lambda_\ell | \mD) d \vect\lambda_\ell \\
\end{equation*}
where $p(\lambda_\ell(t_*)| \vect\lambda_\ell)$ is the standard GP predictive distribution.
We can use the standard GP predictive mean and covariance to predict and quantify the associated uncertainty for $(\lambda_\ell(t_*)| \vect\lambda_\ell^{(s)})^2$ with $\vect\lambda_\ell^{(s)} \sim p(\vect\lambda_\ell | \mD)$,
and then take average over all the posterior samples to get an approximation of $\lambda_\ell^2(t_*)| \mD$.
Therefore, $\mC_{\bx|t_*}(\bx, \bx')|\mD$ can be obtained/approximated by substituting $\lambda_\ell^2(t)$ with $\lambda_\ell^2(t_*)| \mD$ in \eqref{eq:spatkern_t1}
\begin{equation*}
\begin{aligned}
\mC_{\bx|t_*}(\bx, \bx')|\mD &= \sum_{\ell=1}^\infty (\lambda_\ell^2(t_*)|\mD)\, \phi_\ell(\bx) \phi_\ell(\bx') \\
&\approx \frac1S \sum_{s=1}^S \sum_{\ell=1}^L (\lambda_\ell(t_*)| \vect\lambda_\ell^{(s)})^2 \phi_\ell(\bx;\eta_\bx^{(s)}) \phi_\ell(\bx';\eta_\bx^{(s)}), \quad \vect\lambda_\ell^{(s)} \sim p(\vect\lambda_\ell | \mD)
\end{aligned}
\end{equation*}
where $\lambda_\ell(t_*)| \vect\lambda_\ell^{(s)}=\mC_u(t^*,\bt) \mC_u(\bt,\bt)^{-1} \vect\lambda_\ell^{(s)}$.

\subsubsection{Extend Evolution of Spatial Dependence to Neighbors}\label{sec:TESD2neighbor}
Recall that in the definition of $\mC_{\bx|t}$, the fixed basis $\{\phi_\ell(\bx)\}$ is taken from the eigenfunctions of the spatial kernel $\mC_\bx$.
To extend $\mC_{\bx|t}$ as a function of time to other locations based on existing knowledge informed by data, 
one could predict the basis at a new position, namely $\phi_\ell(\bx_*)$, using its known values $\vect\phi_\ell:=\phi_\ell(\bX)$ as in the conditional Gaussian:
\begin{equation*}
\phi_\ell(\bx_*)| \vect\phi_\ell = \mC_\bx(\bx_*,\bX) \mC_\bx(\bX,\bX)^{-1} \vect\phi_\ell = \mC_\bx(\bx_*,\bX) \lambda_\ell^{-2} \vect\phi_\ell, \quad \forall \ell=1, \cdots, L
\end{equation*}
Then, $\mC_{\bx|t}(\bx, \bx_*)|\mD$ can be predicted/approximated by substituting $\phi_\ell(\bx')$ with $\phi_\ell(\bx_*)| \vect\phi_\ell$ in \eqref{eq:spatkern_t1} as follows
\begin{equation*}
\begin{aligned}
\mC_{\bx|t}(\bx,\bx_*)| \mD &= \sum_{\ell=1}^\infty \lambda_\ell^2(t)\phi_\ell(\bx) (\phi_\ell(\bx_*)| \mD) \\
&\approx \frac1S \sum_{s=1}^S \sum_{\ell=1}^L (\lambda_\ell^{(s)}(t))^2 \phi_\ell(\bx;\eta_\bx^{(s)}) (\phi_\ell(\bx_*)| \vect\phi_\ell), \quad \lambda_\ell^{(s)}(t) \sim p(\lambda_\ell(t) | \mD)
\end{aligned}
\end{equation*}

\section{Numerical Experiments}\label{sec:numerics}
In this section, we compare the proposed generalized STGP models \eqref{eq:gstgpm} against a complete spectrum of models with different combinations in stationarity and separability including:
\begin{enumerate}\setcounter{enumi}{-1}
    \item classical stationary separable (stat-sep) model \eqref{eq:sepkern}; 
    \item stationary non-separable model \citep{Gneiting_2002}:
\begin{equation}\label{eq:stat-nonsep}
\textrm{stat-nonsep}: \qquad C(\bz,\bz') = \frac{\sigma^2}{(a|t-t'|^{2\alpha}+1)^{\beta d/2}} \exp\left( -\frac{c\Vert \bx-\bx'\Vert^{2\gamma}}{(a|t-t'|^{2\alpha}+1)^{\beta \gamma}}\right)
\end{equation}
where for simplicity we fix $a=d=1$, $\alpha=\gamma=\half$, $\beta=2$, and let $c=\frac{1}{2\rho}$.
    \item non-stationary separable model \citep{Paciorek_2003}:
\begin{equation}\label{eq:nonstat-sep}
\textrm{nonstat-sep}: \qquad C(\bz,\bz') = \sigma^2 A(\bx, \bx') \exp\left(-Q(\bx, \bx')\right) * A(t, t) \exp\left(-Q(t, t')\right)
\end{equation}
where $A(r,r')=|\Sigma_r(\rho)|^\frac14 |\Sigma_{r'}(\rho)|^\frac14 |(\Sigma_r(\rho)+\Sigma_{r'}(\rho))/2|^{-\half}$ and $Q(r,r')=\tp{(r-r')} [(\Sigma_r(\rho)+\Sigma_{r'}(\rho))/2]^{-1} (r-r')$ with $r$ being $\bx$ or $t$. $\Sigma_r(\rho)$ is the covariance of Gaussian kernel centered at $r$ with correlation length $\rho$.
    \item non-stationary non-separable model \citep{Wang_2020}:
\begin{equation}\label{eq:nonstat-nonsep}
\begin{aligned}
\textrm{nonstat-nonsep}: \qquad C(\bz,\bz') =& \sigma^2 A(\bx, \bx') A(t, t) \left(1+Q(\bx, \bx')+Q(t, t')\right)^{s_0} \cdot\\ &\left(1+Q(\bx, \bx')\right)^{s_1} \left(1+Q(t, t')\right)^{s_2}
\end{aligned}
\end{equation}
\end{enumerate}
Table \ref{tab:models2compare} lists all the models for comparison with their authors and properties.

\begin{table}[ht]\tiny
\begin{center}
\begin{tabular}{l|ccccc}
  \toprule
 Model & Author(s) & Non-stationary & Non-separable & Non-parametric & Sparse \\ 
  \midrule
stat-sep & -- & \xmark & \xmark & \cmark (mean) \xmark (covariance) & \xmark \\
stat-nonsep & \cite{Gneiting_2002} & \xmark & \cmark & \cmark (mean) \xmark (covariance) & \xmark \\
nonstat-sep & \cite{Paciorek_2003} & \cmark & \xmark & \cmark (mean) \xmark (covariance) & \xmark \\
nonstat-nonsep & \cite{Wang_2020} & \cmark & \cmark & \cmark (mean) \xmark (covariance) & \xmark \\
\midrule
qKron-prod (I) & S.Lan & \cmark & \cmark & \cmark (mean) \cmark (covariance) & \xmark \\ 
qKron-sum (II) & S.Lan & \cmark & \cmark & \cmark (mean) \cmark (covariance) & \cmark \\ 
  \bottomrule
\end{tabular}
\caption{Spatiotemporal models for comparison.} 
\label{tab:models2compare}
\end{center}
\end{table}

We evaluate the performance of all the above models in fitting and predicting mean and covariance as functions of time using a simulated spatiotemporal process and an analysis of real brain imaging data.
As a Bayesian non-parametric method for covariance modeling, model II with the quasi Kronecker sum structure (qKron-sum) is shown to be the best in characterizing TESD.
All the numerical codes are publicly available at 
\href{https://github.com/lanzithinking/TESD_gSTGP}{\underline{GitHub repository}}.

\subsection{Simulation} \label{sec:simulation}
In this section, we study a simulated example of non-stationary and non-separable spatiotemporal process.
One can find a similar study for a simulated stationary and non-separable process in Appendix \ref{apx:stat-nonsep}.

\subsubsection{Data Generation}
\begin{figure}[t] 
   \centering
   \includegraphics[width=1\textwidth,height=.35\textwidth]{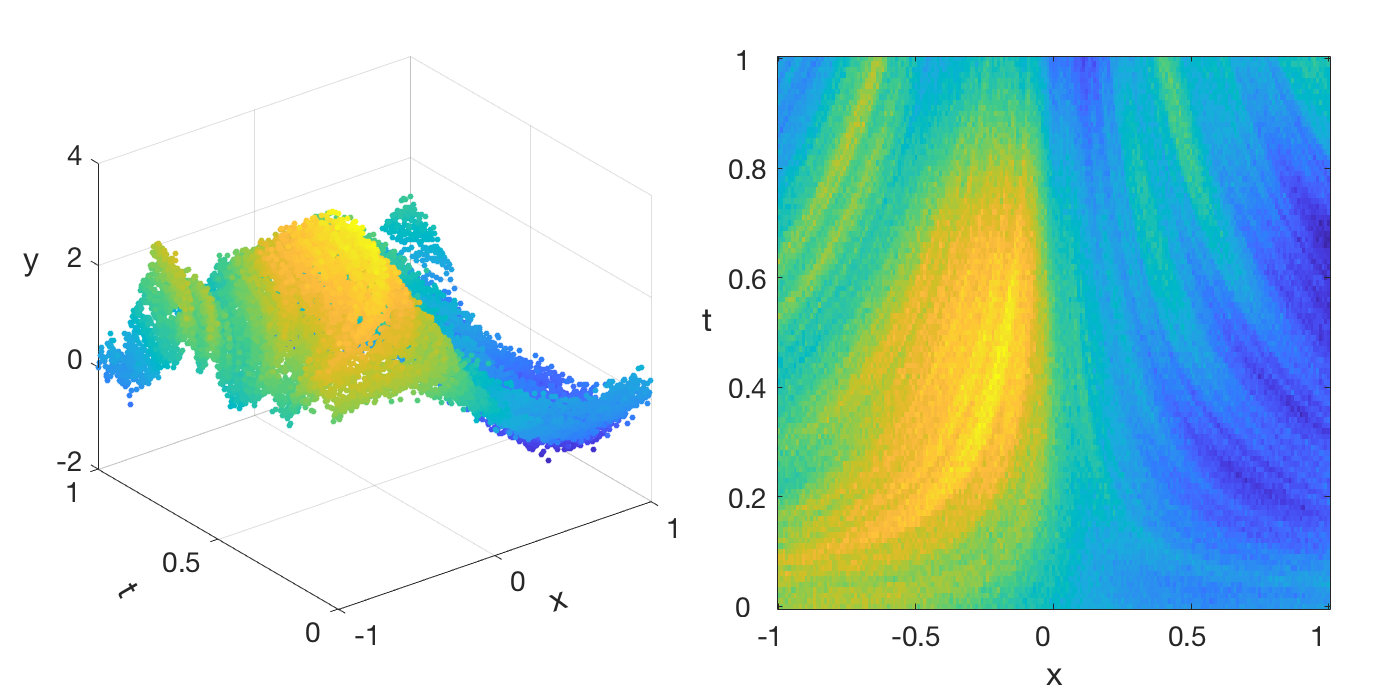} 
   \caption{Simulated spatiotemporal data over $[-1,1]\times[0,1]$, viewed in 3d (left) and projected in space-time domains (right).}
   \label{fig:sim_STproc}
\end{figure}

We consider the following spatiotemporal process with a non-stationary, non-separable covariance for spatial dimension $d=1$:
\begin{equation}\label{eq:simSTP}
\begin{aligned}
y(x,t) &\sim \GP(m, \mC_y), \quad x\in \mX=[-1,1], \; t\in \mT=[0,1] \\
m(x,t) &= \cos(\pi x) \sin(2\pi t) \\
\mC_y^\text{ns}(\bz,\bz') &= \exp\left( -\frac{|x-x'|^2}{2\ell_x} - \frac{|t-t'|^2}{2\ell_t} - \frac{|xt-x't'|}{2\ell_{xt}}\right) + \sigma^2_\eps\delta(\bz=\bz')
\end{aligned}
\end{equation}

To generate observations, we discretize the domain by dividing $\mX$ into $N_x=200$ equal subintervals and $\mT$ into $N_t=100$ equal subintervals.
Setting $\ell_x=0.5$, $\ell_t=0.3$, $\ell_{xt}=\sqrt{\ell_x\ell_t}\approx0.39$ and $\sigma^2_\eps=10^{-2}$, we generate $20301$ data points $\{y_{ij}\}$ over the mesh grid.
Such random process can be repeated for $K$ trials and we plot one of them 
in Figure \ref{fig:sim_STproc}.

\subsubsection{Model Fit}\label{sec:sim_fit}
For simplicity we use a subset of these $20301$ data points taken on an equally spaced sub-mesh with $I=5$ and $J=101$. 
Now we fit the data with STGP models \eqref{eq:gstgpm} respectively. 
We set $\gamma_\ell=\ell^{-\kappa/2}$ with $\kappa=1.2$, $a=[1,1,1]$, $m=[0,0,0]$ for all models; $b=[5,10,10]$, $V=[0.1,0.1,0.01]$ for models 0 and I and $b=[0.1,1,5]$, $V=[1,1,1]$ for model II. 
The truncation number of Mercer's kernel expansion is set to $L=I=5$.
For each experiment, we run MCMC to collect $2.4\times 10^4$ samples, burn in the first 4000, and subsample every other.
The resulting $10^4$ posterior samples are used to estimate the mean function $m|\bY$ and the covariance function $\mC_{y|t}|\bY$. 
For other parametric covariance models (the first four) listed in Table \ref{tab:models2compare} we
let $\sigma^2\sim \Gamma(a', b')$ and $\log \rho\sim \mN(m',V')$.
Though not originally designed to learn TESD, these models are set in the framework of \eqref{eq:gstgp_full} with the same hyper-parameter setting as model I.
We want to test these models in recovering mean function $m(\bx, t)$ and more importantly learning TESD $\mC_{y|t}$ which has the following truth:
\begin{equation}
\begin{aligned}
C_{y|t}^\text{ns}(x,x'):=\Cov[y(x,t), y(x',t)] &= \exp\left( -\frac{|x-x'|^2}{2\ell_x} - \frac{|x-x'|t}{2\ell_{xt}} \right) + \sigma^2_\eps\delta(x=x') \\
\end{aligned}
\end{equation}

\begin{figure}[t] 
   \begin{subfigure}[b]{.495\textwidth}
    \includegraphics[width=1\textwidth,height=.6\textwidth]{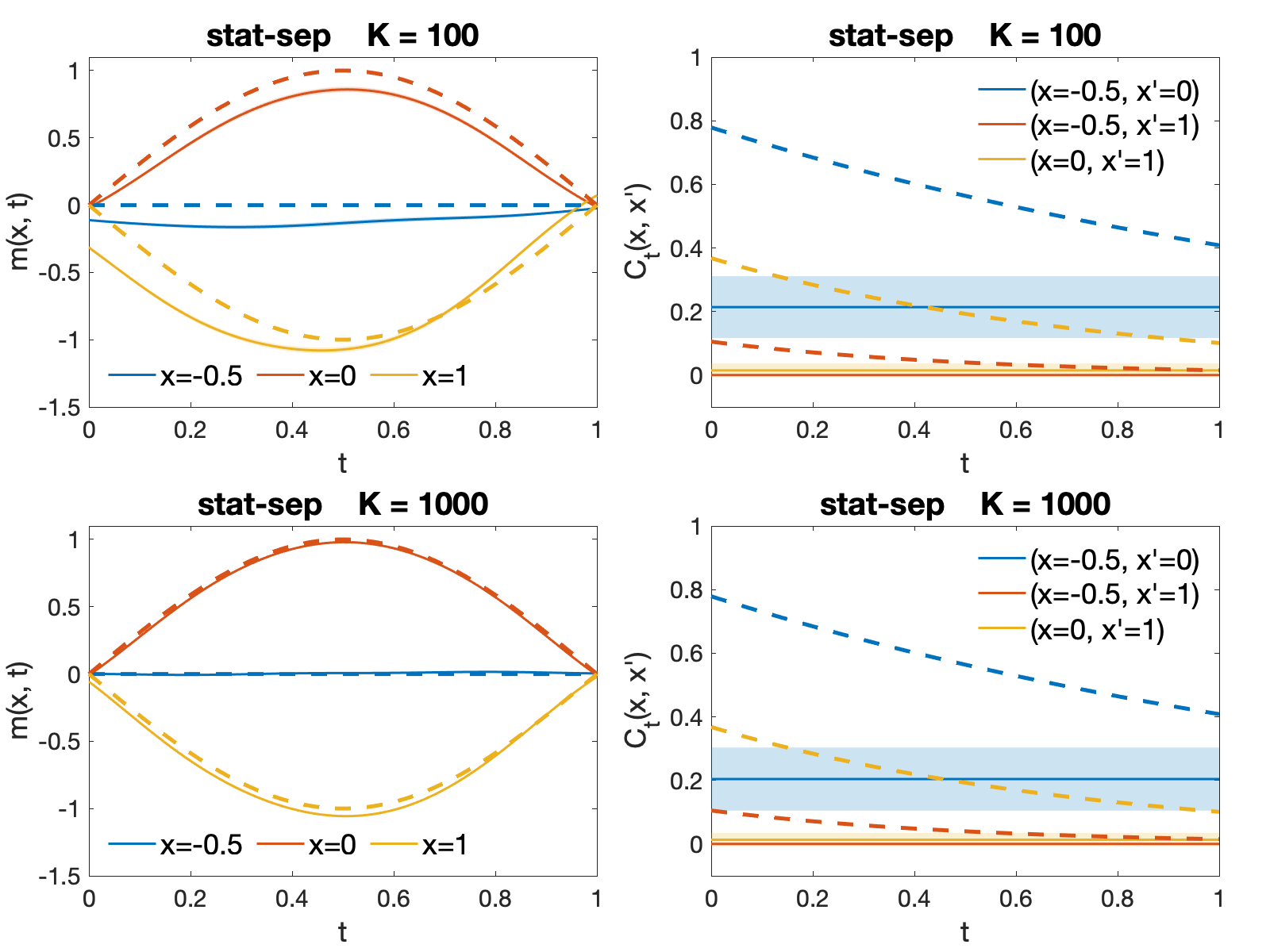} 
    \caption{Stat-sep model}
    \end{subfigure}
    \begin{subfigure}[b]{.495\textwidth}
   \includegraphics[width=1\textwidth,height=.6\textwidth]{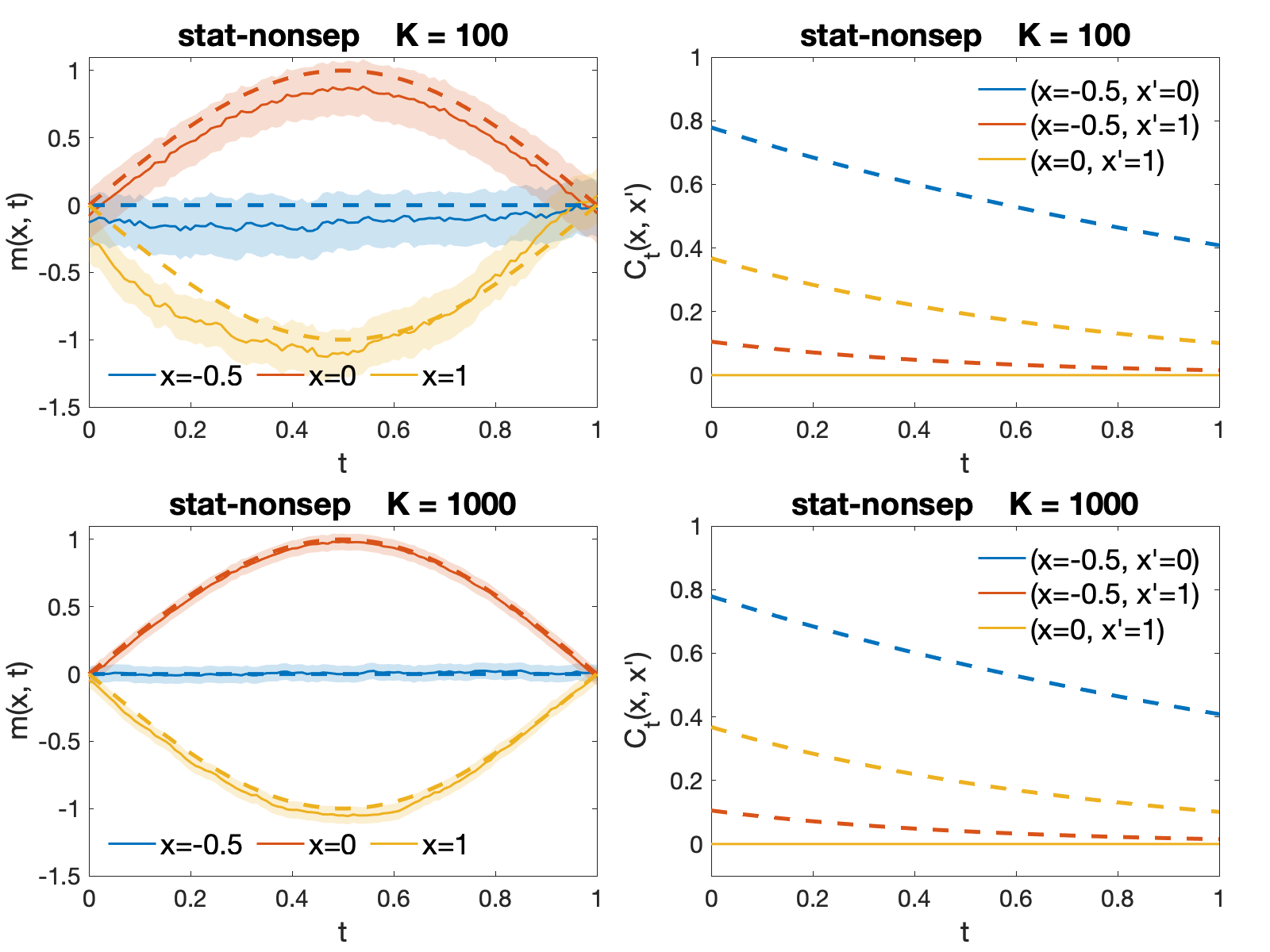}
   \caption{Stat-nonsep model}
   \end{subfigure}
   \begin{subfigure}[b]{.495\textwidth}
   \includegraphics[width=1\textwidth,height=.6\textwidth]{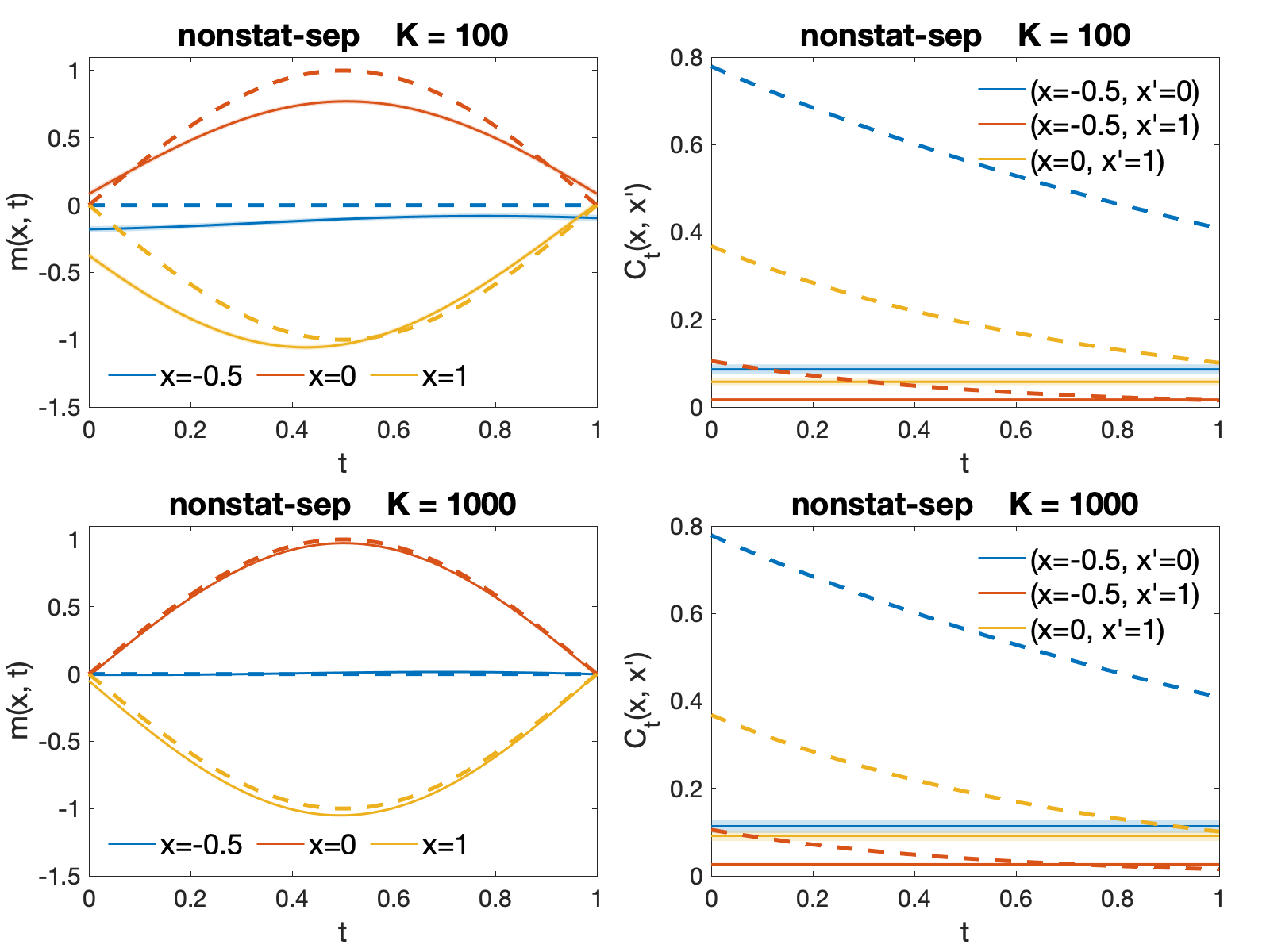}
   \caption{Nonstat-sep model}
   \end{subfigure}
   \begin{subfigure}[b]{.495\textwidth}
   \includegraphics[width=1\textwidth,height=.6\textwidth]{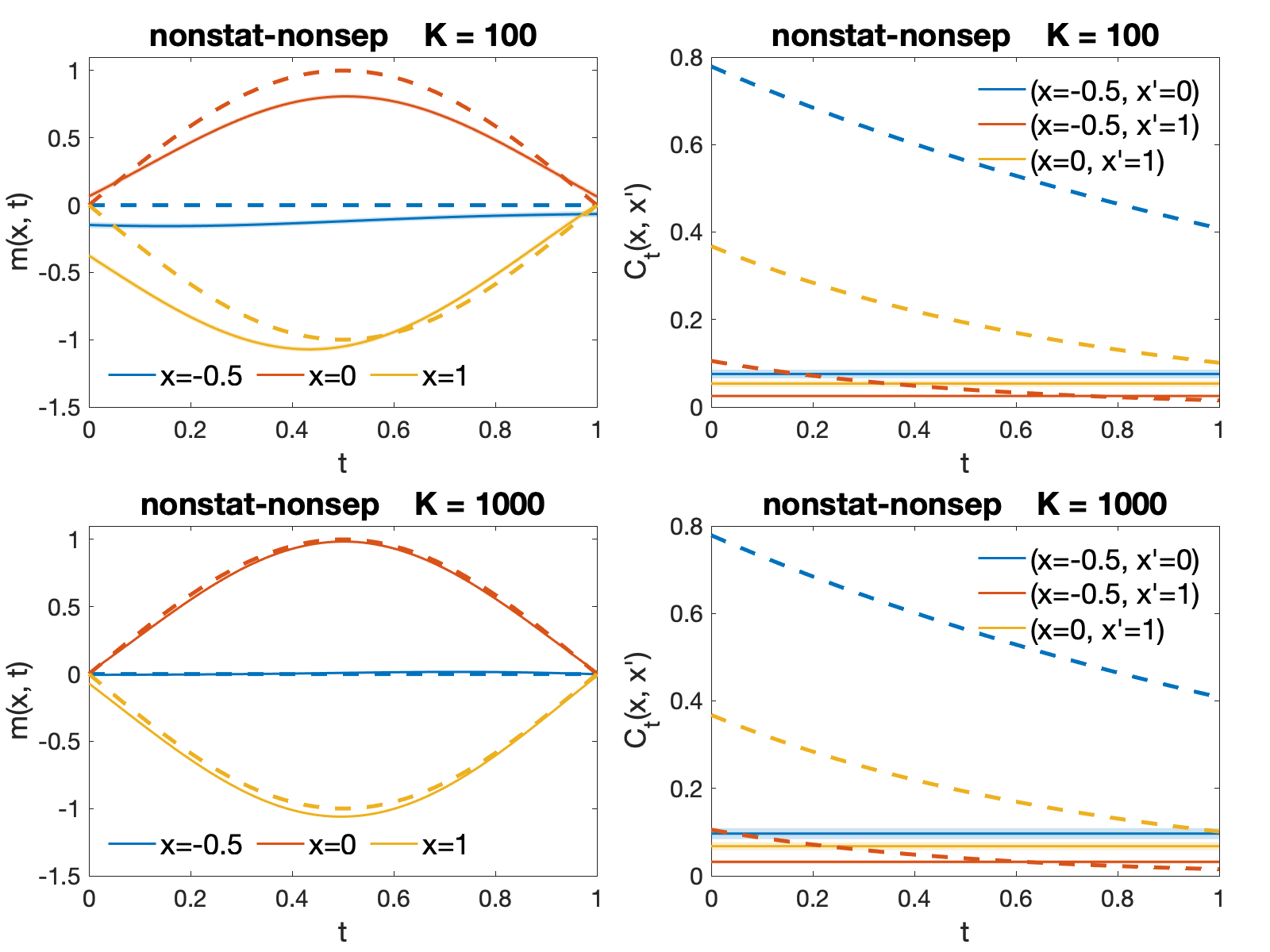}
   \caption{Nonstat-nonsep model}
   \end{subfigure}
    \begin{subfigure}[b]{.495\textwidth}
   \includegraphics[width=1\textwidth,height=.6\textwidth]{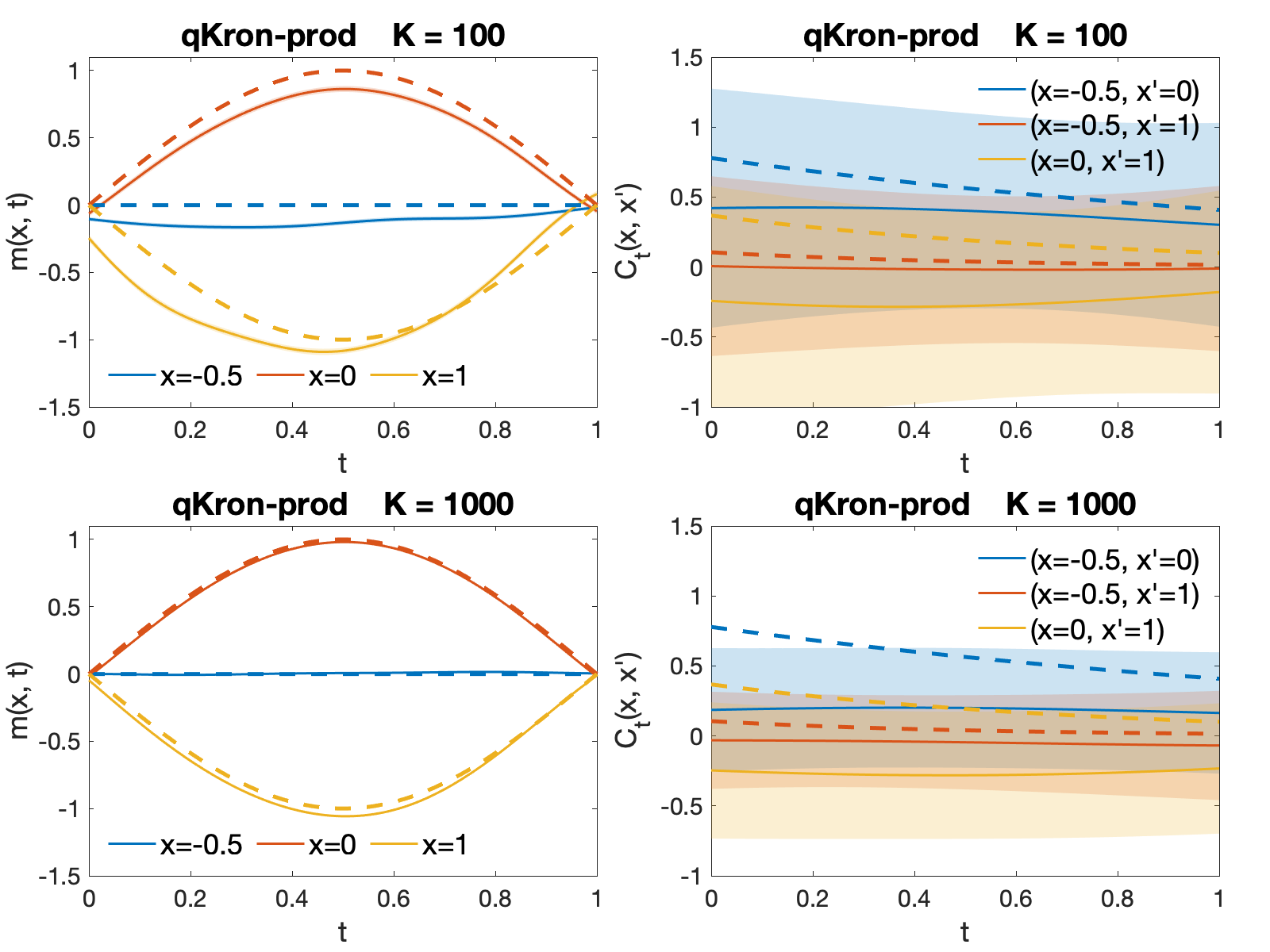} 
   \caption{qKron-prod model (I)}
   \end{subfigure}
   \begin{subfigure}[b]{.495\textwidth}
   \includegraphics[width=1\textwidth,height=.6\textwidth]{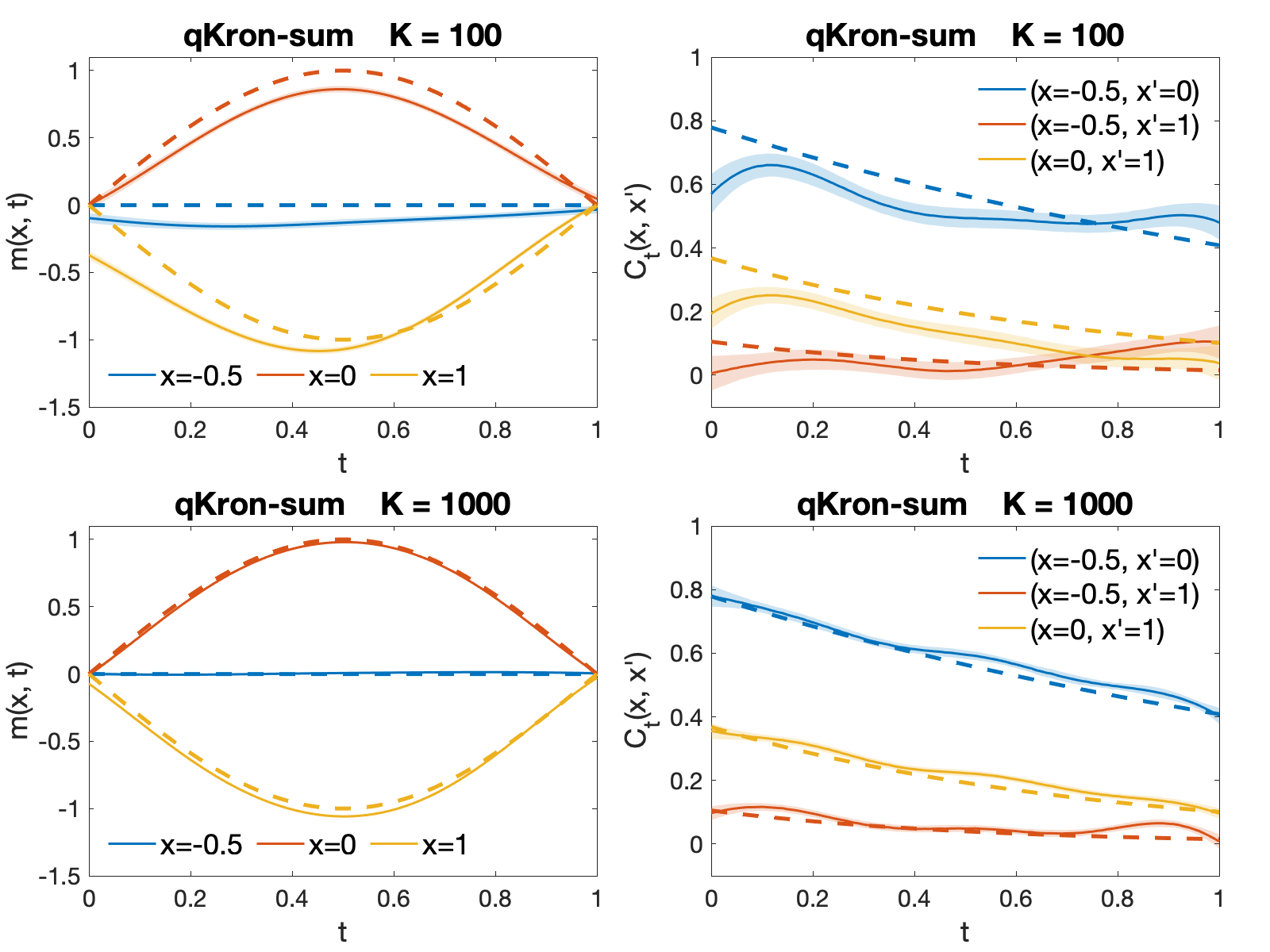}
   \caption{qKron-sum model (II)} \label{fig:fitted_STproc_nonstat_qKron_sum}
   \end{subfigure}
   \caption{A simulated process: selective mean functions $m(x,t)$ (left column) and covariance functions $C_{y|t}^\text{ns}(x,x')$ (right column) fitted by various models with $K=100$ trials of data (upper row) and $K=1000$ trials of data (lower row) on each panel. Dashed lines are true values, solid curves are estimates with shaded credible regions indicating their uncertainty.}
   \label{fig:fitted_STproc_nonstat}
\end{figure}

The posterior estimates by different models are plotted at selective locations in Figure \ref{fig:fitted_STproc_nonstat}.
With growing data information (increasing trial number $K$) \footnote{More repeated trials ($K$) can be viewed as increasing data ($n$) though they are stacked over the same discrete time points $\{t_j\}_{j=1}^J$. Better contraction (closer to truth with better credible band coverage) can be achieved by increasing $J$ distinct time points \citep[as increasing $N$ in Figure 5 of][]{lan_2019}.}, most posterior estimates contract towards some values (with decreasing width of credible bands), 
especially for mean functions (left column of each panel).
However, most of them fail to reconstruct correct TESD (right column of each panel). According to Proposition \ref{prop:essence}, model 0 yields static TESD, and the result by model I is inevitably pulled towards the wrong constant $\sigma^2_\eps$ with increasing data due to the poorly structured likelihood kernel $\mC_{y|m}^\textrm{I}$.
Model II, by contrast, characterizes TESD with reasonable accuracy (Panel \ref{fig:fitted_STproc_nonstat_qKron_sum}).
Although all the parametric covariance models recover the true mean functions, none of them generates correct time-varying TESD, even it is the non-stationary non-separable \citep{Wang_2020}.

We repeat the experiments for 10 times and summarize their model performance in terms of mean squared errors (MSE) in Table \ref{tab:mse}.
As we can see, model II renders comparable MSE in estimating mean functions as other models, but $2\sim 3$ orders of magnitude smaller MSE for estimating TESD, using much less time. 

\begin{table}[H]\tiny
\begin{center}
\begin{tabular}{l|ccc|ccc}
  \toprule
  & \multicolumn{3}{c}{$K=100$} & \multicolumn{3}{c}{$K=1000$} \\
  \cmidrule{2-4} \cmidrule{5-7}
 Model & mean & TESD & time & mean & TESD & time \\ 
  \midrule
stat-sep & 2.17e-2 (1.2e-5) & 0.143 (2e-5) & 2410 (230) & 1.14e-3 (9.3e-7) & 0.143 (2e-5) & 2450 (520) \\
stat-nonsep & 2.18e-2 (1.1e-5) & 0.142 (1.9e-05) & 2340 (350) & 1.15e-3 (6.5e-7) & 0.142 (1.3e-5) & 2330 (310) \\ 
  nonstat-sep & 2.34e-2 (1.1e-5) & 0.111 (2.5e-05) & 1770 (200) & 1.1e-3 (8.6e-7) & 9.52e-2 (1.9e-5) & 1840 (250) \\ 
  nonstat-nonsep & 2.27e-2 (1.4e-5) & 0.114 (2.1e-05) & 3230 (410) & 1.13e-3 (7.9e-7) & 0.105 (2.4e-5) & 3240 (420) \\
  \midrule
  qKron-prod (I) & 2.19e-2 (8.7e-4) & 0.15 (0.11) & 3850 (1700) & 1.14e-3 (6.1e-7) & 0.189 (2.7e-2) & 3760 (1400) \\ 
  qKron-sum (II) & 2.06e-2 (2.6e-5) & {\bf 5.14e-3} (9.5e-5) & {\bf 1720} (160) & 1.13e-3 (1.5e-6) & {\bf 3.85e-4} (4.1e-5) & {\bf 2380} (890) \\ 
   \bottomrule
\end{tabular}
\caption{Mean squared errors (MSE) for fitting mean $m(x,t)$ and TESD $C_{y|t}^\text{ns}(x,x')$.
Table values are median estimates of 10 repeated results with sample standard deviation in the parentheses.} 
\label{tab:mse}
\end{center}
\end{table}

\subsubsection{Model Prediction}
\begin{figure}[t] 
   \begin{subfigure}[b]{.495\textwidth}
   \includegraphics[width=1\textwidth,height=.7\textwidth]{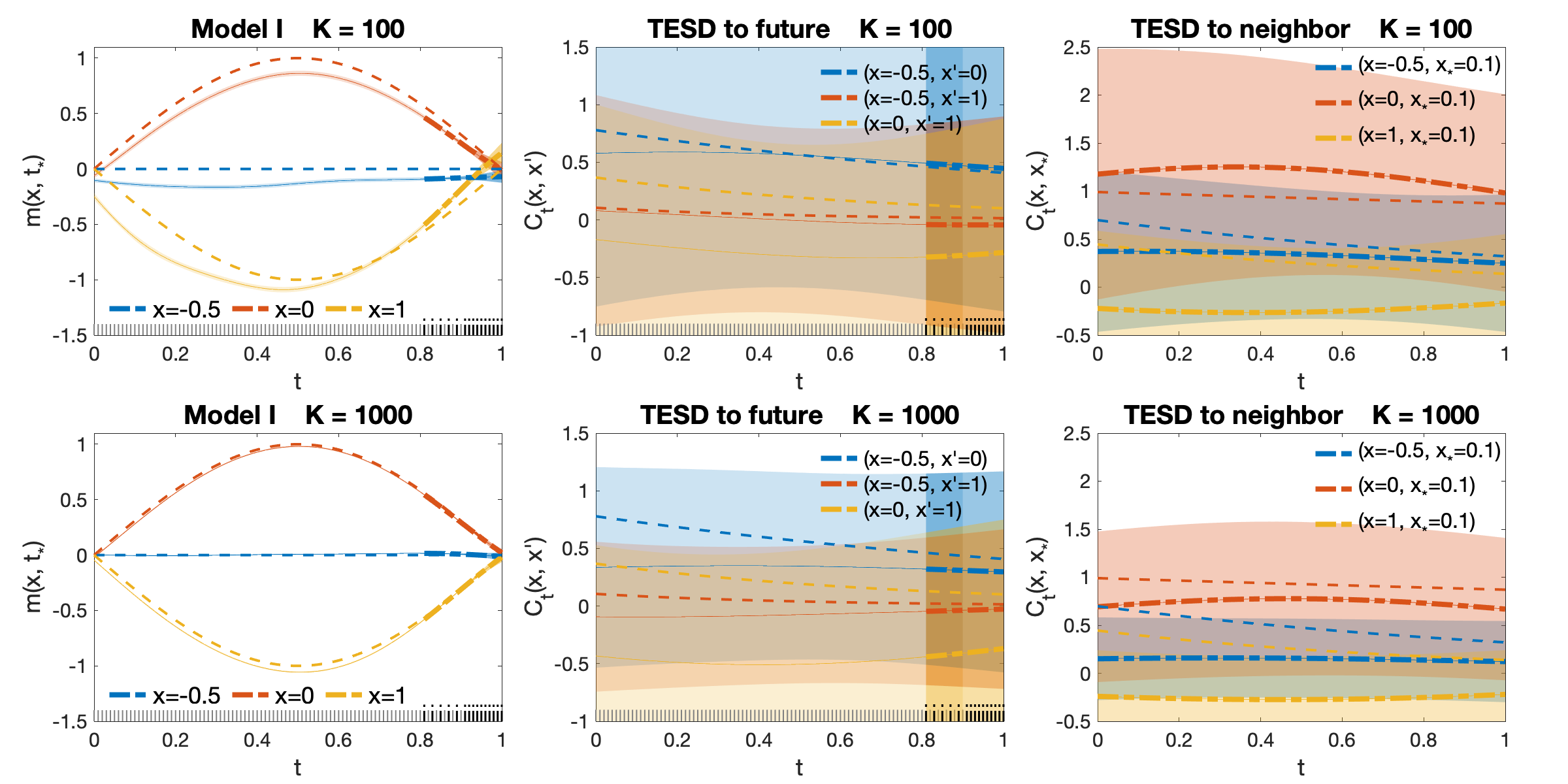} 
   \caption{Model I (qKron-prod)}
   \end{subfigure}
    \begin{subfigure}[b]{.495\textwidth}
   \includegraphics[width=1\textwidth,height=.7\textwidth]{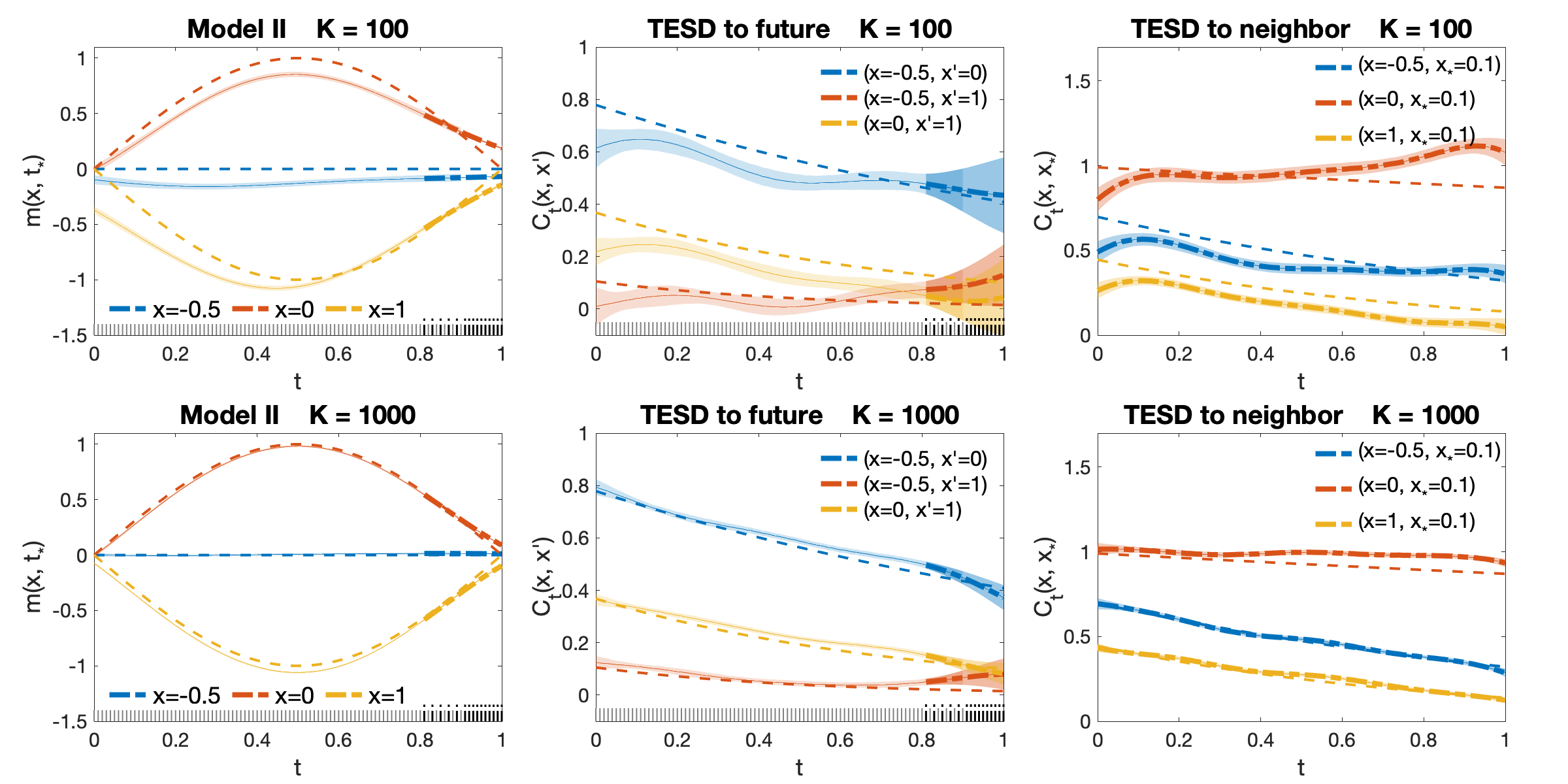}
   \caption{Model II (qKron-sum)}
   \end{subfigure}
   \caption{Prediction of mean functions $m(x,t_*)$ (left column), TESD at future times $t_*$, $C_{y|t_*}^\text{ns}(x,x')$ (middle column), and at neighboring locations $x_*$, $C_{y|t}^\text{ns}(x,x_*)$ (right column), by (a) model I and (b) model II based on $K=100$ trials (upper row) and $K=1000$ trials (lower row) of data on each panel.
   Dashed lines are true values. Solid lines with light shaded regions are results (estimates with credible bands) of the models trained on data indicated by short gray ticks. Thick dash-dot lines with dark shaded regions are results (predictions with credible bands) of the models on testing data indicated by black dash-dot ticks.}
   \label{fig:pred_STproc}
\end{figure}

Now we consider the predictions described in Section \ref{sec:infpred}. 
We train the models based on $85\%$ of the data (shorter gray ticks) and hold out the rest $15\%$ (black dash-dot ticks) for testing as illustrated in Figure \ref{fig:pred_STproc} that plots results only for models I (qKron-prod) and II (qKron-sum).
We first compare the mean prediction in the time direction, i.e. $m(x,t_*)$, which is also a well studied object.
In Figure \ref{fig:pred_STproc}, thick dash-dot lines are the predicted values of the mean function with dark shaded regions as the corresponding credible bands.
In the left columns of these two panels, we see that both models give comparable results in predicting mean functions, with better accuracy for more data (trials $K$).


Next, we consider the prediction of covariance functions, particularly in two types, namely TESD to future (Section \ref{sec:TESD2future}) and TESD to neighbor (Section \ref{sec:TESD2neighbor}).
The middle columns on the two panels of Figure \ref{fig:pred_STproc} show the first type of covariance prediction, i.e. $\mC_{y|t_*}(x,x)$.
The predicted TESD in time follows the trend of the fitted values with model II significantly better than model I (closer to truth with smaller credible bands).
Note there is higher uncertainty (wider credible bands) when such prediction comes into `no-data' (extrapolation) zone compared with that in the interpolation zone where there are still nearby training points.
Lastly, we compare the second type of covariance prediction, $\mC_{y|t}(x,x_*)$, TESD to a new neighbor $x_*=0.1$ in the right columns on two panels of Figure \ref{fig:pred_STproc}.
Note, we do not have any data at this location $x_*=0.1$, yet model II can still extend TESD (thick dash-dot lines) to new locations with decent precision, in reference to the truth (dash lines).
On the contrary, model I yields less desirable results with much higher uncertainty.

We extend the experiments to all models and repeat each for 10 times. Table \ref{tab:mspe} summarizes their mean squared prediction errors (MSPE).
Again we find model II generating significantly smaller MSPE for TESD using less time compared with the other models.
All the numeric evidences strongly support that model II is both effective and efficient in modeling and predicting TESD.

\begin{table}[ht]\scriptsize
\begin{center}
\begin{tabular}{l|l|cccc}
  \toprule
 Trials & Model & mean & TESD to future & TESD to neighbor & time \\ 
  \midrule
& stat-sep & 5.09e-3 (1.4e-6) & 7.05e-2 (1.6e-5) & 0.308 (2.4e-5) & 2280 (140) \\
& stat-nonsep & 2e-2 (3.4e-6) & 7e-2 (1.2e-3) & 0.307 (2.2e-3) & 2100 (290) \\
& nonstat-sep & 1.63e-2 (1.4e-5) & 5.07e-2 (6.2e-4) & 0.247 (2.6e-4) & 1600 (170) \\
$K=100$ & nonstat-nonsep & 1.31e-2 (7.7e-6) & 5.21e-2 (7.7e-4) & 0.257 (2.3e-4) & 2820 (370) \\
\cmidrule{2-6}
 & qKron-prod (I) & 8.32e-3 (7.3e-4) & 0.109 (0.68) & 0.327 (0.66) & 4170 (770) \\ 
 & qKron-sum (II) & 8.65e-3 (2e-4) & {\bf 1.04e-2} (0.34) & {\bf 7.06e-3} (1.3e-4) & {\bf 1800} (96) \\ 
   \midrule
 & stat-sep & 7.95e-4 (3.1e-7) & 7.05e-2 (2.2e-5) & 0.308 (2.4e-5) & 2290 (430) \\ 
 & stat-nonsep & 1.26e-2 (9.2e-7) & 6.95e-2 (8.4e-4) & 0.308 (2.4e-3) & 2150 (290) \\ 
 & nonstat-sep & 7.56e-4 (4.5e-7) & 3.5e-2 (4e-4) & 0.212 (6.4e-5) & 1670 (210) \\ 
 $K=1000$ & nonstat-nonsep & 1.45e-3 (9.3e-7) & 4.23e-2 (1.2e-3) & 0.238 (2.7e-4) & 2820 (370) \\ 
 \cmidrule{2-6}
 & qKron-prod (I) & 4.18e-4 (5.7e-6) & 0.131 (2.4e-2) & 0.187 (0.17) & 3380 (1200) \\ 
  & qKron-sum (II) & 1.6e-3 (1.3e-5) & {\bf 4.39e-3} (1.3e-2) & {\bf 1.48e-3} (4.3e-5) & {\bf 2310} (810) \\ 
  \bottomrule
\end{tabular}
\caption{Mean squared prediction errors (MSPE) for predicting mean $m(x,t_*)$ and TESD's $\mC_{y|t_*}^\text{ns}(x,x)$ and $\mC_{y|t}^\text{ns}(x,x_*)$.
Table values are median estimates of 10 repeated results with sample standard deviation in the parentheses.} 
\label{tab:mspe}
\end{center}
\end{table}

\subsection{Longitudinal Analysis of Alzheimer's Brain Images} 
 \label{sec:ADPET}
In this section we will apply the generalized STGP model to Alzheimer's neuroimaging data to study the association of brain regions in the progression of this disease.
Since other models fail to characterize TESD and involve prohibitive computation in this example, we will mainly focus on model II in the following unless stated otherwise.

Alzheimer's disease (AD) is a chronic neurodegenerative disease that affects patients' brain functions including memory, language, orientation, etc. in the elder population generally above 65.
According to the World Alzheimer Report \citep{adreport2018}, there were about 50 million people worldwide living with dementia in 2018,
and this figure is expected to skyrocket to 132 million by 2050. Yet the cause of AD is poorly understood.
Longitudinal studies have collected high resolution neuroimaging data, genetic data and clinical data in order to better understand the progress of brain degradation.
In this section, we analyze the positron emission tomography (PET) brain imaging data from the Alzheimer's Disease Neuroimaging Initiative project \citep{ADNI}
with two aims:
(i) to characterize the change of the brain structure and function over time; 
and (ii) to detect the spatial correlation between brain regions and describe its temporal evolution (TESD).

\subsubsection{Positron Emission Tomography (PET) data}
We obtain PET scans scheduled at the baseline, 6 months, 1 year, 18 months, 2 years and 3 years from the ADNI study.
There are 51 subjects in this data set, with 14 Cognitively Normal (CN), 27 Mild Cognitive Impairment (MCI) and 19 Alzheimer's Disease (AD).
Among these patients, only the MCI group has data at 18 months and the AD group is followed up until 2 years.
PET brain image scans are obtained and 
processed by co-regsitering to have the same position, averaging over 6 five-minute frames, standardizing to $160\times 160 \times 96$ voxel image grid, and smoothing to have uniform resolution.
A detailed description of PET protocol and acquisition can be found at \href{http://adni.loni.usc.edu/methods/pet-analysis-method/pet-analysis/}{http://adni.loni.usc.edu}.

We focus on a ($48$-th) slice in the middle (horizontal section) and model the images of size $160\times 160$.
For each subject $k$ at a specific time point $t$ during the study, the response function $y_k(x,t)$ represents the pixel value of location $x$ in the image being read. 
Therefore the discrete data $\{y_{ijk}\}$ have dimension $I\times J\times K$, with $I=160^2=25600$, $J\in\{5, 6, 4\}$ and $K\in\{14, 27, 19\}$.
To study the spatial dependency in these brain images, we need a kernel with discrete size $25600\times 25600$, which is enormous if it is a dense matrix.
We introduce a spatial kernel based on the graph Laplacian \citep{shen2010,ng2012,hu2015,huang2018,DUNLOP2020}. The resulting precision matrix is highly sparse (with $0.035\%$ non-zero entries) and thus amenable to an efficient learning of TESD.
See more details in Appendix \ref{apx:graphLap}.

\begin{figure}[t] 
   \centering
   \includegraphics[width=1\textwidth,height=.15\textwidth]{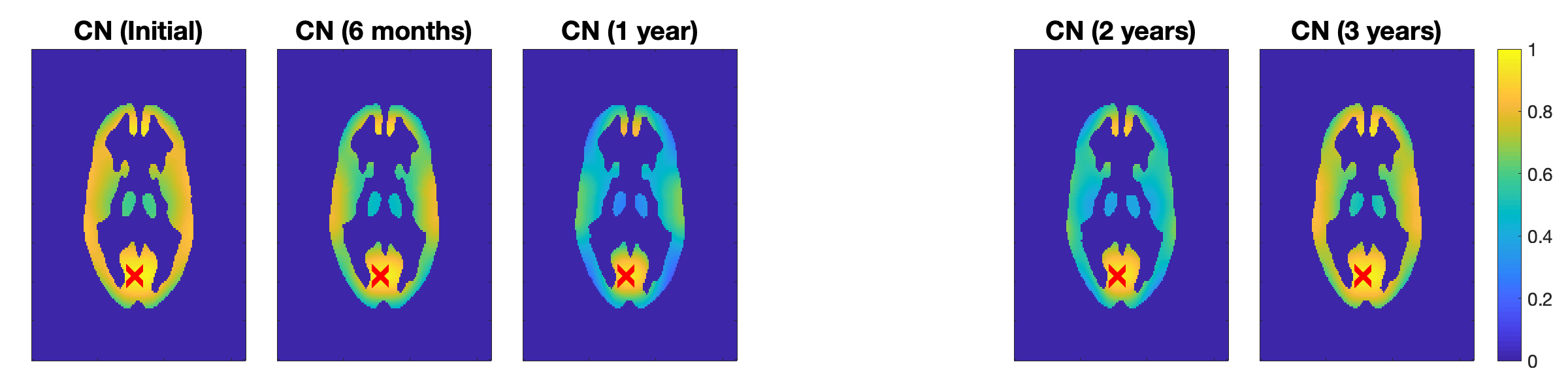} 
   \includegraphics[width=1\textwidth,height=.15\textwidth]{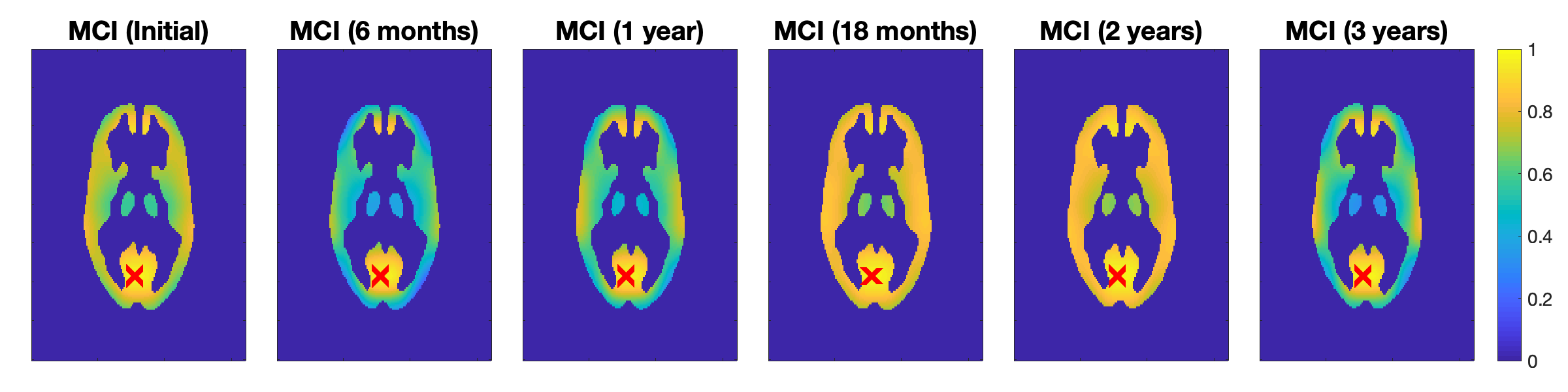} 
   \includegraphics[width=1\textwidth,height=.15\textwidth]{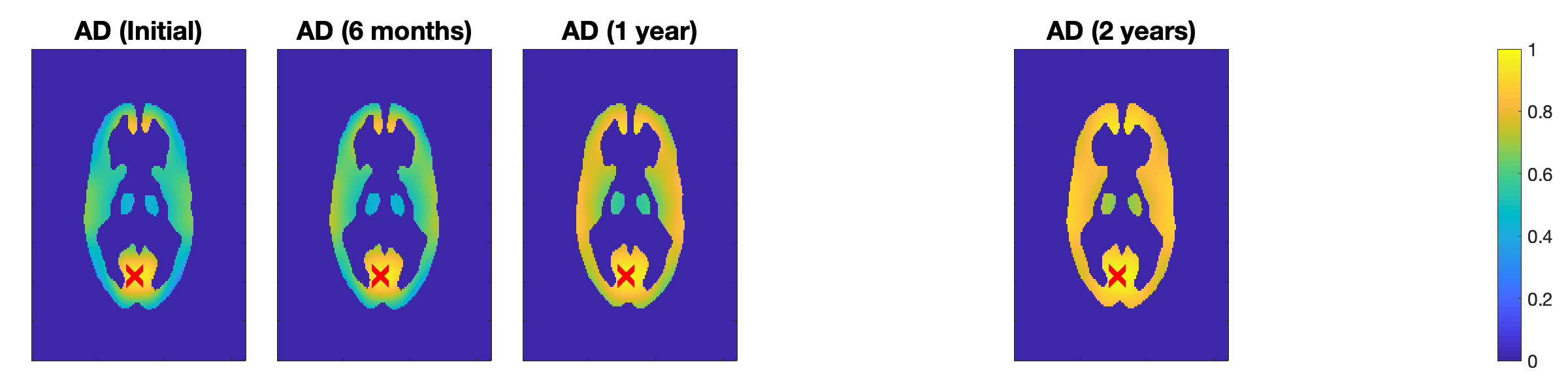} 
   \caption{Estimated correlation between the brain region of interest (ROI) and a selected point of interest (POI, red cross) for CN (top row), MCI (middle row) and AD (bottom row) respectively by the proposed model II (qKron-sum). The color at a point in ROI indicates the correlation between that point and the selected POI in the scale of $[0,1]$.}
   \label{fig:corpoi_PET}
\end{figure}
\subsubsection{Model Fit}\label{sec:PET_fit}
Now we fit model II (qKron-sum) with the graph-Laplacian based spatial kernel \eqref{eq:gLspat} to the PET brain images.
The following setting for hyper-parameters is used: $a=[1,1,1]$, $b=[0.1,1,0.1]$, $m=[0,0,0]$ and $V=[0.1,1,1]$;
however the results are not sensitive to the setting. 
The smoothness of actual time-varying spatial dependence in the brain regions is unknown, thus it is difficult to specify a prior that matches the regularity of the truth. Therefore we choose $\kappa=0$ in the prior model \eqref{eq:randcoeff}-\eqref{eq:spatkern_t1} for $\mC_{\bx|t}$. It results in an improper prior, however regularized by the likelihood (See more details in Figure \ref{fig:dyneigv_PET} and Appendix \ref{apx:moreADres}). Smoother (and more informative) priors tend to blur TESD found here (results not shown).
The truncation number of the Mercer's kernel expansion is set to $L=100$.
We run MCMC to collect $2.4\times 10^4$ samples, burn in the first 4000, and subsample every other.
The resulting $10^4$ samples are used to obtain posterior estimates of mean functions $\bM(t)$ and covariance functions $\bC_y(t)$.
Figure \ref{fig:estm_PET} shows the fitted brain images at 6 scheduled times.
The estimated brain images of patients in the control group (CN) have higher pixel values than the other two groups with bigger (blue) hollow regions.
This can be seen more clearly from the summary of their estimated pixel values in Figure \ref{fig:pixel_PET}. 

Next, we investigate TESD in the PET brain images.
TESD over this discretize field is a matrix valued function of time $t$. At each time the spatial covariance matrix is of size $25600\times 25600$, too big for a direct visualization.
Instead, we examine one row of the spatial correlation matrix by selecting a region of interest (ROI) (chosen based on the pixel values above the $83.5\%$ quantile) and a point of interest (POI) (based on the overall variance, marked as red cross) in the occipital lobe.
Figure \ref{fig:corpoi_PET} plots the ROI-POI correlations across time. 
As these spatial correlations evolve with time, the POI is highly correlated to its nearest region across all the time. It is also interesting to note the high correlation between the POI and some area in the frontal lobe.

\begin{figure}[t] 
   \centering
   \includegraphics[width=1\textwidth,height=.15\textwidth]{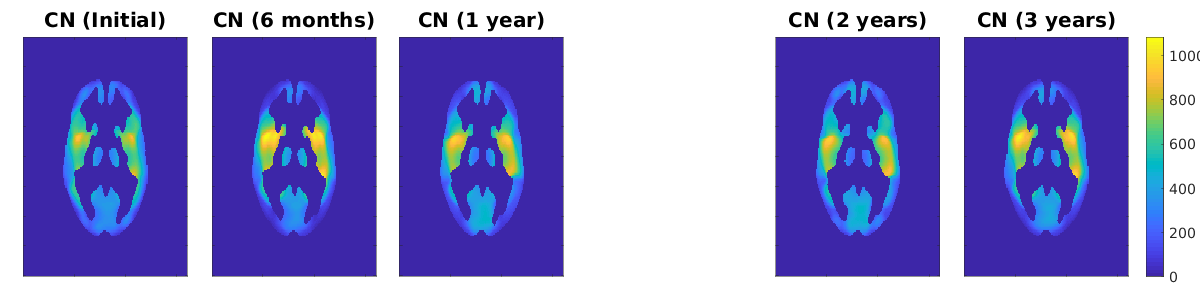} 
   \includegraphics[width=1\textwidth,height=.15\textwidth]{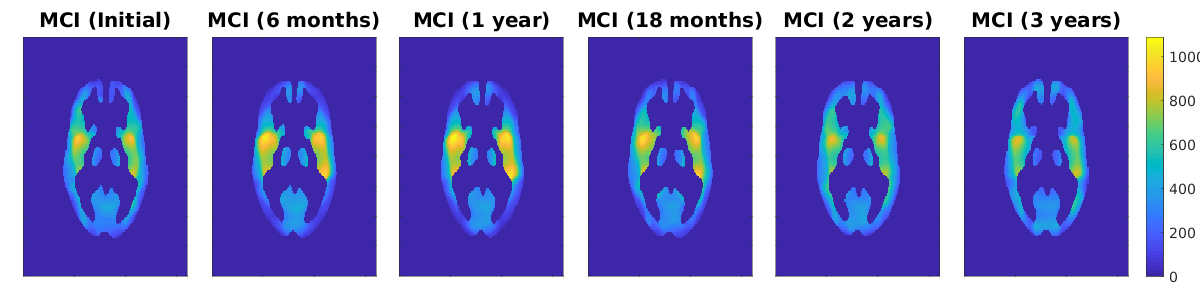} 
   \includegraphics[width=1\textwidth,height=.15\textwidth]{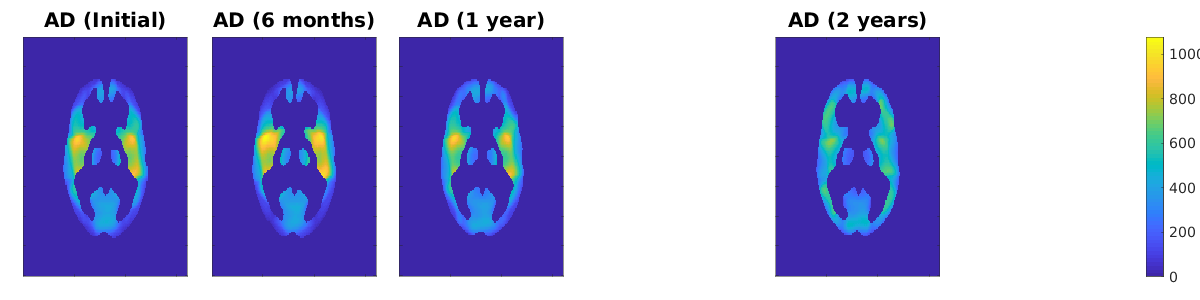} 
   \caption{Estimated graph connection of the brain images for CN (top row), MCI (middle row) and AD (bottom row) respectively by the proposed model II (qKron-sum). The color at each point represents the number of nodes connected to it, indicating the brain activeness.}
   \label{fig:estcor_PET}
\end{figure}

To better summarize TESD as full-size covariance/correlation matrices, we threshold the ($25600^2$) absolute correlation values at the top $10\%$, then we obtain the adjacency matrix ($25600\times 25600$) based on the nonzero values of the correlations. 
Finally we define the \emph{graph connection} as the diagonal of the degree matrix (row sums of the adjacency) projected back to $160\times 160$ mesh. Therefore, the value of each point on the graph connection indicates how many nodes are connected to it. Figure \ref{fig:estcor_PET} plots the graph connections of the brain ROI for different groups.
For each of these graph connections in Figure \ref{fig:estcor_PET}, the truncation at any value yields a network of connected nodes that are the most active.
As seen from Figure \ref{fig:estcor_PET}, these networks are most likely to concentrate on certain region in the temporal lobe.
We successfully characterize the dynamic changing of such connectivity network of in these brain images.
Note that the connectivity becomes weaker (thus the network of connected nodes become smaller) in the later stage for the MCI group (2 and 3 years) and the AD group (2 years),
which could serve as an indicator of brain degradation. 

To compare the generalized STGP models with other spatiotemporal models with parametric covariance listed in Table \ref{tab:models2compare}, we have to reduce the image size to $40\times 40$ for learning TESD (spatial covariance of size $1600\times 1600$) with available computing resources. Note that qKron-sum is the only sparse model. All the other dense models would otherwise require more than 200 GB memory for the original spatial covariance of size $25600\times 25600$ and take more than 1 week to gather only a few thousands of samples even on GPU.

Working with the coarsen images (by sub-sampling image pixels) for the AD patients, we conduct the similar comparison as in Section \ref{sec:sim_fit}. For each model, we collect $1.2\times 10^4$ samples, burn in the first 2000, and subsample every 5. The resulted 2000 samples are used to obtain posterior estimates.
Unlike simulation, there is no true mean or covariance in this example. Therefore, we compare the log-likelihood, weighted MSE and time consumption in Table \ref{tab:loglikmse}. The values of log-likelihood for different models are not very comparable because the log-likelihood is dominated by the log-determinant of the likelihood kernel, which scales linearly with the problem size ($IJK$). We thus consider the weighted MSE, the positive quadratic form of log-likelihood, which is also interpreted as the Mahalanobis distance between actual image and the fitted mean, weighted by the likelihood kernel. The proposed STGP model qKron-sum attains the lowest weighted MSE score using the least amount of time (almost 2\% of that for stat-nonsep and nonstat-nonsep models).

\begin{table}[ht]
\centering
\begin{tabular}{l|cccc}
  \toprule
Model & log-likelihood & weighted MSE & Time (seconds) \\ 
  \midrule
stat-sep & 247059.02 (13.04) & 59075.74 (247.80) & 57650.46 \\ 
  stat-nonsep & 247359.03 (16.73) & 57731.13 (241.00) & 270202.49 \\ 
  nonstat-sep & 143714.69 (63.98) & 60663.36 (238.87) & 177338.31 \\ 
  nonstat-nonsep & 179503.71 (151.24) & 60490.79 (234.76) & 254509.56 \\ 
  \midrule
  qKron-prod (I) & 97405.81 (45591.06) & 75654.64 (58212.91) & 163502.18 \\ 
  qKron-sum (II) & 9079.57 (64.19) & {\bf 4545.54 (75.61)} & {\bf 5087.06} \\ 
  \bottomrule
\end{tabular}
\caption{Log-likelihood, weighted mean squared error (MSE) and time for AD by various models. 
Table values are median estimates of 2000 posterior samples with sample standard deviation in the parentheses.} 
\label{tab:loglikmse}
\end{table}

To investigate the TESD results output by different spatiotemporal models, we plot ROI-POI correlations (as in Figure \ref{fig:corpoi_PET}) in Figure \ref{fig:estcorpoi_multiplemodels} and graph connections (as in Figure \ref{fig:estcor_PET}) in Figure \ref{fig:estcor_multiplemodels} by various models only for the AD patients.
Only the proposed qKron-sum model captures time-varying spatial correlations. 
Note we observe the same brain degradation in the AD group as in Figure \ref{fig:estcor_PET} illustrated by decreasing brain connection.

\begin{figure}[t] 
   \centering
   \includegraphics[width=1\textwidth,height=.15\textwidth]{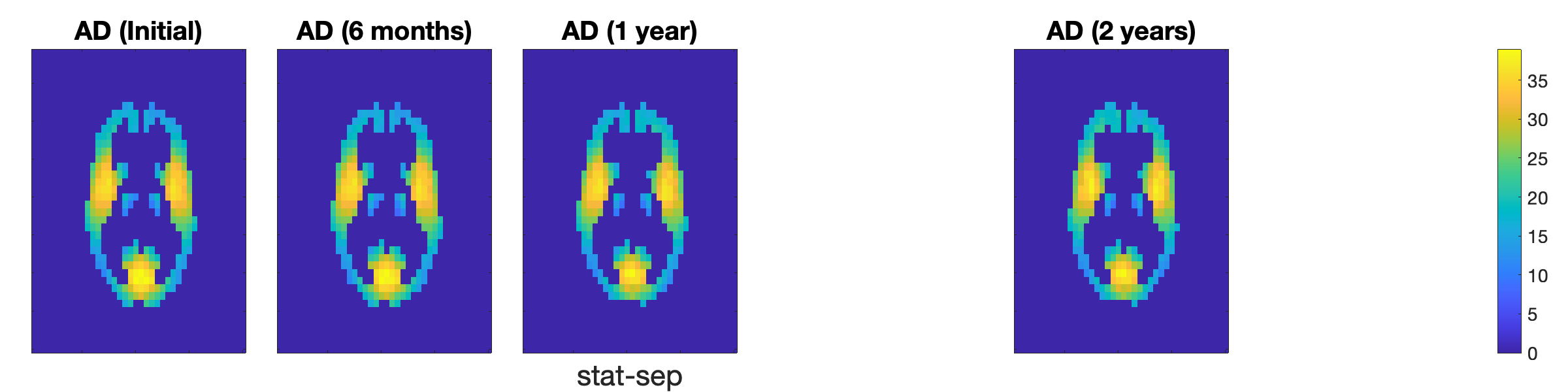}
   \includegraphics[width=1\textwidth,height=.15\textwidth]{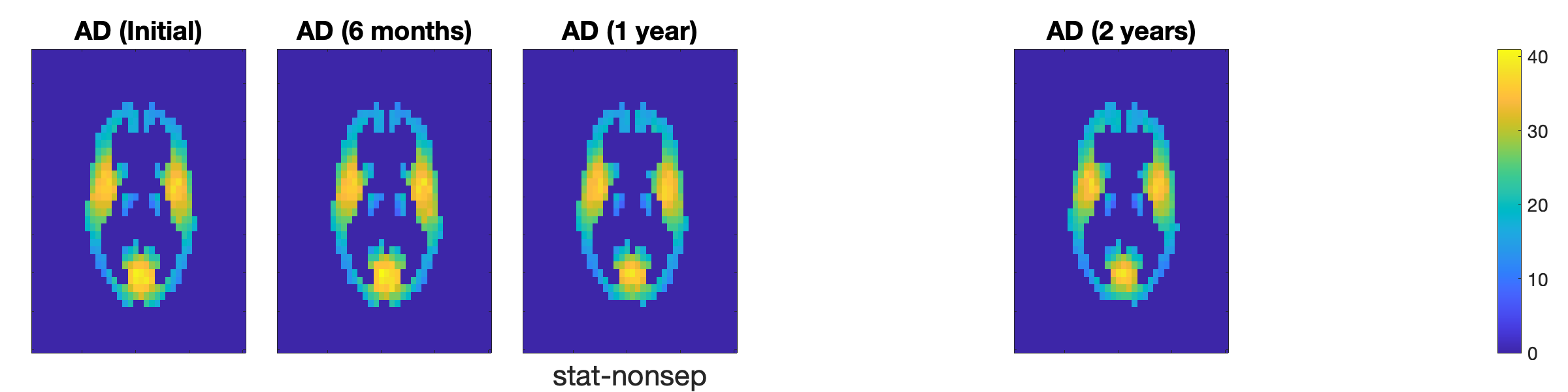}
   \includegraphics[width=1\textwidth,height=.15\textwidth]{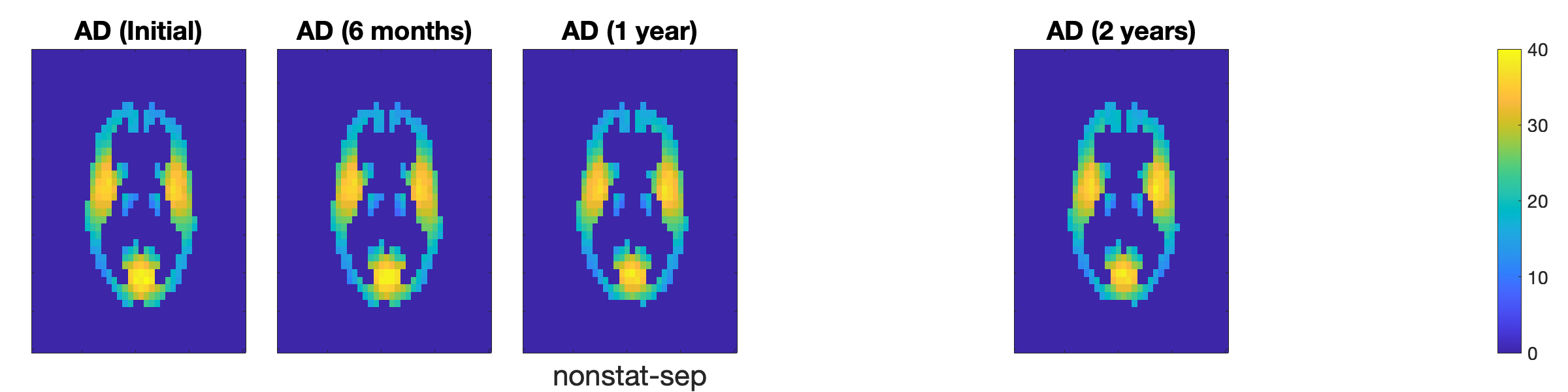}
   \includegraphics[width=1\textwidth,height=.15\textwidth]{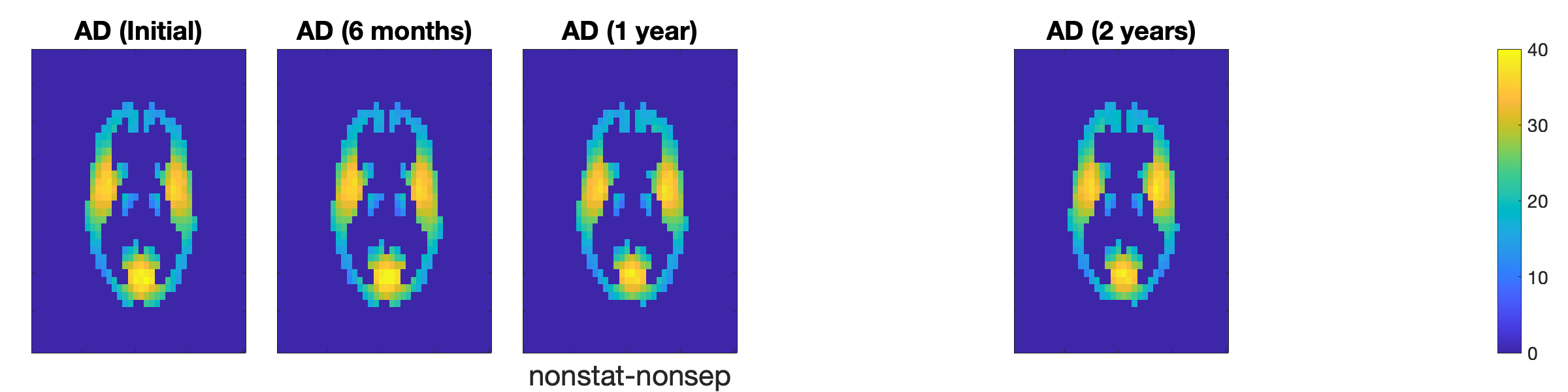}
   \includegraphics[width=1\textwidth,height=.15\textwidth]{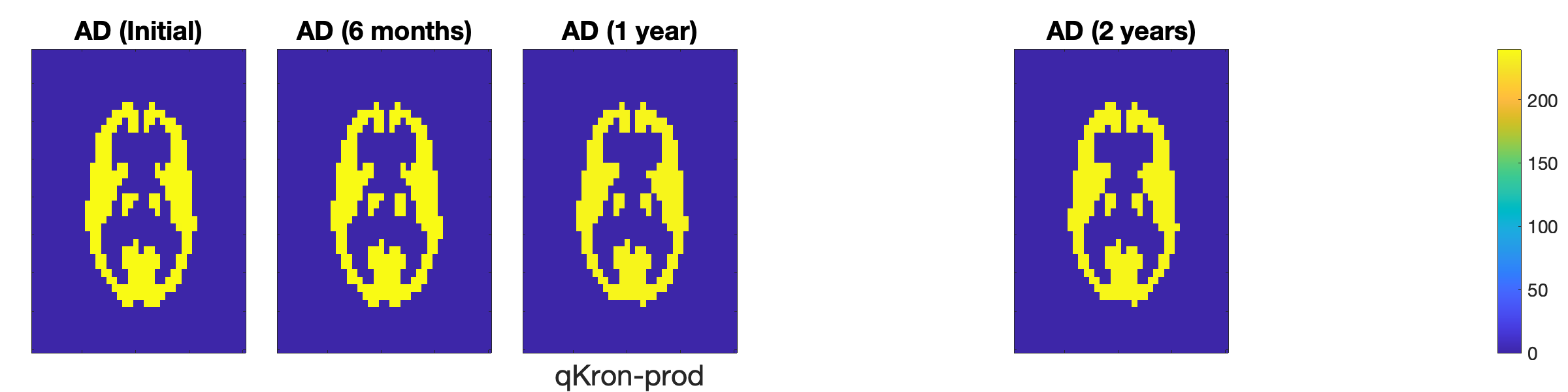}
   \includegraphics[width=1\textwidth,height=.15\textwidth]{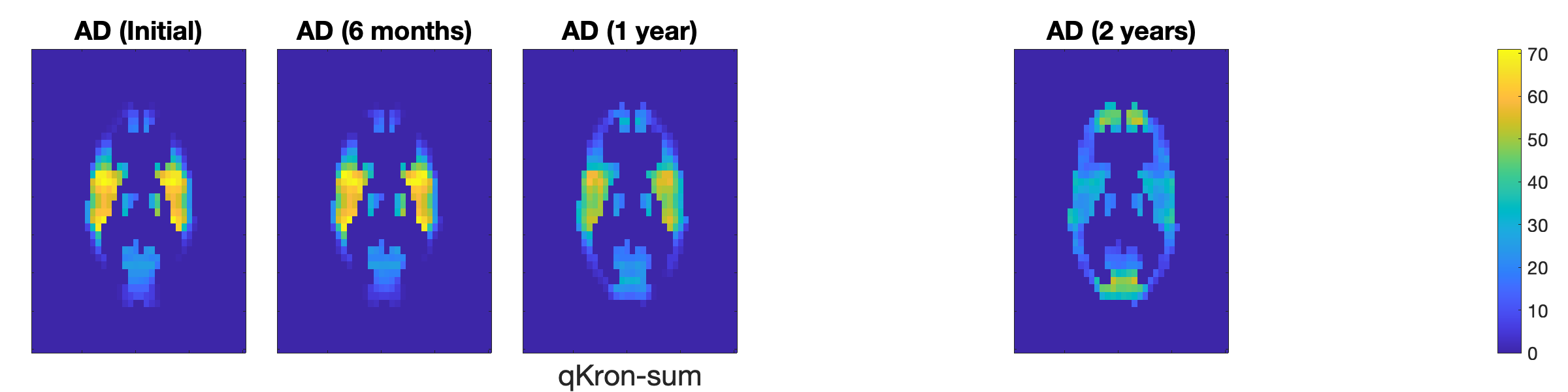} 
   \caption{Estimated graph connection of the brain images for AD by various models. The color at each point represents the number of nodes connected to it, indicating the brain activeness.}
   \label{fig:estcor_multiplemodels}
\end{figure}

\subsubsection{Model Prediction}
Next we hold out the data at the last time point for testing. For each group, the generalized STGP model II (qKron-sum) is built based on the rest of the data.
Then we predict the mean and covariance functions of the brain image at the held-out time point.
Figure \ref{fig:predm_PET} compares the actual individuals' brain images (upper row) with the predicted brain images (lower row) at the last time point. 
We can see that the prediction reflects the basic feature of the brain structure in each group.
Next, Figure \ref{fig:predcorpoi_PET} plots the correlation (TESD) between the brain ROI and the selected POI (marked as red cross) predicted at the last time point.
Note that the POI is less correlated to the middle region (thalamus), especially for the AD group. This is consistent with the fitted results shown in Figure \ref{fig:corpoi_PET}.
Readers can find more numerical results of the neuroimaging analysis in Appendix \ref{apx:moreADres}.

\begin{figure}[t] 
   \centering
   \includegraphics[width=1\textwidth,height=.3\textwidth]{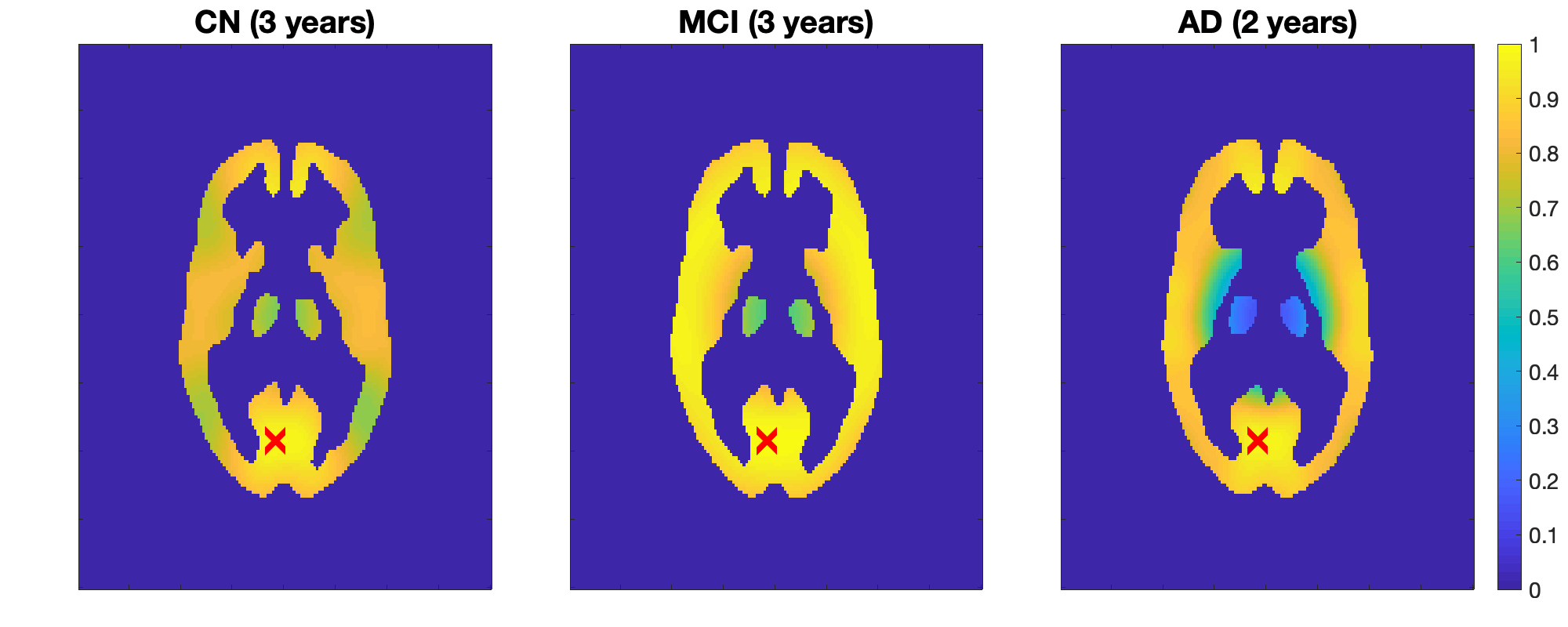}
   \caption{Predicted correlation between the brain region of interest (ROI) and a selected point of interest (POI, red cross) for CN (left), MCI (middle) and AD (right) respectively by the proposed model II (qKron-sum). The color at a point in ROI indicates the correlation between that point and the selected POI in the scale of $[0,1]$.}
   \label{fig:predcorpoi_PET}
\end{figure}

\section{Conclusion}\label{sec:conclusion}
In this paper, we generalize the separable STGP to model TESD in spatiotemporal data.
Instead of treating the space variable $\bx$ and the time variable $t$ as a joint variable $\bz=(\bx,t)$, 
we introduce time-dependence to the spatial kernel by varying its eigenvalues in the Mercer's representation and propose a novel Bayesian nonparametric non-stationary non-separable model for covariance learning.
Theoretic properties of such time-dependent spatial kernel, including the convergence, the regularity of random prior draws and the posterior contraction, have been systematically investigated.
We construct and compare two joint kernels with quasi Kronecker product and sum structures respectively.
The latter is highly sparse and well structured for learning TESD.
The advantage of the proposed (quasi Kronecker sum) model is demonstrated by a simulation study of spatiotemporal process.
It is then applied to analyze PET brain images of Alzheimer's patients 
to describe and predict the change of the brain structure in these patients and uncover TESD in their brain regions in the past and for the future.
The numerical evidences have verified the effectiveness and efficiency of the proposed model in characterizing TESD.

There are multiple future directions. For example, we can model the regularity of the priors through the decaying rate of dynamic eigenvalues \eqref{eq:randcoeff} to learn from data, e.g., $\gamma_\ell \sim\Gamma^{-1}(a_\ell, b_\ell)$ with $b_\ell/a_\ell =\mathcal O(\ell^{-\kappa/2})$ or $(\kappa-1)\sim \Gamma(a,b)$.
The proposed model uses full data on the grid of the space and time but can be readily relaxed to handle missing data.
The model can also be generalized to include covariates (regression) to explain the response variable (process) \citep{hyun2016}. Therefore, the estimation and prediction of TESD, e.g. in the brain images, can be done at individual level. We can further incorporate such information in the covariance and investigate the effect of covariates on TESD. 


The proposed scheme is designed to learn space-time interaction, more specifically, the spatial dependence conditioned commonly evolving time. The other space-time interaction, the spatial variation of temporal correction (SVTC), can be studied in the same vein by introducing space-dependence to the temporal kernel.
This is more related and comparable to the ``coregionalization" model \citep{Banerjee_2015} and could have potential applications in the study of animal migration or climate change.
More generally, we could conduct similar analysis (conditional evolution) of interaction between any two types of information that could go beyond space and time, or even among more than two types, which involves tensor analysis. We leave them to future exploration.

For the longitudinal analysis of AD patients' brain images, subjects who are followed up for the whole study are limited in number. There are more subjects dropped in the middle or missing scheduled scans from the ADNI \citep{ADNI} study thus discarded in the paper. Therefore, there could be large variance in the estimation of TESD due to the insufficient data. As ADNI continues collecting more data, they are expected to contribute a more accurate description of TESD in AD brain images that could facilitate the understanding of the mechanism behind this disease.
Another important topic is the diagnosis of AD. It would be interesting to investigate TESD of the subjects' brains before and after being diagnosed with AD, which could shed more light on the reason of such disease.

\acks{SL is supported by NSF grant DMS-2134256.
Data collection and sharing for this project was funded by the Alzheimer's Disease Neuroimaging Initiative (ADNI) (National Institutes of Health Grant U01 AG024904) and DOD ADNI (Department of Defense award number W81XWH-12-2-0012). ADNI is funded by the National Institute on Aging, the National Institute of Biomedical Imaging and Bioengineering, and through generous contributions from the following: AbbVie, Alzheimer’s Association; Alzheimer’s Drug Discovery Foundation; Araclon Biotech; BioClinica, Inc.; Biogen; Bristol-Myers Squibb Company; CereSpir, Inc.; Cogstate; Eisai Inc.; Elan Pharmaceuticals, Inc.; Eli Lilly and Company; EuroImmun; F. Hoffmann-La Roche Ltd and its affiliated company Genentech, Inc.; Fujirebio; GE Healthcare; IXICO Ltd.; Janssen Alzheimer Immunotherapy Research \& Development, LLC.; Johnson \& Johnson Pharmaceutical Research \& Development LLC.; Lumosity; Lundbeck; Merck \& Co., Inc.; Meso Scale Diagnostics, LLC.; NeuroRx Research; Neurotrack Technologies; Novartis Pharmaceuticals Corporation; Pfizer Inc.; Piramal Imaging; Servier; Takeda Pharmaceutical Company; and Transition Therapeutics. The Canadian Institutes of Health Research is providing funds to support ADNI clinical sites in Canada. Private sector contributions are facilitated by the Foundation for the National Institutes of Health (\url{http://fnih.org}). The grantee organization is the Northern California Institute for Research and Education, and the study is coordinated by the Alzheimer’s Therapeutic Research Institute at the University of Southern California. ADNI data are disseminated by the Laboratory for Neuro Imaging at the University of Southern California.}



\newpage
\appendix
\begin{center}
\textbf{\LARGE SUPPLEMENTARY MATERIAL}
\end{center}

\numberwithin{equation}{section}
\numberwithin{lem}{section}
\numberwithin{thm}{section}
\numberwithin{figure}{section}

\section{Proofs}

\wellpose*
\begin{proof}[Proof of Theorem \ref{thm:wellpose}]
\label{apx:wellpose}
We first prove both series \eqref{eq:jtkern} and \eqref{eq:likern} converge in $L^1(\mZ\times\mZ)$.
Note for \eqref{eq:jtkern} we have
\begin{align*}
&\sum_{\ell=1}^\infty \left| \int_\mZ \int_\mZ \lambda_\ell(t) \mC_t(t,t') \lambda_\ell(t') \phi_\ell(\bx) \phi_\ell(\bx') d\bz d\bz' \right| \\
&\leq \sum_{\ell=1}^\infty \left| \int_\mT \int_\mT \lambda_\ell(t) \mC_t(t,t') \lambda_\ell(t') dt dt' \right| \int_\mX \int_\mX |\phi_\ell(\bx) \phi_\ell(\bx')| d\bx d\bx' \\
&\lesssim \sum_{\ell=1}^\infty |\langle \lambda_\ell , \mC_t \lambda_\ell  \rangle| (\Vert\phi_\ell(\bx)\Vert_2^2+\Vert\phi_\ell(\bx')\Vert_2^2)/2 \\
&\leq \Vert \mC_t \Vert \sum_{\ell=1}^\infty \Vert \lambda_\ell \Vert_2^2 = \Vert \mC_t \Vert \Vert \lambda \Vert_{2,2}^2 < +\infty
\end{align*}
And for \eqref{eq:likern} we can bound
\begin{align*}
&\sum_{\ell=1}^\infty \left| \int_\mZ \int_\mZ \lambda_\ell^2(t) \delta_t(t') \phi_\ell(\bx) \phi_\ell(\bx') d\bz d\bz' \right| \\
&\leq \sum_{\ell=1}^\infty \int_\mT \lambda_\ell^2(t) dt \int_\mX \int_\mX |\phi_\ell(\bx) \phi_\ell(\bx')| d\bx d\bx'
\lesssim 
\Vert \lambda \Vert_{2,2}^2 < +\infty
\end{align*}
The convergence of series \eqref{eq:jtkern} and \eqref{eq:likern} follows by the dominated convergence theorem.

Now we prove the non-negativeness.
$\forall f(\bz)\in L^2(\mZ)$, denote $f_\ell(t):=\int_\mX f(\bz) \phi_\ell(\bx) d\bx$. Then we have
\begin{align*}
\langle f(\bz), \mC_m^\text{I} f(\bz') \rangle &= \langle f(\bz), \int_\mZ \sum_{\ell=1}^\infty \lambda_\ell(t) \mC_t(t,t') \lambda_\ell(t') \phi_\ell(\bx) \phi_\ell(\bx') f(\bz') d\bz' \rangle \\
&= \langle f(\bz), \sum_{\ell=1}^\infty \int_\mT \lambda_\ell(t) \mC_t(t,t') \lambda_\ell(t') \phi_\ell(\bx) f_\ell(t') dt' \rangle \\
&= \sum_{\ell=1}^\infty \int_\mT \int_\mT f_\ell(t)\lambda_\ell(t) \mC_t(t,t') \lambda_\ell(t')f_\ell(t') dt dt' \\
&= \sum_{\ell=1}^\infty \langle f_\ell \lambda_\ell, \mC_t \lambda_\ell f_\ell \rangle \geq 0
\end{align*}
where the convergence can be shown as above.
Similarly we have
\begin{align*}
\langle f(\bz), \mC_{y|m}^\text{II} f(\bz') \rangle &= \langle f(\bz), \int_\mZ \sum_{\ell=1}^\infty \lambda_\ell^2(t) \delta(t=t') \phi_\ell(\bx) \phi_\ell(\bx') f(\bz') d\bz' \rangle \\
&= \langle f(\bz), \sum_{\ell=1}^\infty \int_\mT \lambda_\ell^2(t) \delta(t=t') \phi_\ell(\bx) f_\ell(t') dt' \rangle
= \sum_{\ell=1}^\infty \int_\mT \lambda_\ell^2(t) f_\ell^2(t) dt \geq 0
\end{align*}
Therefore we complete the proof.
\end{proof}

\KLexpan*
\begin{proof}[Proof of Theorem \ref{thm:KLexpan}]
\label{apx:KLexpan}
Note that $\{\phi_\ell(\bx)\}_{\ell=1}\infty$ is an orthonormal basis for $L^2(\mX)$, therefore we have the series representation \eqref{eq:KLexpan} of $f(\cdot, t)$ for each $t\in \mT$.
Then we can calculate for $\mC_{\bx|t}^\half \mC_{\bx|t'}^\half \qprod \mC_t$ as in \eqref{eq:jtkern}
\begin{align*}
\bbE[f_\ell(t)] &= \bbE\left[\int_\mX f(\bx, t) \phi_\ell(\bx) d\bx\right] = \int_\mX \bbE[f(\bx, t)] \phi_\ell(\bx) d\bx = 0 \\
\bbE[f_\ell(t) f_{\ell'}(t')] &= \bbE\left[\int_\mX f(\bx, t) \phi_\ell(\bx) d\bx \int_\mX f(\bx', t') \phi_{\ell'}(\bx') d\bx' \right] \\
&= \int_\mX \int_\mX \bbE[f(\bx, t) f(\bx', t')]\phi_\ell(\bx) \phi_{\ell'}(\bx') d\bx d\bx' \\
&= \int_\mX \int_\mX \mC_{\bx|t}^\half \mC_{\bx|t'}^\half \otimes \mC_t(\bz, \bz') \phi_\ell(\bx) \phi_{\ell'}(\bx') d\bx d\bx' \\
&= \sum_{\tilde\ell=1}^\infty \lambda_{\tilde\ell}(t) \mC_t(t,t') \lambda_{\tilde\ell}(t') \int_\mX \int_\mX \phi_{\tilde\ell}(\bx) \phi_{\tilde\ell}(\bx') \phi_\ell(\bx) \phi_{\ell'}(\bx') d\bx d\bx'  \\
&= \lambda_\ell(t) \mC_t(t,t') \lambda_\ell(t') \delta_{\ell\ell'}
\end{align*}
Similarly we have $\bbE[f_\ell(t) f_{\ell'}(t')] = \lambda_\ell^2(t) \delta(t=t') \delta_{\ell\ell'}$ for $\mC_\bz=\mC_{\bx|t} \qprod \mI_t$ using the definition \eqref{eq:likern}.

Lastly, we prove the convergence of the infinite sum. Denote $F_L(\bx, t)= \sum_{\ell=1}^L f_\ell(t) \phi_\ell(\bx)$.
We have for $\mC_{\bx|t}^\half \mC_{\bx|t'}^\half \qprod \mC_t$
\begin{align*}
&\bbE[|f-F_L|^2] = \bbE[f^2] + \bbE[F_L^2] -2 \bbE[f F_L] \\
&= \mC_{\bx|t}^\half \mC_{\bx|t'}^\half \qprod \mC_t(\bz, \bz) + \bbE\left[\sum_{\ell=1}^L\sum_{\ell'=1}^L f_\ell(t) f_{\ell'}(t) \phi_\ell(\bx)\phi_{\ell'}(\bx)\right] -2 \bbE\left[f \sum_{\ell=1}^L f_\ell(t) \phi_\ell(\bx)\right]\\
&= \mC_{\bx|t}^\half \mC_{\bx|t'}^\half \qprod \mC_t(\bz, \bz) + \sum_{\ell=1}^L \lambda_\ell^2(t) \mC_t(t,t) \phi_\ell^2(\bx) - 2 \sum_{\ell=1}^L \int_\mX \bbE[f(\bx, t) f(\bx', t)] \phi_\ell(\bx') \phi_\ell(\bx) d\bx'\\
&= \mC_{\bx|t}^\half \mC_{\bx|t'}^\half \qprod \mC_t(\bz, \bz) + \sum_{\ell=1}^L \lambda_\ell^2(t) \mC_t(t,t) \phi_\ell^2(\bx) \\
&\phantom{=}- 2 \sum_{\ell=1}^L \sum_{\ell'=1}^L \lambda_{\ell'}^2(t) \mC_t(t,t) \int_\mX \phi_\ell(\bx') \phi_\ell(\bx) \phi_{\ell'}(\bx) \phi_{\ell'}(\bx') d\bx' \\
&= \mC_{\bx|t}^\half \mC_{\bx|t'}^\half \qprod \mC_t(\bz, \bz) - \sum_{\ell=1}^L \lambda_\ell^2(t) \mC_t(t,t) \phi_\ell^2(\bx) 
=\sum_{\ell=L+1}^\infty \lambda_\ell^2(t) \mC_t(t,t) \phi_\ell^2(\bx) \to 0, \; \textrm{as} \; L\to \infty
\end{align*}
The same argument (by replacing $\mC(t,t')$ with $\delta(t=t')$) yields the $L^2_\bbP$ convergence of the expansion \eqref{eq:KLexpan} for $\mC_\bz=\mC_{\bx|t} \qprod \mI_t$.
\end{proof}

\regularity*
\begin{proof}[Proof of Theorem \ref{thm:regularity}]
\label{apx:regularity}
We compute the expectation of the $(2,s,2)$-norm of $f$
\begin{align*}
\bbE[\Vert f\Vert_{2,s,2}^2] &= \sum_{\ell=1}^\infty \ell^{2s} \bbE [\Vert f_\ell\Vert_2^2] = \sum_{\ell=1}^\infty \ell^{2s} \int_\mT \bbE[f_\ell^2(t)] dt \\
&=\begin{cases}
\sum_{\ell=1}^\infty \ell^{2s} \int_\mT \mC_t(t,t) \lambda_\ell^2(t) dt \lesssim \sum_{\ell=1}^\infty \ell^{2s} \Vert \lambda_\ell\Vert_2^2, & if\; \mC_\bz=\mC_{\bx|t}^\half \mC_{\bx|t'}^\half \qprod \mC_t\\
\sum_{\ell=1}^\infty \ell^{2s} \Vert \lambda_\ell\Vert_2^2, &if\; \mC_\bz=\mC_{\bx|t} \qprod \mI_t
\end{cases} \\
&= \Vert \lambda\Vert_{2,s,2}^2 < +\infty
\end{align*}
where the equality on the second line follows from Theorem \ref{thm:KLexpan}.
This implies $f\in \ell^{2,s}(L^2(\mT))$ in probability.

Now we prove the H\"older continuity of the random function $f$ using Kolmogorov’s celebrated continuity test \citep[Theorem 3.42 of][]{Hairer2009} and \cite[Theorem 30 in section A.2.5 of][]{Dashti_2017}.
First, we consider the marginal function. By Jensen's inequality
\begin{align*}
\bbE[|f(\bx)-f(\bx')|^2] &\leq \int_\mT \bbE[|f(\bx, t)-f(\bx',t)|^2] dt = \sum_{\ell, \ell'} \int_\mT \bbE[f_\ell(t) f_{\ell'}(t)] dt \Delta\phi_\ell \Delta\phi_{\ell'} \\
&= \sum_{\ell=1}^\infty \Vert\lambda_\ell\Vert_2^2 |\phi_\ell(\bx)-\phi_{\ell}(\bx')|^2
\leq \sum_{\ell=1}^\infty \Vert\lambda_\ell\Vert_2^2 \min\{2\Vert \phi_\ell\Vert_\infty^2, \mathrm{Lip}(\phi_\ell)^2 |\bx-\bx'|^2 \} \\
&\leq 2 \sum_{\ell=1}^\infty \Vert\lambda_\ell\Vert_2^2 \Vert \phi_\ell\Vert_\infty^{2-\delta}  \mathrm{Lip}(\phi_\ell)^\delta |\bx-\bx'|^\delta
\lesssim \sum_{\ell=1}^\infty \ell^{\delta} \Vert\lambda_\ell\Vert_2^2 |\bx-\bx'|^\delta \\
&\leq \Vert\lambda\Vert_{2,s,2}^2 |\bx-\bx'|^\delta \; for \; \delta<2s
\end{align*}
where we used that $\min\{a, bx^2\}\leq a^{1-\frac{\delta}{2}} b^\frac{\delta}{2}|x|^\delta$ for $\delta\in[0,2]$.
Then by Kolmogorov’s continuity theorem there is a modification $\tilde f(\bx)$ of $f(\bx)$ in $C^{0,s'}(\mX)$ for $s'<\delta/2 <s$.

Lastly, we consider the full function $f(\bx,t)$.
\begin{align*}
&\bbE[|f(\bx,t)-f(\bx',t')|^2] = \bbE\left[\left|\sum_{\ell=1}^\infty f_\ell(t)\Delta\phi_\ell + \Delta f_\ell \phi_\ell(\bx')\right|^2\right] \\
&=\sum_{\ell, \ell'} \bbE[f_\ell(t) f_{\ell'}(t)] \Delta\phi_\ell \Delta\phi_{\ell'} + 2\sum_{\ell, \ell'} \bbE[f_\ell(t) \Delta f_{\ell'}] \Delta\phi_\ell \phi_{\ell'}(\bx') + \sum_{\ell, \ell'} \bbE[\Delta f_\ell \Delta f_{\ell'}] \phi_\ell(\bx') \phi_{\ell'}(\bx') \\
&\leq 2\left(\sum_{\ell=1}^\infty \lambda_\ell^2(t) |\Delta\phi_\ell|^2 + \sum_{\ell=1}^\infty \bbE[|\Delta f_\ell|^2] \phi_\ell^2(\bx') \right)
\leq 2\left(\sum_{\ell=1}^\infty \Vert\lambda_\ell\Vert_\infty^2 |\Delta\phi_\ell|^2 + \sum_{\ell=1}^\infty Q_{\lambda_\ell, \mC_t}(t,t') \Vert \phi_\ell\Vert_\infty^2 \right)\\
&\lesssim \sum_{\ell=1}^\infty \ell^{\delta} \Vert\lambda_\ell\Vert_\infty^2 |\bx-\bx'|^\delta + \sum_{\ell=1}^\infty \min\{\lambda_\ell^2(t)+\lambda_\ell^2(t'), C\ell^2 \Vert \lambda_\ell\Vert_\infty^2 |t-t'|^2 \} \\
&\lesssim \sum_{\ell=1}^\infty \ell^{\delta} \Vert\lambda_\ell\Vert_\infty^2 |\bx-\bx'|^\delta + \sum_{\ell=1}^\infty \ell^\delta \Vert\lambda_\ell\Vert_\infty^2 |t-t'|^\delta \\
&\lesssim \Vert\lambda\Vert_{2,s,\infty}^2 |\bz-\bz'|^\delta \; for \; \delta<2s
\end{align*}
where $\Delta f:=f(t)-f(t')$, $\Delta \phi:=\phi(\bx)-\phi(\bx')$, and the first inequality is due to Cauchy-Schwarz inequality and $2ab\leq a^2+b^2$.
The conclusion follows by Kolmogorov’s continuity theorem.
\end{proof}

\mgGP*
\begin{proof}[Proof of Corollary \ref{cor:mgGP}]
\label{apx:mgGP}
$f_\ell(t)$ can be viewed as infinite weighted sum of GP $f(\bx, t)$ thus becomes another GP.
This can be made rigorous by
approximating $\phi_\ell(\bx)$ with a sequence of simple functions $\phi_{n,\ell}=\sum_{i=1}^n a_i \bm{1}_{A_i}$ with disjoint $\{A_i\}$:
\begin{equation*}
f_\ell(t) =  \int_\mX f(\bx, t) \phi_\ell(\bx) d\bx = \lim_{n\to+\infty} \int_\mX f(\bx, t) \phi_{n,\ell}(\bx) d\bx = \sum_{i=1}^\infty a_i \int_{A_i} f(\bx, t) d\bx
\end{equation*}
Note for $\forall t\in\mT$, $f_{A_i}(t):=\int_{A_i} f(\bx, t) d\bx$ coincides with the Riemann integral.
Thus $\{f_{A_i}(t)\}$ are jointly normal as a limit of (Riemann) sum of (weighted) joint Gaussian random variables.
Therefore $f_\ell(t)$ is normal for any fixed $t\in\mT$. The same argument applies to $\bt=(t_1,\cdots, t_k)$ replacing $t$.
Thus it concludes the proof.
\end{proof}

For the dynamic spatial kernels $\mC_i = \sum_{\ell=1}^\infty \lambda_{i,\ell}^2(t)\phi_\ell \otimes \phi_\ell$,
we consider the Gaussian likelihood models $p_i\sim \mathcal N_n(\bdm_i(t), \bC_i(t))$, with $\bC_i= \sum_{\ell=1}^n \lambda_{i,\ell}^2(t)\phi_\ell(\bx) \otimes \phi_\ell(\bx')=\vect\Phi\vect\Lambda_i\tp{\vect\Phi}$, for $i=0,1$.
For $\lambda_i \in\ell^{1,s}(L^\infty(\mT))$ with some $s>0$,
we can bound the Hellinger distance $d_H$, Kullback-Leibler (K-L) divergence ($K(p_0,p_1):=\E_0 (\log (p_0/p_1))$) and K-L variation ($V(p_0,p_1):=\E_0 (\log (p_0/p_1))^2$) 
between two models with their difference in eigenvalues measured by $\Vert \cdot\Vert_{1,s,\infty}$ in the following lemma.
\begin{lem}
\label{lem:hKVbd}
Let $p_i\sim \mathcal N_n(0, \bC_i(t))$ be Gaussian models for $i=0,1$, with $\{\lambda_{i,\ell}^2(t)\}$ being the eigenvalues of $\bC_i=\vect\Phi\vect\Lambda_i\tp{\vect\Phi}$ satisfying Assumption \ref{asmp:eigbound}.
Then we have
\begin{itemize}
\item $d_H(p_0, p_1) \lesssim \Vert \lambda_0 - \lambda_1\Vert_{1,s,\infty}^\half$
\item $K(p_0, p_1) \lesssim \Vert \lambda_0 - \lambda_1\Vert_{1,s,\infty}$
\item $V(p_0, p_1) \lesssim \Vert \lambda_0 - \lambda_1\Vert_{1,s,\infty}^2$
\end{itemize}
\end{lem}
\begin{proof}
First we calculate the Kullback-Leibler divergence
\begin{equation*}
K(p_0, p_1) = \half \left\{ \tr (\bC_1^{-1}\bC_0-\bI) + \tp{(\bdm_1-\bdm_0)} \bC_1^{-1} (\bdm_1-\bdm_0) + \log \frac{|\bC_1|}{|\bC_0|} \right\}
\end{equation*}
Consider $\vect\bdm_i\equiv 0$. By the non-negativity of K-L divergence we have for general $\bC_i>0$,
\begin{equation}\label{eq:logdetbd}
\log \frac{|\bC_0|}{|\bC_1|} \leq \tr (\bC_1^{-1}\bC_0 - \bI) 
\end{equation}
Therefore we can bound K-L divergence
\begin{equation*}
K(p_0, p_1) \leq \half \{ \tr (\bC_1^{-1}\bC_0 - \bI) + \tr (\bC_0^{-1}\bC_1 - \bI) \} \leq 2 C \Vert \lambda_0 - \lambda_1\Vert_{1,s,\infty}
\end{equation*}
where we use
\begin{equation}\label{eq:tracebd}
\tr (\bC_1^{-1}\bC_0 - \bI) = \sum_{\ell=1}^n \lambda_{1,\ell}^{-2}(t) (\lambda_{0,\ell}^2(t)-\lambda_{1,\ell}^2(t)) 
\leq 2C \sum_{\ell} c_\ell^{-2} \Vert \lambda_{0,\ell}-\lambda_{1,\ell}\Vert_\infty \leq 2C \Vert \lambda_0 - \lambda_1\Vert_{1,s,\infty}
\end{equation}

Now we calculate the following K-L variation
\begin{equation*}
V(p_0,p_1) = \half \tr ((\bC_1^{-1}\bC_0 - \bI)^2) + \tp{(\bdm_1-\bdm_0)} \bC_1^{-1} \bC_0 \bC_1^{-1} (\bdm_1-\bdm_0) + K^2(p_0, p_1)
\end{equation*}
Consider $\bdm_i\equiv 0$ and we can bound it by the similar argument as \eqref{eq:tracebd}
\begin{equation*}
V(p_0,p_1) \leq C^2 \Vert \lambda_0 - \lambda_1\Vert_{2,s,\infty}^2 + 4 C^2 \Vert \lambda_0 - \lambda_1\Vert_{1,s,\infty}^2 \lesssim \Vert \lambda_0 - \lambda_1\Vert_{1,s,\infty}^2
\end{equation*}
It is easy to see that the centered K-L variation $V_0(p_0,p_1)=\V_0 (\log (p_0/p_1))=\E_0 (\log (p_0/p_1)-K(p_0, p_1))^2$ can be bounded
\begin{equation*}
V_0(p_0,p_1) \leq C^2 \Vert \lambda_0 - \lambda_1\Vert_{2,s,\infty}^2
\end{equation*}

Lastly, the squared Hellinger distance for multivariate Gaussians can be calculated
\begin{equation*}
h^2(p_0, p_1) = 1 - \frac{|\bC_0 \bC_1|^{1/4}}{\left|\frac{\bC_0+\bC_1}{2}\right|^{1/2}} \exp\left\{-\frac18 \tp{(\bdm_0-\vect\bdm_1)}\left(\frac{\bC_0+\bC_1}{2}\right)^{-1}(\bdm_0-\bdm_1)\right\}
\end{equation*}
Consider $\bdm_i\equiv 0$. Notice that $1-x\leq -\log x$, and by \eqref{eq:logdetbd} we can bound the squared Hellinger distance using the similar argument in \eqref{eq:tracebd}
\begin{equation*}
\begin{aligned}
h^2(p_0, p_1) &\leq \log \frac{\left|\frac{\bC_0+\bC_1}{2}\right|^{1/2}}{|\bC_0 \bC_1|^{1/4}} 
\leq \half \tr (\bC_0^{-\half} \bC_1^{-\half} (\bC_0+\bC_1)/2 - \bI ) \\
& \leq \frac14 \{ \tr (\bC_1^{-\half} \bC_0^\half - {\bf I}) + \tr(\bC_0^{-\half} \bC_1^\half - {\bf I}) \}
\leq \half \Vert \lambda_0 - \lambda_1\Vert_{1,s,\infty}
\end{aligned}
\end{equation*}
\end{proof}

Following \cite{Ghosal_2017}, now we prove the following posterior contraction about $\mC_{\bx|t}$ in model II. 
For the convenience of discussion, we fix all hyper-parameters at their optimal values. One can refer to \cite{vandervaart09,vanderVaart11} for varying them to scale GP.
\postcontrCII*
\begin{proof}[Proof of Theorem \ref{thm:postcontrCII}]
\label{apx:postcontrCII}
We use Theorem 1 of \cite{ghosal2007} and it suffices to verify the following two conditions (the entropy condition (2.4), and the prior mass condition (2.5)) for some universal constants $\xi, K>0$ and sufficiently large $k\in\mathbb N$:
\begin{align}
\sup_{\eps>\eps_n} \log N(\xi\eps/2, \{\lambda\in \Theta_n: d_{n,H}(\lambda,\lambda_{n,0})<\eps\},d_{n,H}) &\leq n\eps_n^2 \label{eqa:entropy}\\
\frac{\Pi_n(\lambda\in\Theta_n: k\eps_n<d_{n,H}(\lambda,\lambda_{n,0})<2 k\eps_n)}{\Pi_n(\bar B_n(\lambda_{n,0}, \eps_n))} & \leq e^{Kn\eps_n^2 k^2/2} \label{eqa:priormass}
\end{align}
where the left side of \eqref{eqa:entropy} is called \emph{Le Cam dimension} \citep{LeCam_1973,LeCam_1975}, logarithm of the minimal number of $d_{n,H}$-balls of radius $\xi\eps/2$ needed to cover a ball of radius $\eps$ around the true  value $\lambda_{n,0}$;
$\bar B_n(\lambda_{n,0}, \eps):=\{\lambda\in\Theta: \frac1n\sum_{j=1}^n K_j(\lambda_{n,0},\lambda)\leq \eps^2, \frac1n\sum_{j=1}^n V_j(\lambda_{n,0},\lambda)\leq \eps^2\}$,
with $K_j(\lambda_{n,0},\lambda)=K(P_{\lambda_{n,0},j}, P_{\lambda,j})$ and $V_j(\lambda_{n,0},\lambda)=V(P_{\lambda_{n,0},j}, P_{\lambda,j})$.
For each $1\leq \ell \leq n$, define the coordinate rate function
\begin{equation}
\varphi_{\lambda_0,\ell}(\eps_{n,\ell}) = \inf_{h\in\mbH_\ell:\Vert h-\lambda_{0,\ell}\Vert_\infty\leq \eps_{n,\ell}} \half \Vert h\Vert_{\mbH_\ell}^2 - \log \Pi_\ell(\Vert \lambda_\ell\Vert_\infty <\eps_{n,\ell})
\end{equation}
For each Gaussian random element $\lambda_\ell \in \mbB_\ell=L^\infty(\mT)$, we have $\lambda_{0,\ell}\in\bar\mbH_\ell$
and the measurable set $B_{n,\ell}\subset \mbB_\ell$ \citep[c.f. Theorem 2.1 of][]{vanderVaart08} such that
\begin{align}
\log N(3\eps_{n,\ell}, B_{n,\ell}, \Vert \cdot \Vert_\infty) &\leq 6 Cn\eps_{n,\ell}^2 \label{eqa:gp_entropy}\\
\Pi_\ell(\lambda_\ell\notin B_{n,\ell}) &\leq e^{-Cn\eps_{n,\ell}^2} \label{eqa:gp_complement}\\
\Pi_\ell(\Vert \lambda_\ell-\lambda_{0,\ell}\Vert_\infty <2\eps_{n,\ell}) &\geq e^{-n\eps_{n,\ell}^2}  \label{eqa:gp_priormass}
\end{align}
Now let $\eps_{n,\ell}=2^{-\ell} \ell^{-s}\eps_n^2$ for $\ell=1,\cdots,n$. 
Set $\Theta_n=\{\lambda\in \Theta\cap\ell^{1,s}(L^\infty(\mT)): \lambda_\ell \in B_{n,\ell}\}\subset \Theta$, and $N(\eps_n, \Theta_n, d_{n,H}) = \max_{1\leq \ell \leq n} N(3\eps_{n,\ell}, B_{n,\ell}, \Vert \cdot \Vert_\infty)$.
By Lemma \ref{lem:hKVbd} and \eqref{eqa:gp_entropy}, we have the following global entropy bound because $d_{n,H}^2(\lambda,\lambda')\leq \Vert \lambda-\lambda'\Vert_{1,s,\infty}\leq \eps_n^2$ for $\forall \lambda, \lambda'\in\Theta_n$.
\begin{equation*}
\log N(\eps_n, \Theta_n, d_{n,H}) \leq 6Cn(2^{-\ell} \ell^{-s}\eps_n^2)^2 \leq Cn\eps_n^4 \leq n\eps_n^2
\end{equation*}
which is stronger than the local entropy condition \eqref{eqa:entropy}.
Now by Lemma \ref{lem:hKVbd} and \eqref{eqa:gp_priormass} we have
\begin{align*}
\Pi_n(\bar B_n(\lambda_{n,0}, \eps_n)) & \geq \Pi_n(\Vert \lambda_{n,0}-\lambda\Vert_{1,s,\infty}\leq \eps_n^2, \Vert \lambda_{n,0}-\lambda\Vert_{1,s,\infty}^2\leq \eps_n^2) \\
&= \Pi_n(\Vert \lambda_{n,0}-\lambda\Vert_{1,s,\infty}\leq \eps_n^2) 
\geq \exp\left\{\sum_{\ell=1}^n \log \Pi_\ell(\Vert \lambda_\ell-\lambda_{0,\ell}\Vert_\infty <2\eps_{n,\ell}) \right\} \\
&\geq e^{-n\sum_{\ell=1}^n\eps_{n,\ell}^2} = e^{-Knk^2\eps_n^4/2}, \quad with \; K=2, \; k^2=\sum_{\ell=1}^n 2^{-2\ell} \ell^{-2s}
\end{align*}
Then \eqref{eqa:priormass} is immediately satisfied because the numerator is bounded by 1.
Therefore the proof is completed.
\end{proof}
\begin{rk}
This theorem generalizes Theorem 2.2 of \cite{lan_2019} where the spatial domain has fixed size $D$. Therefore the Hellinger metric, KL divergence and variance are easier to bound (Lemma B.1).
Note we do not have the complementary assertion as in Lemma 1 of \cite{ghosal2007} thus the resulting contraction is only on $\Theta_n$, weaker than that in Theorem 2.2 of \cite{lan_2019}.
\end{rk}

\contrateCII*
\begin{proof}[Proof of Theorem \ref{thm:contrateCII}]
\label{apx:contrateCII}
First, we prove that the negative logarithm of small ball probability $\varphi_0(\eps)=-\log\Pi(\Vert \lambda\Vert_{2,2}<\eps)=\mathcal O(\eps^{-\frac{2}{\kappa-1}})$.
Apply Karhunen-Lo\'eve theorem to $u_\ell$ in model \eqref{eq:randcoeff} to get
$u_\ell(t) = \sum_{i=1}^\infty Z_{\ell,i} \xi_i \phi_i(t)$ with $Z_{\ell,i} \overset{iid}{\sim} \mathcal N(0,1)$ and $\{\xi_i^2, \phi_i\}$ being the eigen-pairs of $\mC_u$.
Note $\bbE[\Vert u_\ell\Vert_2^2] = \sum_{\i=1}^\infty\xi_i^2=\tr(\mC_u)=1$.
Because normal densities with standard deviations $\sigma\geq \tau$ satisfy $\phi_\sigma(z)/\phi_\tau(z)\geq \tau/\sigma$ for every $z\in \mathbb R$, we have
\begin{align*}
\mathrm P\left(\sum_{\ell\leq L}\gamma_\ell^2 \Vert u_\ell\Vert_2^2<\eps^2\right) &= \mathrm P\left(\sum_{\ell\leq L} \gamma_\ell^2 \sum_{i=1}^\infty Z_{\ell,i}^2 \xi_i^2  <\eps^2\right)
= \int_{\sum_{\ell\leq L} \sum_{i\in\mathbb N} z_{\ell,i}^2  <\eps^2} \prod_{\ell\leq L,i\in \mathbb N} \phi_{\gamma_\ell \xi_i}(z_{\ell,i}) dz_{\ell,i} \\
&\geq \prod_{\ell=1}^L \frac{\gamma_L}{\gamma_\ell} \mathrm P\left(\sum_{\ell=1}^L \sum_{i=1}^\infty \gamma_L^2 Z_{\ell,i}^2 \xi_i^2 <\eps^2\right)
\gtrsim \left(\frac{L!}{L^L}\right)^{\frac{\kappa}{2}} \mathrm P\left( \gamma_L^2 \sum_{\ell=1}^L \Vert u_\ell\Vert_2^2 <\eps^2\right) \\
&\geq e^{-L\kappa/2} \half
\end{align*}
for $L$ large enough such that $\gamma_L^{-2}L^{-1} \eps^2\geq 1$ by the central limit theorem.
This is satisfied when $L\gtrsim \eps^{-2/(\kappa-1)}$.
On the other hand, by Markov's inequality,
\begin{align*}
\mathrm P\left(\sum_{\ell> L}\gamma_\ell^2 \Vert u_\ell\Vert_2^2<\eps^2\right) &\geq 1-\frac{1}{\eps^2} \sum_{\ell> L} \bbE[\gamma_\ell^2 \Vert u_\ell\Vert_2^2]
\geq 1-\frac{1}{\eps^2} \int_L^\infty x^{-\kappa} dx = 1-\frac{1}{(\kappa-1)L^{\kappa-1}\eps^2} \geq \half
\end{align*}
for $L$ large enough such that $(\kappa-1)L^{\kappa-1}\eps^2\geq 2$,
i.e., $L\geq \eps^{-2/(\kappa-1)} \left(\frac{2}{\kappa-1}\right)^{\frac{1}{\kappa-1}}\geq \eps^{-2/(\kappa-1)} e^{-1/(2e)}$.
Therefore we have $\mathrm P\left(\sum_{\ell=1}^\infty\gamma_\ell^2 \Vert u_\ell\Vert_2^2<2\eps^2\right) \gtrsim e^{-L\kappa/2} \frac{1}{2^2}$.
Thus the best upper bound for $\varphi_0(\eps)\lesssim L\kappa/2 \lesssim \eps^{-2/(\kappa-1)}$.

Next, we show the de-centering function $\inf_{h\in\mbH:\Vert h-\lambda\Vert_{2,2}\leq \eps} \Vert h\Vert_\mbH^2 \leq \Vert\lambda\Vert_{2,s,2}^{\kappa/s} \eps^{-(\kappa-2s)/s}$ if $\lambda\in \ell^{2,s}(L^2(\mT))$ for $s<\kappa/2$.
For every $L\in\mathbb N$, $\lambda^L:=\{\lambda_\ell\}_{\ell=1}^L\in \mbH$. Its square $(2,2)$-distance to $\lambda$ and square RHKS-norm satisfy
\begin{align*}
\Vert \lambda^L - \lambda\Vert_{2,2}^2 &= \sum_{\ell>L} \Vert \lambda_\ell \Vert_2^2 \leq L^{-2s} \Vert \lambda\Vert_{2,s,2}^2 \\
\Vert \lambda^L \Vert_\mbH^2 &= \sum_{\ell=1}^L \gamma_\ell^{-2} \Vert \lambda_\ell \Vert_2^2 \leq \Vert \lambda\Vert_{2,s,2}^2 \max_{1\leq \ell \leq L} \gamma_\ell^{-2} \ell^{-2s} \lesssim \Vert \lambda\Vert_{2,s,2}^2 \max_{1\leq \ell \leq L} \ell^{\kappa-2s} 
\end{align*}
Choosing the minimal integer $L$ such that $L\geq \Vert \lambda\Vert_{2,s,2}^{1/s}\eps^{-1/s}$ yields the result.

Finally, when the true parameter $\lambda_0\in \ell^{2,s}(L^2(\mT))$, then we get the minimal solution to the rate equation $\varphi_{\lambda_0}(\eps_n)\leq n\eps_n^2$ by
setting both $\eps_n^{-2/(\kappa-1)}\lesssim n\eps_n^2$ and $\eps_n^{-(\kappa-2s)/s} \lesssim n\eps_n^2$, which gives the rate of posterior contraction $n^{-(\frac{\kappa-1}{2}\wedge s)/\kappa}$.
\end{proof}

\begin{rk}
Posterior contraction rate is the minimal solution to the rate equation $\varphi_{\lambda_0}(\eps_n)\leq n\eps_n^2$. Therefore any rate slower than the result given above is also `a' contraction rate. The minimax rate $n^{-s/(2s+1)}$ can be attained if and only if $(\kappa-1)/2=s$, when the prior regularity matches that of the truth. When this does not happen, we can only expect suboptimal rates.
\end{rk}

\predmean*
\begin{proof}[Proof of Proposition \ref{prop:predmean}]
\label{apx:predmean}
Compute using the following formula
\begin{equation*}
\begin{aligned}
p(m(\bz_*)|\mD) &= \int p(m(\bz_*), \bM |\mD) d \bM = \int p(m(\bz_*)| \bM) p(\bM | \mD) d \bM \\
&\propto \int p(m(\bz_*)| \bM) p(\bM)  p(\mD|\bM) d \bM
= \int p(\bM, m(\bz_*)) p(\mD|\bM) d \bM
\end{aligned}
\end{equation*}
Completing the square to integrate out $\bM$ and completing the square for $m(\bz_*)$ we have
\begin{align*}
m(\bz_*)|\mD &\sim \mN(m', C') \\ 
(C')^{-1} &= C_{m_*|\bM}^{-1} - C_{m_*}^{-1} \tp c_* C_{\bM|m_*}^{-1} \bC_\text{\tiny post} C_{\bM|m_*}^{-1} c_* C_{m_*}^{-1} , \;
m' = C' C_{m_*}^{-1} \tp c_* C_{\bM|m_*}^{-1} \bC_\text{\tiny post} \bC_{\bY|\bM}^{-1} K \bar{\bY} \\
C_{m_*|\bM} &:= C_{m_*} - \tp c_* \bC_\bM^{-1} c_*, \quad  C_{\bM|m_*} := \bC_\bM - c_* C_{m_*}^{-1} \tp c_* ,
\quad \bC_\text{\tiny post}^{-1} = C_{\bM|m_*}^{-1} + K \bC_{\bY|\bM}^{-1}
\end{align*}

By Sherman-Morrison-Woodbury formula, we further have
\begin{align*}
(C')^{-1} 
=& C_{m_*|\bM}^{-1} - C_{m_*}^{-1} \tp c_* C_{\bM|m_*}^{-1} (C_{\bM|m_*}^{-1} + K \bC_{\bY|\bM}^{-1})^{-1} C_{\bM|m_*}^{-1} c_* C_{m_*}^{-1} \\
=& C_{m_*}^{-1} + C_{m_*}^{-1} \tp c_* (\bC_\bM - c_* C_{m_*}^{-1} \tp c_*)^{-1} c_* mC_{m_*}^{-1} \\
& - C_{m_*}^{-1} \tp c_* [ C_{\bM|m_*}^{-1} - (C_{\bM|m_*} + K^{-1} \bC_{\bY|\bM})^{-1}] c_* C_{m_*}^{-1} \\
=& C_{m_*}^{-1} + C_{m_*}^{-1} \tp c_* (C_{\bM|m_*} + K^{-1} \bC_{\bY|\bM})^{-1} c_* C_{m_*}^{-1} \\
=& C_{m_*}^{-1} + C_{m_*}^{-1} \tp c_* [ (\bC_\bM + K^{-1} \bC_{\bY|\bM}) - c_* C_{m_*}^{-1} \tp c_* ]^{-1} c_* C_{m_*}^{-1} \\
=& [ C_{m_*} - \tp c_* (\bC_\bM + K^{-1} \bC_{\bY|\bM})^{-1} c_* ]^{-1}
\end{align*}
and
\begin{align*}
m' 
=& C' C_{m_*}^{-1} \tp c_* C_{\bM|m_*}^{-1} (C_{\bM|m_*}^{-1} + K \bC_{\bY|\bM}^{-1})^{-1} \bC_{\bY|\bM}^{-1} K \bar{\bY} \\
=& C' C_{m_*}^{-1} \tp c_* (C_{\bM|m_*} + K^{-1} \bC_{\bY|\bM})^{-1} \bar{\bY} \\
=& \tp c_* [ I_{m_*} - (\bC_\bM + K^{-1} \bC_{\bY|\bM})^{-1} c_* C_{m_*}^{-1} \tp c_* ] (C_{\bM|m_*} + K^{-1} \bC_{\bY|\bM})^{-1} \bar{\bY} \\
=& \tp c_* (\bC_\bM + K^{-1} \bC_{\bY|\bM})^{-1} \bar{\bY}
\end{align*}
\end{proof}


\section{Posterior Inference} \label{apx:postinf}
Discretize the spatial $\mX$ and time $\mT$ domains with $I$ and $J$ points respectively.
Denote the observations on the discrete domain as $I\times J$ matrices $\bY_k$ for $k=1,\cdots, K$ trials, and thus $\bY_{I\times J\times K}=\{{\bf Y}_1,\cdots, {\bf Y}_K\}$.
We summarize model I as follows
\begin{equation}\label{eq:model1}
\begin{aligned}
\bY_k|\bM, \sigma^2_\eps \sim \mMN(\bM, \sigma^2_\eps\bI_\bx, \bI_t), &\quad \bM_{I\times J} = m(\bX,\bt) \\
m(\bx, t) \sim \GP(0, \mC_{\bx|t} \qprod \mC_t), &\quad \mC_t(t, t') = \sigma^2_t \exp(-0.5\Vert t-t'\Vert^s/\rho_t^s) \\
 \mC_\bx(\bx, \bx') = \sigma^2_\bx \exp(-0.5\Vert \bx-\bx'\Vert^s/\rho_\bx^s), &\; \mC_{\bx|t}^\half \mC_{\bx|t'}^\half \qprod \mC_t (\bz,\bz') = \sum_{\ell=1}^\infty \lambda_\ell(t) \lambda_\ell(t') \phi_\ell(\bx) \phi_\ell(\bx')  \mC_t(t,t') \\
\lambda_\ell(t)=\gamma_\ell u_\ell(t), \; u_\ell(\cdot) \overset{iid}{\sim} \GP(0, \mC_u), &\quad \mC_u(t, t') = \sigma^2_u \exp(-0.5\Vert t-t'\Vert^s/\rho_u^s) \\
\sigma^2_* \sim \Gamma^{-1}(a_*,b_*), &\quad \log\rho_* \sim \mathcal N(m_*,V_*), \quad * = \eps, \bx, t, \,\textrm{or}\, u
\end{aligned}
\end{equation}
and model II in the following
\begin{equation}\label{eq:model2}
\begin{aligned}
\VEC(\bY_k)|\bM, \bC_{\bx|t} \sim \mN(\VEC(\bM), \bC_{\bx|t}), &\quad \bM_{I\times J} = m(\bX,\bt),\; \bC_{\bx|t} = \mC_{\bx|t}(\bX,\bX; \bt) \\
m(\bx, t) \sim \GP(0, \mI_\bx \otimes \mC_t), &\quad \mC_t(t, t') = \sigma^2_t \exp(-0.5\Vert t-t'\Vert^s/\rho_t^s) \\
\mC_\bx(\bx, \bx') = \sigma^2_\bx \exp(-0.5\Vert \bx-\bx'\Vert^s/\rho_\bx^s), &\quad \mC_{\bx|t}(\bx,\bx'; t) = \sum_{\ell=1}^\infty \lambda_\ell^2(t)\phi_\ell(\bx) \phi_\ell(\bx') \\
\lambda_\ell(t)=\gamma_\ell u_\ell(t), \; u_\ell(\cdot) \overset{iid}{\sim} \GP(0, \mC_u), &\quad \mC_u(t, t') = \sigma^2_u \exp(-0.5\Vert t-t'\Vert^s/\rho_u^s) \\
\sigma^2_* \sim \Gamma^{-1}(a_*,b_*), &\quad \log\rho_* \sim \mathcal N(m_*,V_*), \quad * = \bx, t, \,\textrm{or}\,  u
\end{aligned}
\end{equation}

Truncate the kernel expansion \eqref{eq:spatkern_t1} or \eqref{eq:jtkern} at some $L$ terms.
We now focus on obtaining the posterior probability of $\bM_{I\times J}, \vect\Lambda_{J\times L}$, $\vect\sigma^2:=(\sigma^2_\eps, \sigma^2_\bx, \sigma^2_t, \sigma^2_u)$ and $\vect\rho:=(\rho_\bx, \rho_t, \rho_u)$ in the models \eqref{eq:model1} \eqref{eq:model2}.
Transform the parameters $\vect\eta:=\log(\vect\rho)$ for the convenience of calculation.
Denote $\sigma^2_\bz=(\sigma^2_\bx,\sigma^2_t)$, and $\eta^2_\bz=(\eta^2_\bx,\eta^2_t)$.
Denote $\bC_t=\mC_t(\bt,\bt)$, and $\bC_u=\mC_u(\bt,\bt)$.
Let $\bC_\bx:=\mC_\bx(\bX,\bX)=\vect\Phi \vect\Lambda_0^2\tp{\vect\Phi}$ where $\vect\Lambda_0=\diag(\{\lambda^0_\ell\})$. 
Then $\bC_{\bx|t}^j:=\mC_{\bx|t_j}(\bX,\bX)=\vect\Phi \diag(\vect\Lambda^2_j) \tp{\vect\Phi}$ where $\vect\Lambda_j=\{\lambda_{j\ell}\}$ is the $j$-th row of $\vect\Lambda$.
Denote $\bC_*(\sigma^2_*,\eta_*) = \sigma^2_* \bC_{0*}(\eta_*)$ where $*=\bx, t, \bz, \,\textrm{or}\, u$.
Once the spatial eigen-basis $\vect\Phi$ has been calculated, it will be shared across all the following calculation.
Since only normalized eigen-basis $\vect\Phi(\eta_\bx)$ is used, we can set $\sigma^2_\bx\equiv 1$ and exclude it from posterior distributions.

Notice that $\bC_\bz = \mC_t \qprod \mC_{\bx|t}^\half \mC_{\bx|t'}^\half (\bZ,\bZ)$ for model I is a full $IJ\times IJ$ matrix;
while $\bC_{\bx|t} = \diag(\{\bC_{\bx|t}^j\}_{j=1}^J)$ for model II is a block diagonal matrix formed by $J$ blocks of $I\times I$ matrices.
Both $\bC_\bz(\vect\Lambda)$ and $\bC_{\bx|t}(\vect\Lambda)$ are defined through the Mercer's expansions with fixed eigen-basis $\vect\Phi$ and newly modeled eigenvalues $\vect\Lambda$.
We make some simplifications before proceeding the calculation of posteriors.
Due to the linear independence requirement for $\vect\Phi$, we have $L\leq I$. Therefore $\bC_\bz$ is in general degenerate, and so is $\bC_{\bx|t}^j$ if $L<I$.

\subsection{Model I}
First, for model I \eqref{eq:model1} we have
\begin{equation*}
\begin{aligned}
& \log p(\bM, \vect\Lambda, \vect\sigma^2, \vect\eta | \bY) \\
=& \log p(\bY| \bM, \sigma^2_\eps) +  \log p(\bM | \vect\Lambda, \sigma^2_t, \eta_\bz) + \log p(\vect\Lambda | \sigma^2_u, \eta_u) + \sum_{*=\eps, t, u} \log p(\sigma^2_*) + \sum_{*=\bx, t, u} \log p(\eta_*) \\
=& - \half IJK \log \sigma^2_\eps - \frac{\sigma^{-2}_\eps}{2} \sum_{k=1}^K  \tr( \tp{(\bY_k-\bM)} (\bY_k-\bM) ) \\
&- \half \log |\bC_\bz(\vect\Lambda,\sigma^2_t,\eta_\bz)| - \half \tp{\VEC(\bM)} \bC_\bz^{-1} \VEC(\bM) \\
& - J\tp{\bm{1}} \log |\vect\gamma(\eta_\bx)| -\frac{L}{2} \log |\bC_u(\sigma^2_u,\eta_u)| - \half \tr (\tp\bU \bC_u^{-1} \bU) \\
& - \sum_{*=\eps, t, u} (a_*+1) \log \sigma^2_* + b_* \sigma^{-2}_* - \sum_{*=\bx, t, u} \half (\eta_*-m_*)^2/V_* \\
\end{aligned}
\end{equation*}
where $\vect\gamma(\eta_\bx)$ may (chosen as eigenvalues of $\mC_\bx$) or may not (chosen as in \eqref{eq:randcoeff}) depend on $\eta_\bx$.

\noindent $(\vect\sigma^2)$. \quad
Note the prior for $\vect\sigma^2$ is conditionally conjugate. For $* = \eps, t, \,\textrm{or}\, u$,
\begin{equation*}
\begin{aligned}
\sigma^2_* | \cdot &\sim \Gamma^{-1}(a'_*,b'_*), \quad a'_* = a_* + \Delta^a_*, \quad b'_* = b_* + \Delta^b_* \\
\Delta^a_\eps &= \half IJK, \quad \Delta^a_t = \half IJ, \quad \Delta^a_u = \half JL\\
\Delta^b_\eps &= \half \sum_{k=1}^K \tr( \tp{(\bY_k-\bM)} (\bY_k-\bM) ),\quad \Delta^b_t = \half \tp{\VEC(\bM)} \bC_{0\bz}^{-1} \VEC(\bM), \;
\Delta^b_u = \half \tr (\tp{\bU} \bC_{0u}^{-1} \bU)
\end{aligned}
\end{equation*}

\noindent $(\vect\eta)$. \quad
Given $* = \bx, t, \,\textrm{or}\, u$, we could sample $\eta_*$ using the slice sampler \citep{neal03}, which only requires log-posterior density and works well for scalar parameters,
\begin{equation*}
\begin{aligned}
\log p(\eta_\bx| \cdot) &= - \half \log |\bC_{0\bz}(\vect\Lambda,\eta_\bz)| - \half \tp{\VEC(\bM)} \bC_{0\bz}^{-1} \VEC(\bM) \sigma^{-2}_t  - J\tp{\bm{1}} \log |\vect\gamma(\eta_\bx)| - \half (\eta_\bx-m_\bx)^2/V_\bx \\
\log p(\eta_t| \cdot) &= - \half \log |\bC_{0\bz}(\vect\Lambda,\eta_\bz)| - \half \tp{\VEC(\bM)} \bC_{0\bz}^{-1} \VEC(\bM) \sigma^{-2}_t - \half (\eta_t-m_t)^2/V_t \\
\log p(\eta_u| \cdot) &=  -\frac{L}{2} \log |\bC_{0u}(\eta_u)| - \half \tr (\tp{\vect\Lambda} \bC_{0u}^{-1} \vect\Lambda)  \sigma^{-2}_u - \half (\eta_u-m_u)^2/V_u
\end{aligned}
\end{equation*}

\noindent $(\bM)$. \quad
By the definition of STGP prior, we have $\VEC(\bM) | \vect\Lambda,\sigma^2_t,\eta_\bz \sim \mN_{IJ}(\bzero, \bC_\bz(\vect\Lambda,\sigma^2_t,\eta_\bz) )$.
On the other hand, one can write the log-likelihood function as
\begin{equation*}
\begin{aligned}
\log p(\bM ; \bY) &=
-\frac{\sigma^{-2}_\eps}{2} \sum_{k=1}^K \tr( \tp{(\bY_k-\bM)} (\bY_k-\bM) ) \\
&= -\half \sum_{k=1}^K \tp{(\VEC({\bf Y}_k)-\VEC(\bM))} (\sigma^2_\eps \bI)^{-1} (\VEC({\bf Y}_k)-\VEC(\bM))
\end{aligned}
\end{equation*}
Therefore we have the analytic posterior
\begin{equation*}
\begin{aligned}
\VEC(\bM) | \cdot &\sim \mathcal N_{ND}(\bM', \bC'), \quad \bM' = \bC' \sigma^{-2}_\eps \sum_{k=1}^K \VEC({\bf Y}_k), \\
\bC' &= \left( \bC_\bz^{-1} + K \sigma^{-2}_\eps\bI \right)^{-1} = \bC_\bz \left( \bC_\bz + K^{-1}\sigma^2_\eps\bI \right)^{-1} K^{-1}\sigma^2_\eps\bI
\end{aligned}
\end{equation*}

\noindent $(\vect\Lambda)$. \quad
Using a similar argument by 
matrix Normal prior for $\vect\Lambda$, we have $\vect\Lambda | \sigma^2_u,\eta_u \sim \mMN(\bzero, \bC_u(\sigma^2_u,\eta_u), \diag(\vect\gamma^2) )$.
Therefore, we could use the elliptic slice sampler \citep[ESS,][]{murray10}, which only requires the log-likelihood
\begin{equation*}
\log p(\vect\Lambda; \bM) = - \half \log |\bC_\bz(\vect\Lambda,\sigma^2_t,\eta_\bz)| - \half \tp{\VEC(\bM)} \bC_\bz^{-1} \VEC(\bM)
\end{equation*}

\subsection{Model II}
Now we consider model II \eqref{eq:model2}
\begin{equation*}
\begin{aligned}
& \log p(\bM, \vect\Lambda, \vect\sigma^2, \vect\eta | \bY) \\
=& \log p(\bY| \bM, \bC_{\bx|t}(\vect\Lambda,\eta_\bx) ) +  \log p(\bM | \sigma^2_t, \eta_t) + \log p(\vect\Lambda | \sigma^2_u, \eta_u) + \sum_{*= t, u} \log p(\sigma^2_*) + \sum_{*=\bx, t, u} \log p(\eta_*) \\
=& - \frac{K}{2} \log |\bC_{\bx|t}(\vect\Lambda,\eta_\bx)| - \half \sum_{k=1}^K \tp{\VEC(\bY_k-\bM)} \bC_{\bx|t}^{-1} \VEC(\bY_k-\bM) \\
&- \frac{I}{2} \log |\bC_t(\sigma^2_t,\eta_t)| - \half \tr ( \bC_t^{-1} \tp{\bM} \bM) \\
& - J\tp{\bm{1}} \log |\vect\gamma(\eta_\bx)| -\frac{L}{2} \log |\bC_u(\sigma^2_u,\eta_u)| - \half \tr (\tp\bU \bC_u^{-1} \bU) \\
& - \sum_{*=t, u} (a_*+1) \log \sigma^2_* + b_* \sigma^{-2}_* - \sum_{*=\bx, t, u} \half (\eta_*-m_*)^2/V_* \\
\end{aligned}
\end{equation*}

\noindent $(\vect\sigma^2)$. \quad
Note the prior for $\vect\sigma^2$ is conditionally conjugate. For $* = t \,\textrm{or}\, u$,
\begin{equation*}
\begin{aligned}
\sigma^2_* | \cdot &\sim \Gamma^{-1}(a'_*,b'_*), \quad a'_* = a_* + \Delta^a_*, \quad b'_* = b_* + \Delta^b_* \\
\Delta^a_t &= \half IJ, \quad \Delta^a_u = \half JL, \quad
\Delta^b_t = \half \tr ( \bC_{0t}^{-1} \tp{\bM} \bM), \quad \Delta^b_u = \half \tr (\tp{\bU} \bC_{0u}^{-1} \bU)
\end{aligned}
\end{equation*}

\noindent $(\vect\eta)$. \quad
Given $* = \bx, t, \,\textrm{or}\, u$, we could sample $\eta_*$ using the slice sampler \citep{neal03}, which only requires log-posterior density and works well for scalar parameters,
\begin{equation*}
\begin{aligned}
\log p(\eta_\bx| \cdot) =& - \frac{K}{2} \log |\bC_{\bx|t}(\vect\Lambda,\eta_\bx)| - \half \sum_{k=1}^K \tp{\VEC(\bY_k-\bM)} \bC_{\bx|t}^{-1} \VEC(\bY_k-\bM) \\
& - J\tp{\bm{1}} \log |\vect\gamma(\eta_\bx)| - \half (\eta_\bx-m_\bx)^2/V_\bx \\
\log p(\eta_t| \cdot) =& - \frac{I}{2} \log |\bC_{0t}(\eta_t)| - \half \tr ( \bC_{0t}^{-1} \tp{\bM} \bM) \sigma^{-2}_t - \half (\eta_t-m_t)^2/V_t \\
\log p(\eta_u| \cdot) =&  -\frac{L}{2} \log |\bC_{0u}(\eta_u)| - \half \tr (\tp{\vect\Lambda} \bC_{0u}^{-1} \vect\Lambda)  \sigma^{-2}_u - \half (\eta_u-m_u)^2/V_u
\end{aligned}
\end{equation*}

\noindent $(\bM)$. \quad
By the definition of STGP prior, we have $\VEC(\bM) | \sigma^2_t,\eta_t \sim \mN_{IJ}(\bzero, \bC_t(\sigma^2_t,\eta_t) \otimes \bI_\bx )$.
On the other hand, one can write the log-likelihood function as
\begin{equation*}
\begin{aligned}
\log p(\bM ; \bY) &= - \half \sum_{k=1}^K \tp{\VEC(\bY_k-\bM)} \bC_{\bx|t}^{-1} \VEC(\bY_k-\bM) \\
&= -\half \sum_{k=1}^K \tp{(\VEC(\bM)-\VEC({\bf Y}_k))} \bC_{\bx|t}^{-1} (\VEC(\bM)-\VEC({\bf Y}_k))
\end{aligned}
\end{equation*}
Therefore we have the analytic posterior
\begin{equation*}
\begin{aligned}
\VEC(\bM) | \cdot &\sim \mathcal N_{ND}(\bM', \bC'), 
\quad \bM' = \bC' \bC_{\bx|t}^{-1} \sum_{k=1}^K \VEC({\bf Y}_k), \\
 \bC' &= \left( \bC_t^{-1}\otimes \bI_\bx + K \bC_{\bx|t}^{-1} \right)^{-1} =(\bC_t\otimes \bI_\bx) \left( \bC_t\otimes \bI_\bx + K^{-1} \bC_{\bx|t} \right)^{-1} K^{-1} \bC_{\bx|t}
\end{aligned}
\end{equation*}

\noindent $(\vect\Lambda)$. \quad
Using a similar argument by 
matrix Normal prior for $\vect\Lambda$, we have $\vect\Lambda | \sigma^2_u,\eta_u \sim \mMN(\bzero, \bC_u(\sigma^2_u,\eta_u), \diag(\vect\gamma^2) )$.
Therefore, we could use the elliptic slice sampler \citep[ESS,][]{murray10}, which only requires the log-likelihood
\begin{equation*}
\log p(\vect\Lambda; \bY) = - \frac{K}{2} \log |\bC_{\bx|t}(\vect\Lambda,\eta_\bx)| - \half \sum_{k=1}^K \tp{\VEC(\bY_k-\bM)} \bC_{\bx|t}^{-1} \VEC(\bY_k-\bM)
\end{equation*}

\subsection{Computational Advantage of Model II}\label{apx:compadv}

The most intensive computation as above involves the inverse and determinant of the posterior covariance kernel $\bC^*$ for two models:
\begin{equation*}
\bC^*_\text{I}:=\bC_\bz(\vect\Lambda,\sigma^2_t,\eta_\bz) + K^{-1} \sigma^2_\eps\bI,\qquad
\bC^*_\text{II}:=\bC_t(\sigma^2_t,\eta_t)\otimes \bI_\bx + K^{-1} \bC_{\bx|t}(\vect\Lambda,\eta_\bx)
\end{equation*}
Their structure dictates different amount of computation required. Actually, we can show that the kernel of model II, $\bC^*_\text{II}$, has computational advantage over that for model I.

Note, according to the definition of the dynamic spatial kernel \eqref{eq:spatkern_t1}, we can rewrite
\begin{equation*}
\bC_{\bx|t} = \diag(\{\bC_{\bx|t}^j\}_{j=1}^J) = \diag( \{\vect\Phi \diag(\vect\Lambda^2_j) \tp{\vect\Phi}\} )
= (\bI_t\otimes \vect\Phi) \diag(\tp\VEC (\vect\Lambda^2)) (\bI_t\otimes \tp{\vect\Phi})
\end{equation*}
where $\tp\VEC(\cdot)$ is row-wise vectorization.
Then by the Sherman-Morrison-Woodbury formula we have
\begin{align*}
(\bC^*_\text{II})^{-1} &= (\bC_t^{-1}\otimes \bI_\bx) - (\bC_t^{-1}\otimes \bI_\bx) (\bI_t\otimes \vect\Phi) \cdot \\ 
& \phantom{= \quad}[ K\diag(\tp\VEC (\vect\Lambda^{-2})) + (\bI_t\otimes \tp{\vect\Phi}) (\bC_t^{-1}\otimes \bI_\bx) (\bI_t\otimes \vect\Phi) ]^{-1} 
 (\bI_t\otimes \tp{\vect\Phi}) (\bC_t^{-1}\otimes \bI_\bx)\\
&= (\bC_t^{-1}\otimes \bI_\bx) - (\bC_t^{-1}\otimes \vect\Phi)  [ K\diag(\tp\VEC (\vect\Lambda^{-2})) + (\bC_t^{-1}\otimes \bI_L) ]^{-1} 
(\bC_t^{-1}\otimes \tp{\vect\Phi})\\
&= \bC_t^{-1}\otimes (\bI_\bx- \vect\Phi \tp{\vect\Phi}) + (\bI_t\otimes \vect\Phi) [ K^{-1}\diag(\tp\VEC (\vect\Lambda^2)) + (\bC_t \otimes \bI_L) ]^{-1} 
(\bI_t\otimes \tp{\vect\Phi})
\end{align*}
Similarly we have
\begin{align*}
 \bC' &= \bC_t\otimes (\bI_\bx- \vect\Phi \tp{\vect\Phi}) + (\bI_t\otimes \vect\Phi) [ K\diag(\tp\VEC (\vect\Lambda^{-2})) + (\bC_t^{-1} \otimes \bI_L) ]^{-1} 
(\bI_t\otimes \tp{\vect\Phi}) \\
(\bC')^\half &= \bC_t^\half\otimes (\bI_\bx- \vect\Phi \tp{\vect\Phi}) + (\bI_t\otimes \vect\Phi) [ K\diag(\tp\VEC (\vect\Lambda^{-2})) + (\bC_t^{-1} \otimes \bI_L) ]^{-\half} (\bI_t\otimes \tp{\vect\Phi})
\end{align*}
where we use the following calculation that is numerically more stable
\begin{align*}
&[ K\diag(\tp\VEC (\vect\Lambda^{-2})) + (\bC_t^{-1} \otimes \bI_L) ]^{-1} 
=  K^{-1}\diag(\tp\VEC (\vect\Lambda^2) \\ &- K^{-1}\diag(\tp\VEC (\vect\Lambda^2) [K^{-1}\diag(\tp\VEC (\vect\Lambda^2)) + (\bC_t \otimes \bI_L) ]^{-1} K^{-1}\diag(\tp\VEC (\vect\Lambda^2)
\end{align*}
Based on the matrix determinant lemma we can calculate
\begin{align*}
\det(\bC^*_\text{II}) &= \det [ K\diag(\tp\VEC (\vect\Lambda^{-2})) + (\bI_t\otimes \tp{\vect\Phi}) (\bC_t^{-1}\otimes \bI_\bx) (\bI_t\otimes \vect\Phi) ]\\
&\phantom{=\;} \det (K^{-1}\diag(\tp\VEC (\vect\Lambda^2))) \det (\bC_t\otimes \bI_\bx) \\
&= \det [ \diag(\tp\VEC (\vect\Lambda^{-2})) + K^{-1}(\bC_t^{-1}\otimes \bI_L) ] \prod_{j,\ell} \lambda_{j\ell}^2 \det (\bC_t)^I \\
&= \det [ \bC_t\otimes \bI_L + K^{-1} \diag(\tp\VEC (\vect\Lambda^{2})) ] \det (\bC_t)^{I-L}
\end{align*}

However in model I, we note that $\bC_\bz=\bC_{\bx|t}^\half \bC_{\bx|t'}^\half \odot (\bC_t\otimes \bm{1}_{I\times I})$, where $\odot$ is element-wise multiplication, and $\bm{1}_{I\times I}$ is an $I\times I$ matrix with all elements $1$.
According to \eqref{eq:jtkern}, we have
\begin{equation*}
\bC_{\bx|t}^\half \bC_{\bx|t'}^\half = [\vect\Phi \diag(\vect\Lambda_j) \diag(\vect\Lambda_{j'}) \tp{\vect\Phi}] = (\bI_t\otimes \vect\Phi) \VEC \{\diag(\vect\Lambda_j)\}\tp\VEC \{\diag(\vect\Lambda_{j'}) \} (\bI_t\otimes \tp{\vect\Phi})
\end{equation*}
where $\VEC\{\cdot\}$ and $\tp\VEC\{\cdot\}$ are column/row wise vectorization of block matrices.
Applying the above equation to the inverse or determinant of $\bC^*_\text{I}$ does not simplify computation in general.


\subsection{Spatial Kernel Based On Graph Laplacian}\label{apx:graphLap}
\emph{Graph Laplacian}, also known as discrete Laplace operator, is a matrix representation of a graph.
It is a popular tool for image processing, clustering and semi-supervised/unsupervised learning on graphs \citep{chung1997,smola2003}.
For a weighted graph $G=(Z,W)$ with $Z$ being the vertices $\{x_i\}_{i=1}^n$ of the graph and $W$ being the edge weight matrix,
the graph Laplacian $\mathrm{L}$ is defined as follows
\begin{equation}\label{eq:gLap}
\mathrm{L}=D-W, \quad W=[w_{ij}], \; w_{ij}=\eta_\eps(|x_i-x_j|), \quad D=\diag\{d_{ii}\}, \; d_{ii} = \sum_{x_j\sim x_i} w_{ij}
\end{equation}
where $\eta_\eps$ is some distance function, e.g. Euclidean distance, $D$ is called degree matrix, and $x_i\sim x_j$ means two vertices $x_i, x_j$ connected with an edge.
When $w_{ij}\equiv 1$, $W$ is also called adjacency matrix, denoted as $A$.
If we assume $x_j\in \Omega$ are sampled i.i.d from a probability measure $\mu$ supported on the graph domain $\Omega$ with smooth Lebesgue density $\rho$ bounded above and below by positive constants, then $\mathrm{L}$ can be viewed as an approximation of the Laplace operator $\mL$ in the following PDE:
\begin{equation*}
\mL u = -\frac{1}{\rho} \nabla \cdot (\rho^2\nabla u),\quad x\in \Omega, \qquad \frac{\pa u}{\pa n} = 0, \quad  x\in \pa \Omega.
\end{equation*}

Based on the graph Laplacian, we can define the following discrete spatial kernel for the brain images \citep{DUNLOP2020}
\begin{equation}\label{eq:gLspat}
C_\bx = (s_nL + \tau^2 I)^{-s}, \quad s_n=o\left(\frac{1}{n^{1-2/d}\log^{\delta(d=2)/2+2/d}n}\right)
\end{equation}
where $d$ is the spatial dimension, i.e. $d=2$ for the chosen slice of brain images.
We choose $s_n=\frac{1}{n^{1-2/d}\log^{1+2/d}n}$ in this experiment.
Further assuming conditions (open, connected, and with smooth boundary) on the graph domain $\Omega$, \cite{DUNLOP2020} prove that
for $s>d/2$ and $\tau\geq0$, Gaussian measure $\mN(0,\mC)$ with $\mC=(\mL+\tau^2\mI)^{-s}$ is well-defined on the weighted Hilbert space $L^2_\mu$.
In another word, the graph-Laplacian based spatial kernel \eqref{eq:gLspat} is well-behaved for large graphs including the brain images we investigate with $n=25600$ nodes.

To obtain the spatial kernel \eqref{eq:gLspat} for the analysis of PET scans, we first construct the graph Laplacian.
On the $160\times 160$ mesh grid, each node is connected to its $(2w+1)^2-1$ neighbors, where we can choose $w=1$ for example.
Depending on the location, some nodes may have $2w(w+1)+w$ neighbors (on the edge) or $w(w+1)+w$ neighbors (at the corner).
The resulting graph Laplacian matrix $\mathrm{L}$ has the size $25600\times 25600$ but is highly sparse (with the density of non-zero entries $3.4864\times 10^{-4}$).
We also assume a hyper-prior for $\tau^2\sim \log-\mN(m_\bx,V_\bx)$ and fix $s=2$ in this experiment.
Then for given $\tau^2$, we calculate the precision matrix $C_\bx^{-1}$ based on \eqref{eq:gLspat}. Hence the dense covariance matrix $C_\bx$ of size $25600\times 25600$ is not directly calculated in the inference procedure.

\begin{figure}[t] 
   \begin{subfigure}[b]{.495\textwidth}
    \includegraphics[width=1\textwidth,height=.6\textwidth]{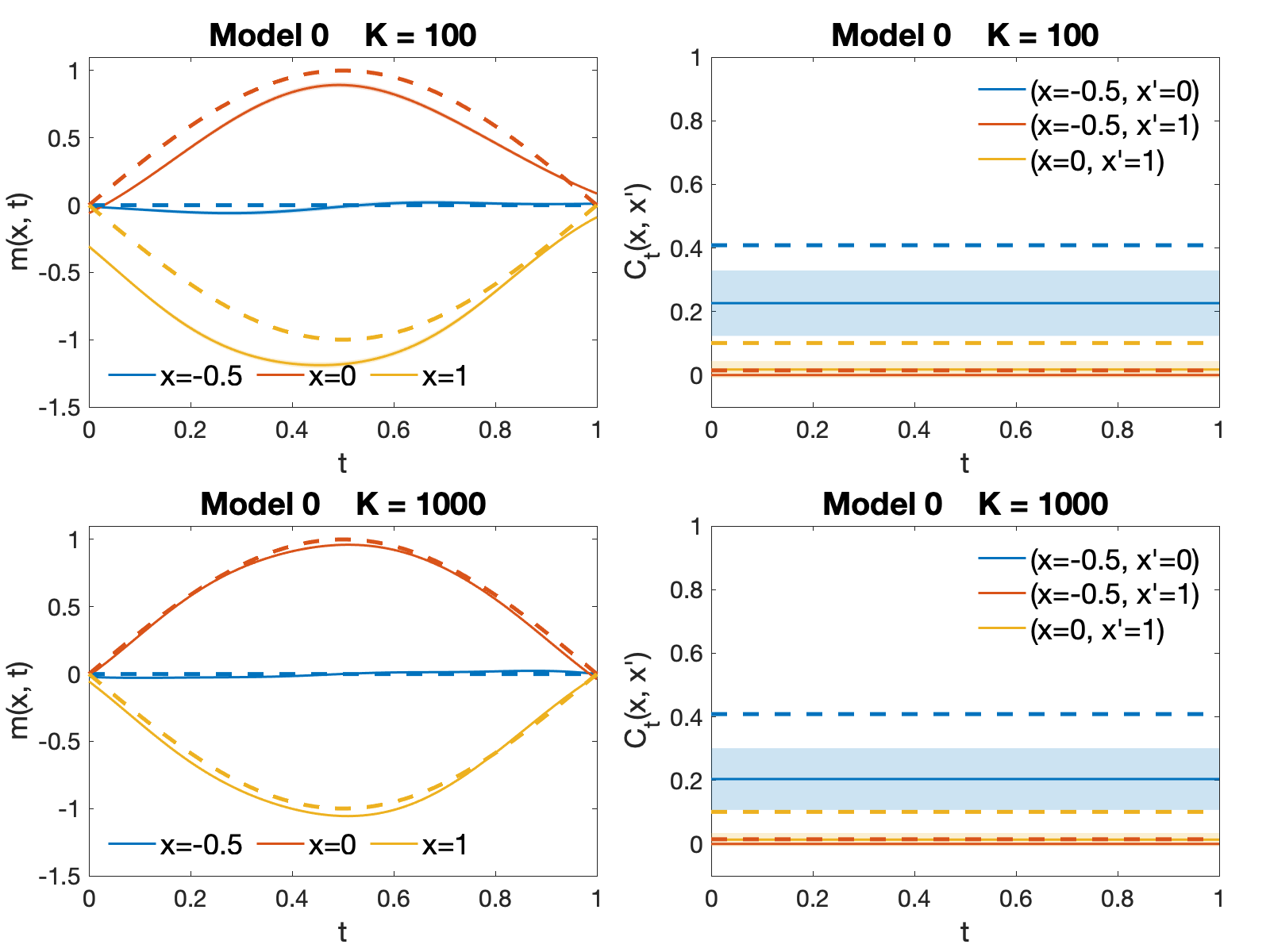}
    \caption{Model 0 (stat-sep)}
    \end{subfigure}
    \begin{subfigure}[b]{.495\textwidth}
   \includegraphics[width=1\textwidth,height=.6\textwidth]{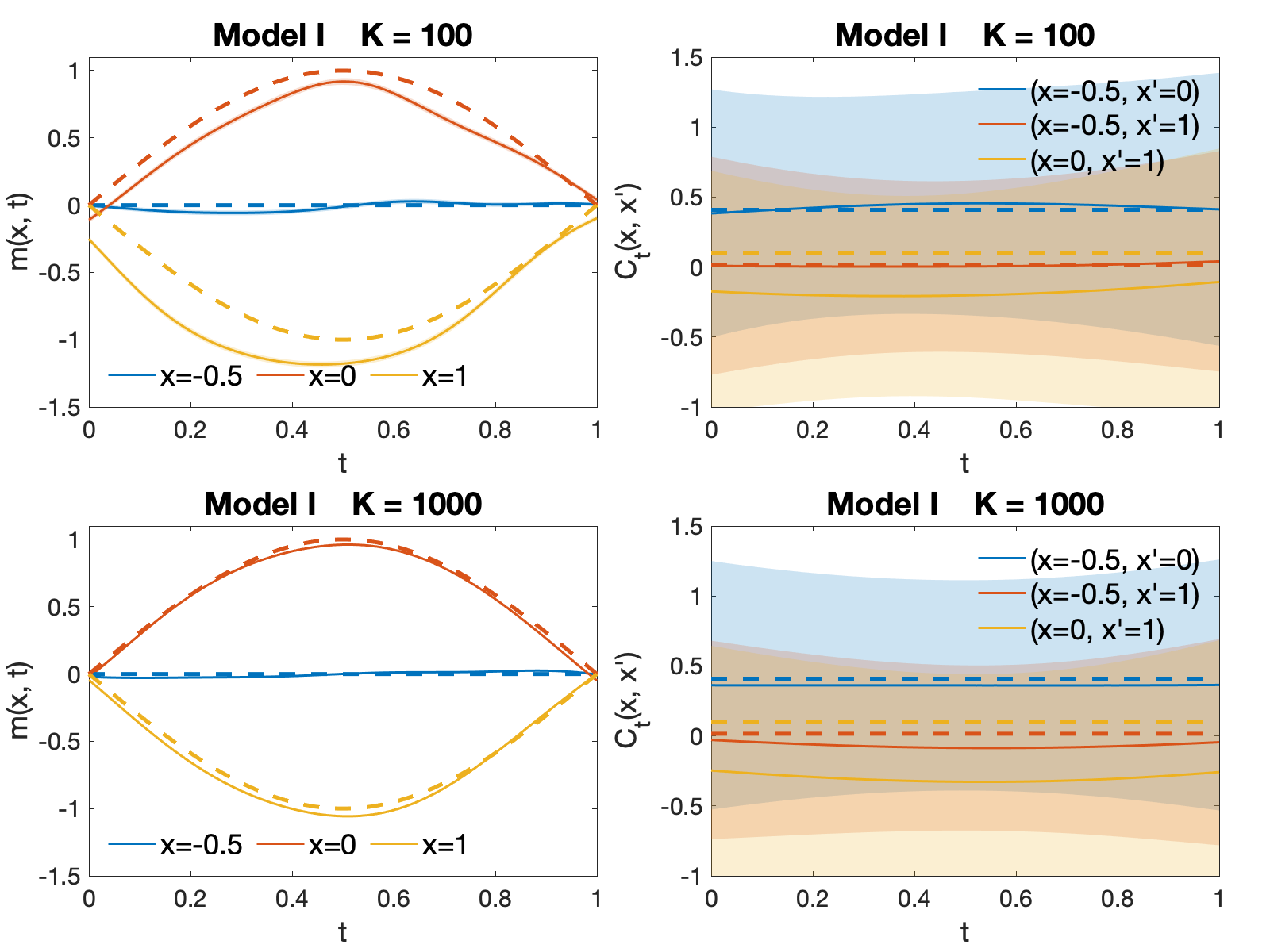}
   \caption{Model I (qKron-prod)}
   \end{subfigure}
   \begin{subfigure}[b]{.495\textwidth}
   \includegraphics[width=1\textwidth,height=.6\textwidth]{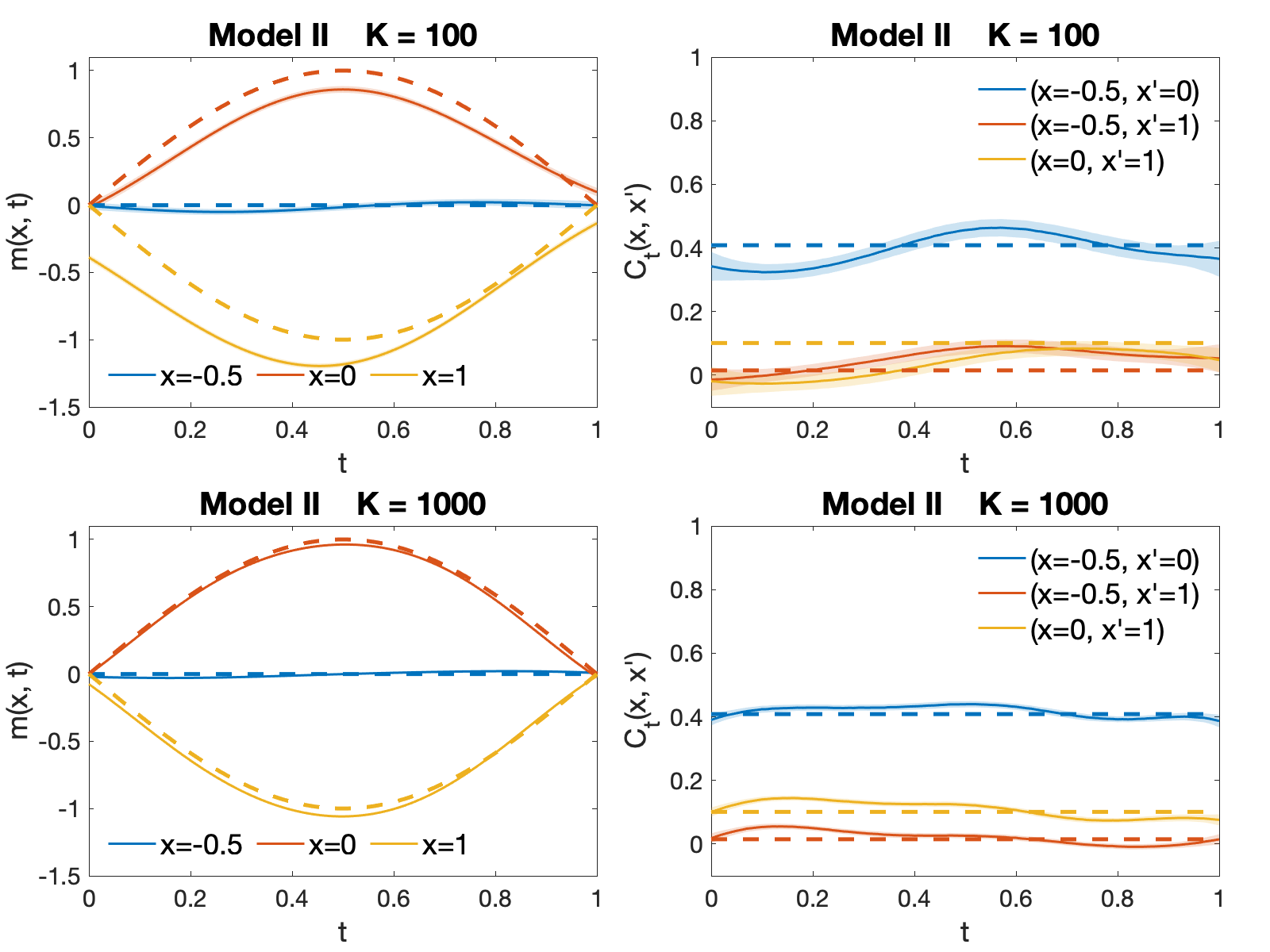}
   \caption{Model II (qKron-sum)}
   \end{subfigure}
   \begin{subfigure}[b]{.495\textwidth}
   \includegraphics[width=1\textwidth,height=.6\textwidth]{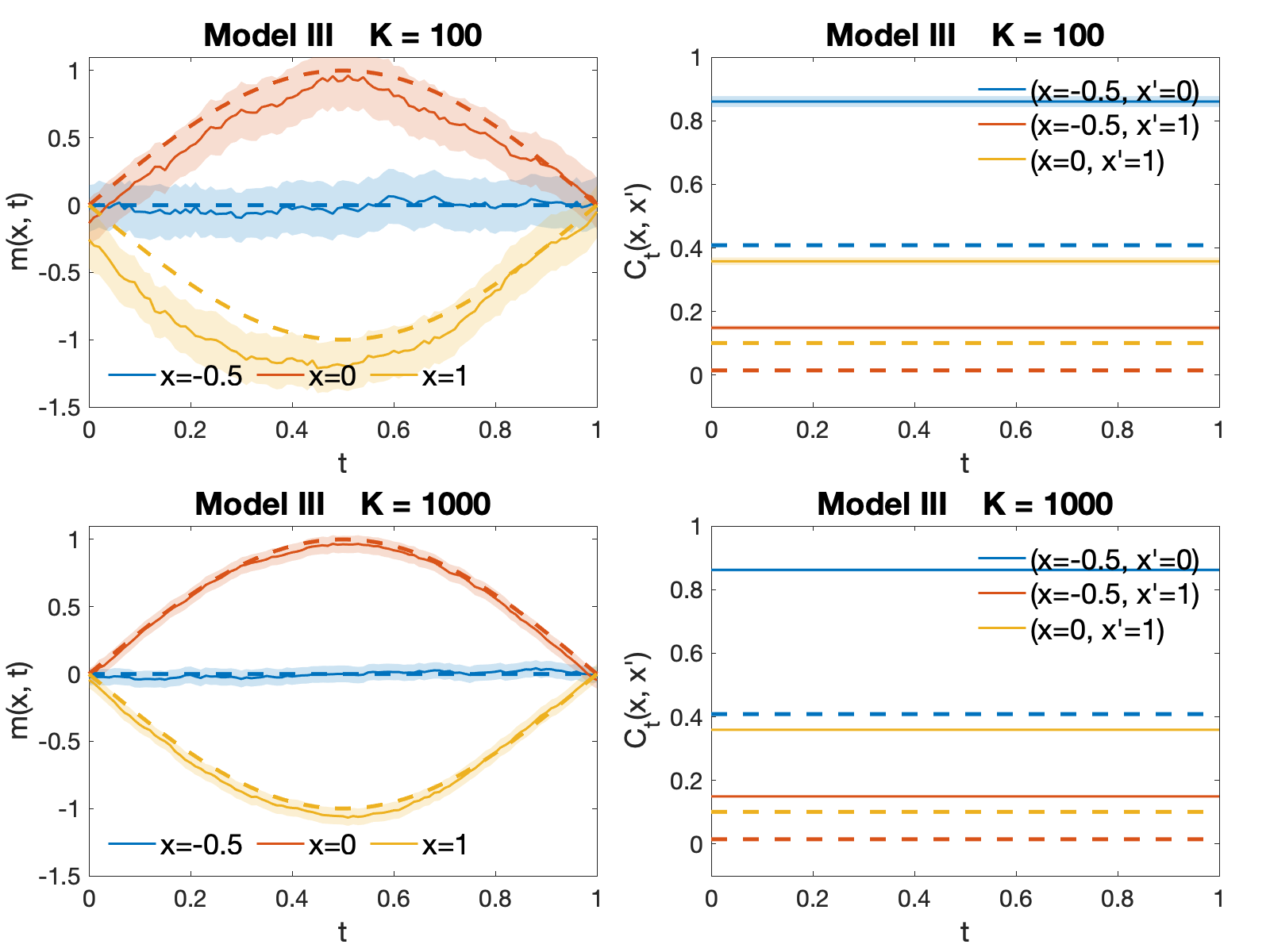}
   \caption{Model III (stat-nonsep)}
   \end{subfigure}
   \caption{Stationary Process: selective mean functions $m(x,t)$ (left column) and covariance functions $C_{y|t}^\text{s}(x,x')$ (right column) fitted by (a) model 0, (b) model I, (c) model II and (d) model III with $K=100$ trials of data (upper row) and $K=1000$ trials of data (lower row) on each panel. Dashed lines are true values, solid curves are estimates with shaded credible regions indicating their uncertainty.}
   \label{fig:fitted_STproc_stat}
\end{figure}
\section{More Numerical Results} \label{apx:more_numerics}

\subsection{Simulated Stationary Non-separable Process}\label{apx:stat-nonsep}
For the covariance $\mC_y$ in the model \eqref{eq:simSTP}, we specify 
the following stationary and non-separable kernel:
\begin{align}
\mC_y^\text{s}(\bz,\bz') &= \exp\left( -\frac{|x-x'|^2}{2\ell_x} - \frac{|t-t'|^2}{2\ell_t} - \frac{|x-x'|}{2\ell_{xt}(|t-t'|+1)}\right) \frac{1}{|t-t'|+1}+ \sigma^2_\eps\delta(\bz=\bz') \label{eq:stat_cov} \\
\end{align}
where the stationary non-separable covariance is modified from \cite{Gneiting_2002}.
We have the following true TESD constant over time:
\begin{equation}
C_{y|t}^\text{s}(x,x'):=\Cov[y(x,t), y(x',t)] = \exp\left( -\frac{|x-x'|^2}{2\ell_x} - \frac{|x-x'|}{2\ell_{xt}} \right) + \sigma^2_\eps\delta(x=x')
\end{equation}


Estimates generated by MCMC samples are plotted at selective locations in Figure \ref{fig:fitted_STproc_stat}.
All models produce estimates mean functions contracting to the truth.
However, only model II gives faithful estimates of covariance functions (TESD). 
Note, despite of the first two separable terms in \eqref{eq:stat_cov}, the stationary data are generated mostly according to stat-nonsep model \eqref{eq:stat-nonsep} with $\sigma^2=1$ and $c=\frac{1}{2\ell_{xt}}$, yet stat-nonsep model is still not flexible enough to correctly estimate constant TESD.

\subsection{Longitudinal Analysis of Brain Images} \label{apx:moreADres}

\begin{figure}[t] 
   \centering
   \includegraphics[width=1\textwidth,height=.15\textwidth]{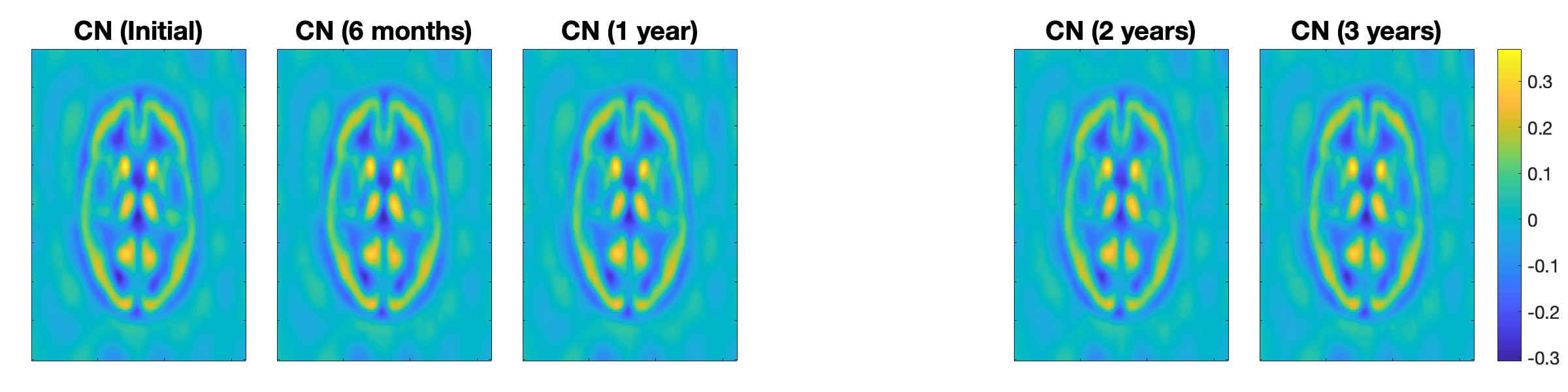} 
   \includegraphics[width=1\textwidth,height=.15\textwidth]{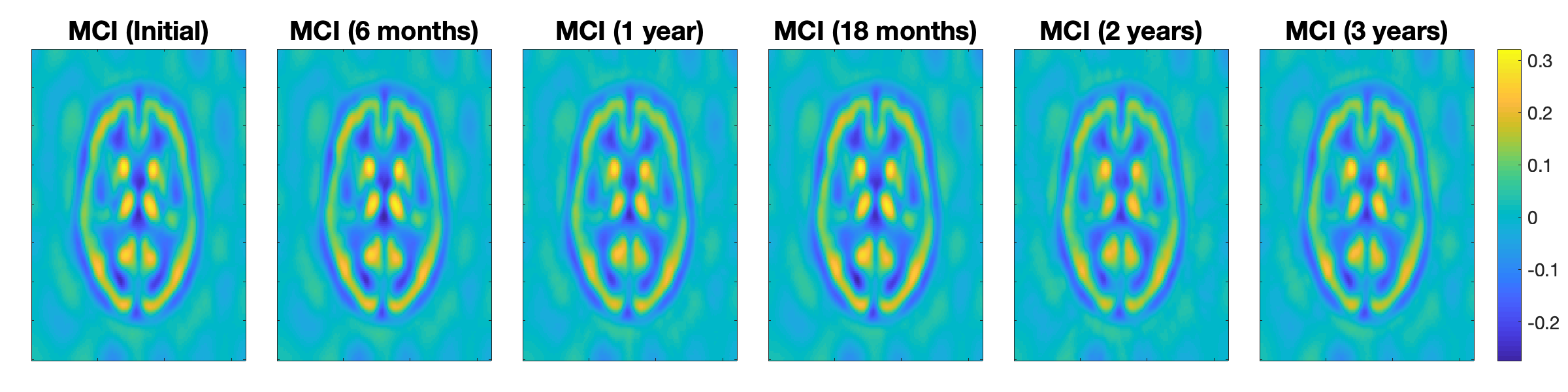} 
   \includegraphics[width=1\textwidth,height=.15\textwidth]{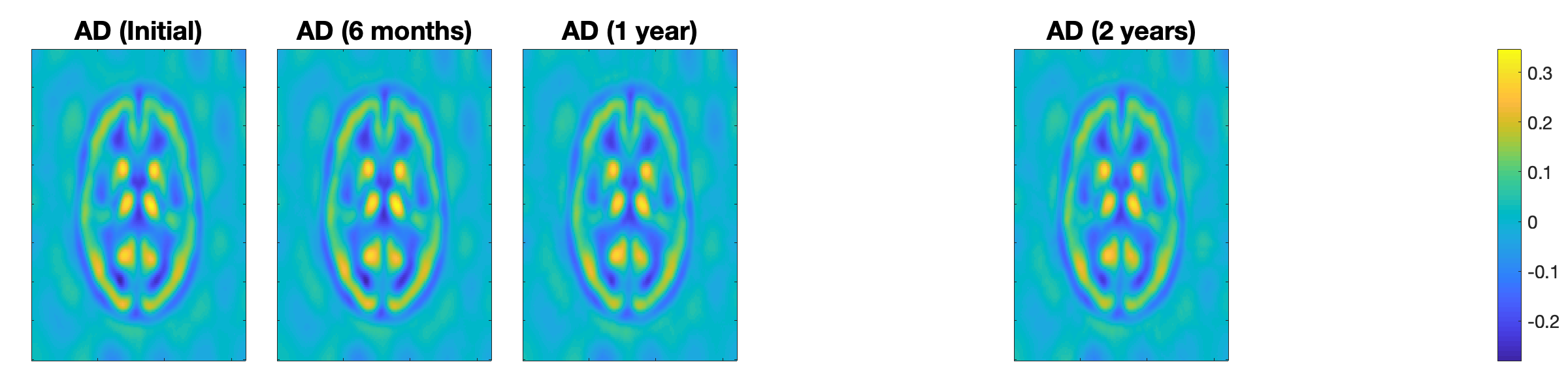} 
   \caption{Estimated brain images of ADNI patients as a function of time for CN (top row), MCI (middle row) and AD (bottom row) respectively by the proposed model II (qKron-sum).}
   \label{fig:estm_PET}
\end{figure}

\begin{figure}[t] 
   \centering
   \includegraphics[width=1\textwidth,height=.3\textwidth]{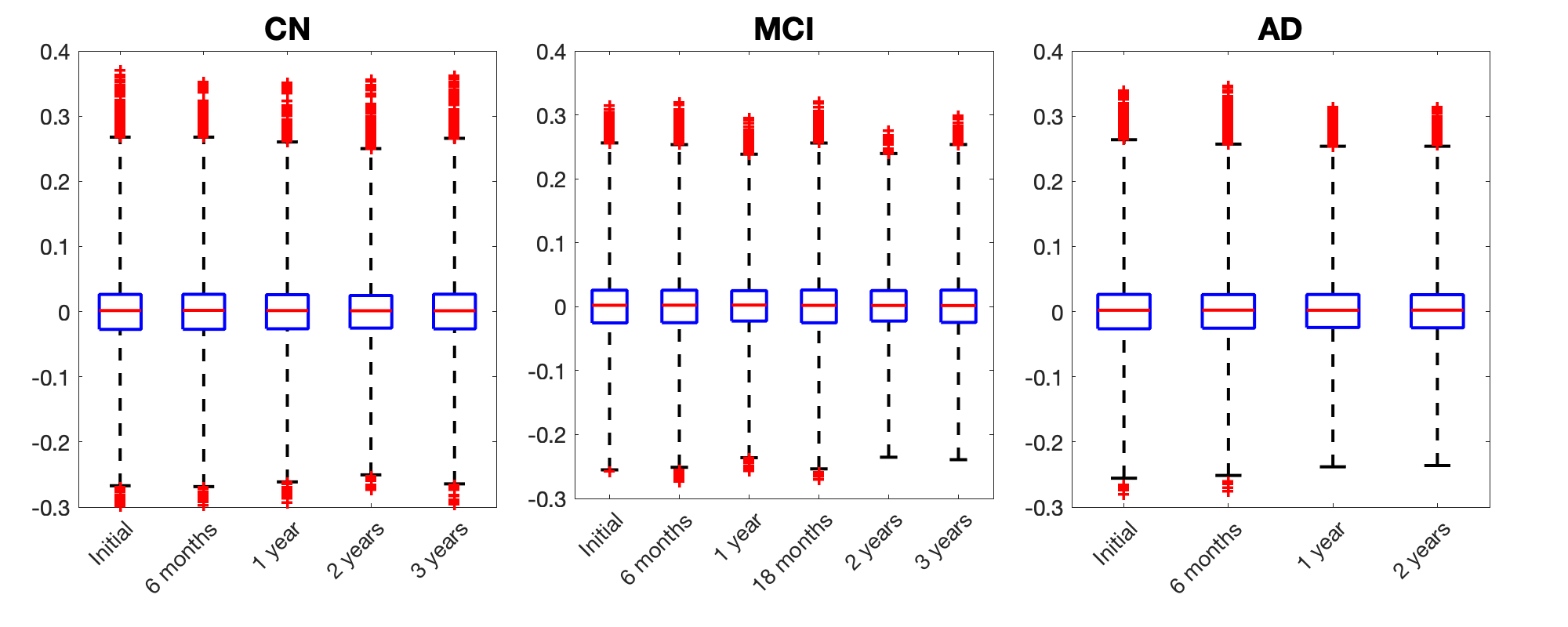}
   \caption{Pixel values of the estimated brain images for CN (left), MCI (middle) and AD (right) respectively indicate more hollow area increased with time in the latter two groups.}
   \label{fig:pixel_PET}
\end{figure}

Figure \ref{fig:pixel_PET} shows that the highest quantiles (horizontal bars) of AD patients decrease with time.
This means there are increasing `hollow' area in these brain images (especially in the MCI and AD groups) as time goes by, indicating the brain shrinkage.

\begin{figure}[t] 
   \centering
   \includegraphics[width=1\textwidth,height=.3\textwidth]{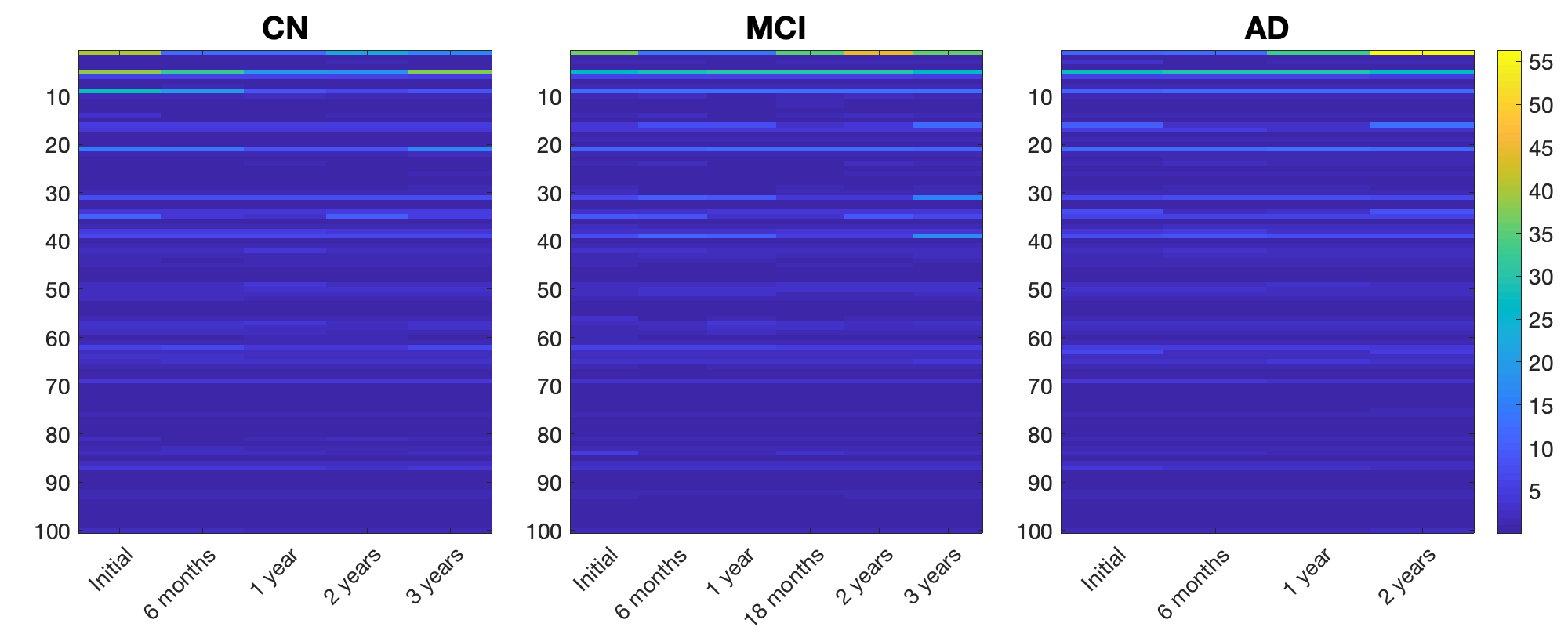}
   \caption{Dynamic eigenvalues in the generalized STGP model analyzing PET brain images for CN (left), MCI (middle) and AD (right) respectively by the proposed model II (qKron-sum).}
   \label{fig:dyneigv_PET}
\end{figure}

Figure \ref{fig:dyneigv_PET} compares the dynamic eigenvalues $\lambda_\ell^2(t)$ for different groups.
Interestingly, they do not decrease monotonically in $\ell$ (on y-axis) but rather damp out as $\ell$ becomes larger.
When $\ell$ gets close to $L=100$, the magnitude of $\lambda_\ell^2(t)$ becomes small enough to be negligible.

\begin{figure}[H] 
   \centering
   \includegraphics[width=1\textwidth,height=.125\textwidth]{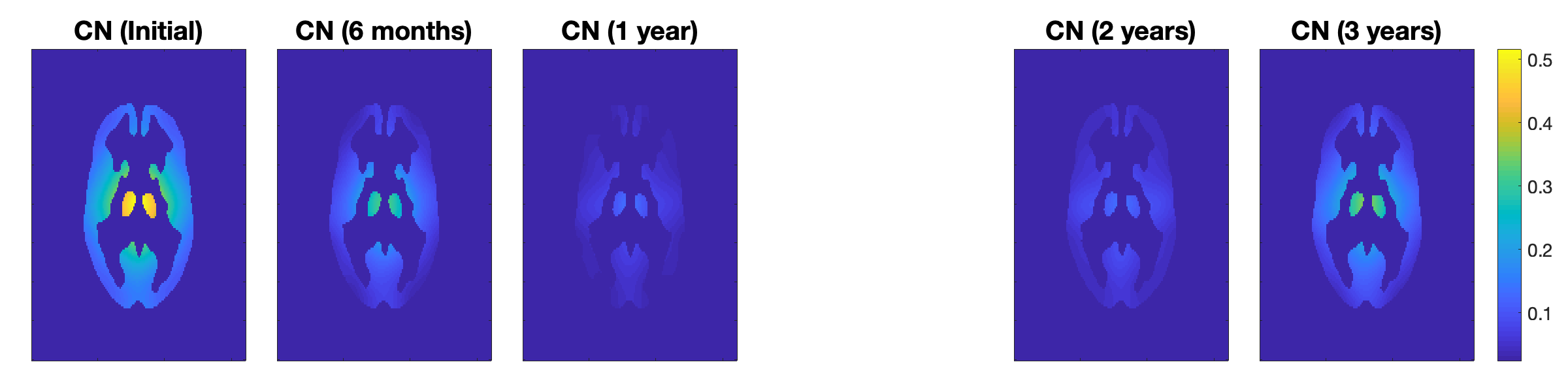} 
   \includegraphics[width=1\textwidth,height=.125\textwidth]{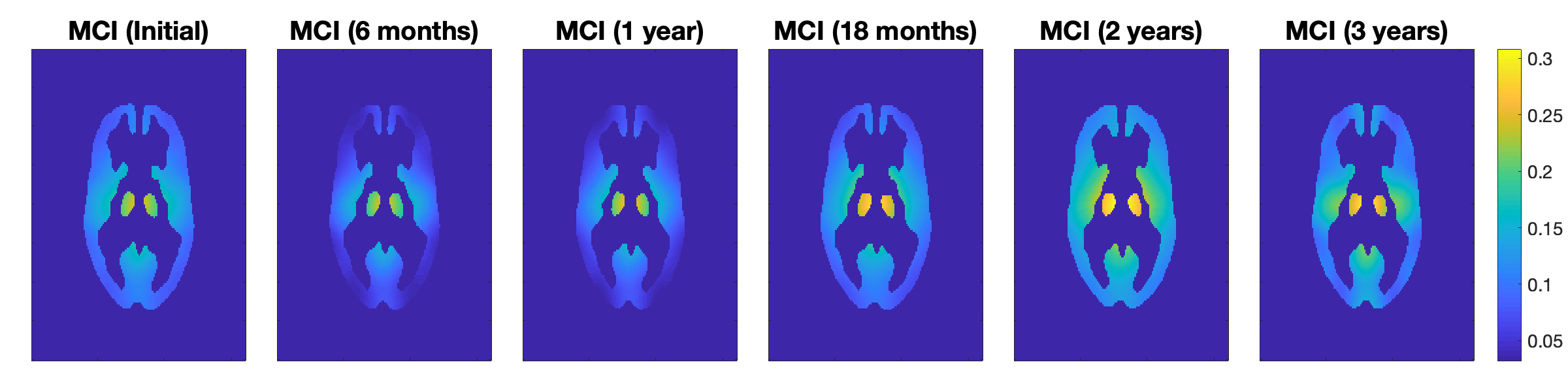} 
   \includegraphics[width=1\textwidth,height=.125\textwidth]{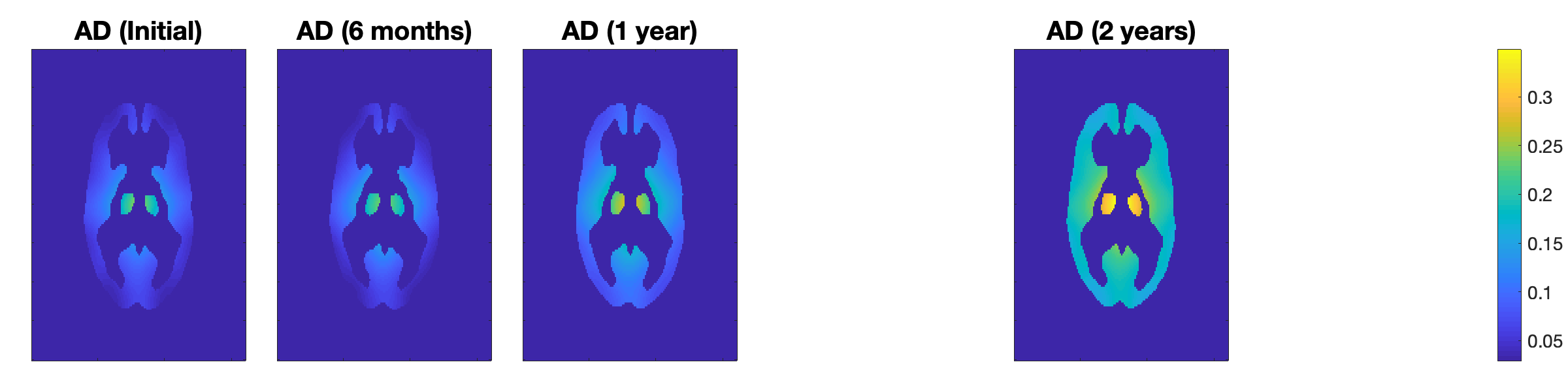} 
   \caption{Estimated variance of the brain images for CN (top row), MCI (middle row) and AD (bottom row) respectively by the proposed model II (qKron-sum).}
   \label{fig:estvar_PET}
\end{figure}

In Section \ref{sec:PET_fit} we summarize the correlation between the brain ROI and POI.
In fact, we have more results regarding TESD presented in different forms.
Figure \ref{fig:estvar_PET} shows the estimated variances of the brain images as functions of time.
They are all small across different groups with small variation along the time. Comparatively, the thalamus and some part of the temporal lobe are more active than the rest of the brain.

\begin{figure}[t] 
   \centering
   \includegraphics[width=1\textwidth,height=.15\textwidth]{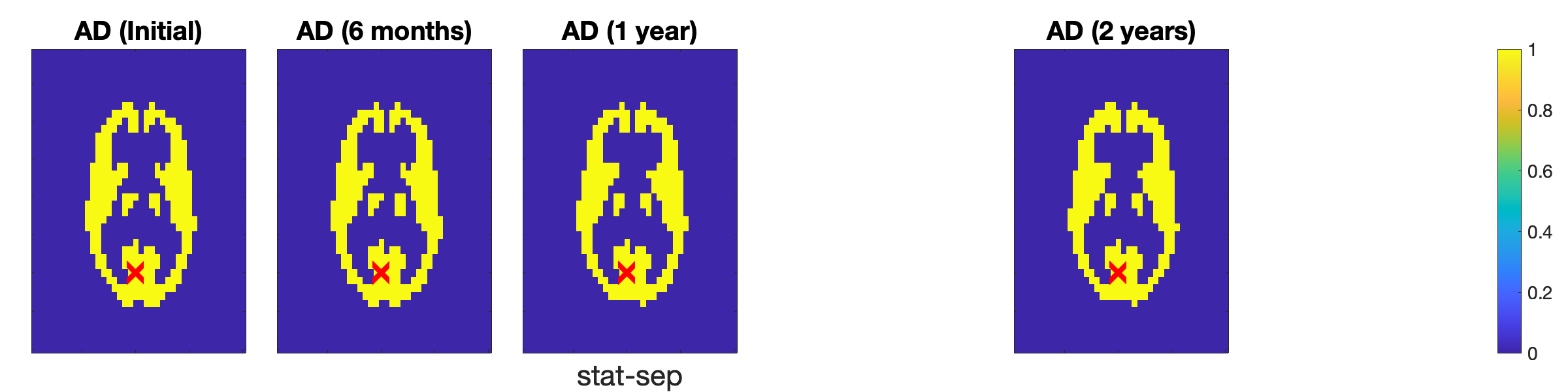}
   \includegraphics[width=1\textwidth,height=.15\textwidth]{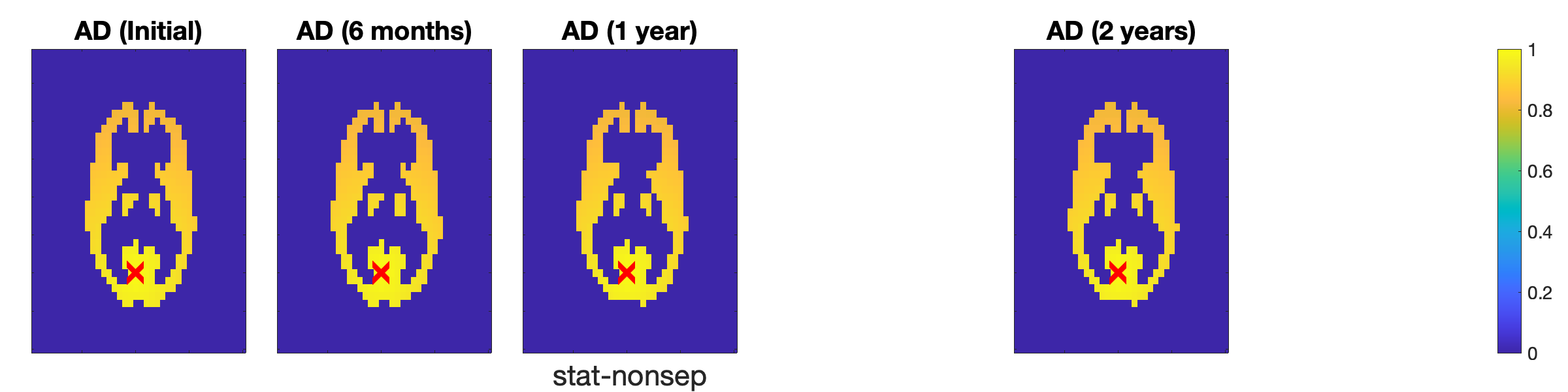}
   \includegraphics[width=1\textwidth,height=.15\textwidth]{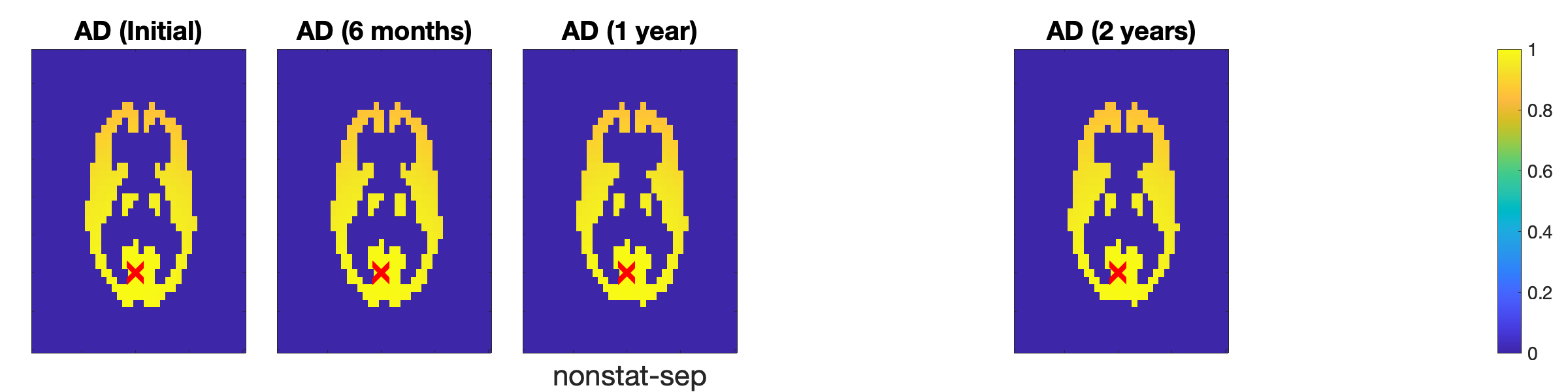}
   \includegraphics[width=1\textwidth,height=.15\textwidth]{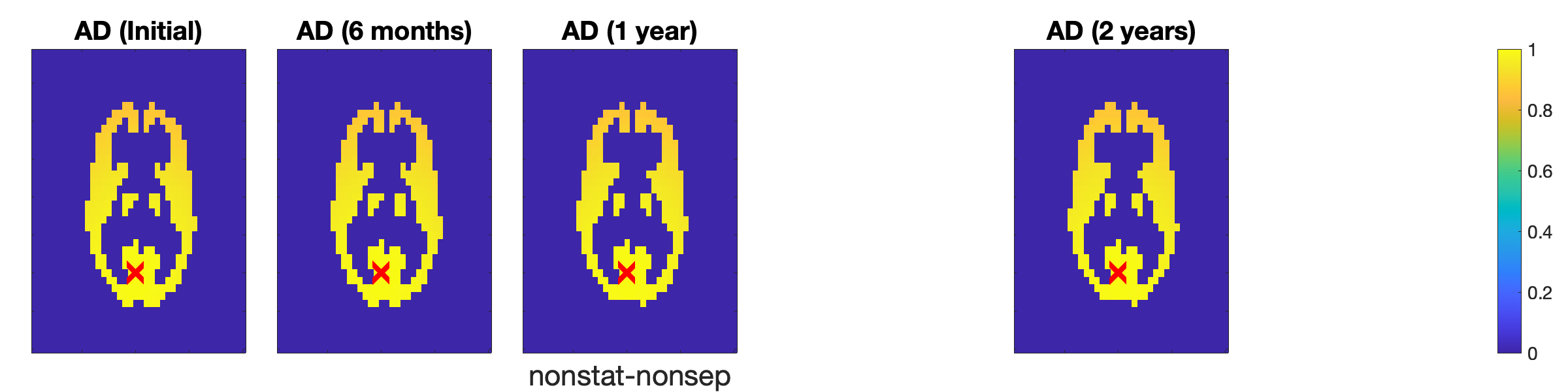}
   \includegraphics[width=1\textwidth,height=.15\textwidth]{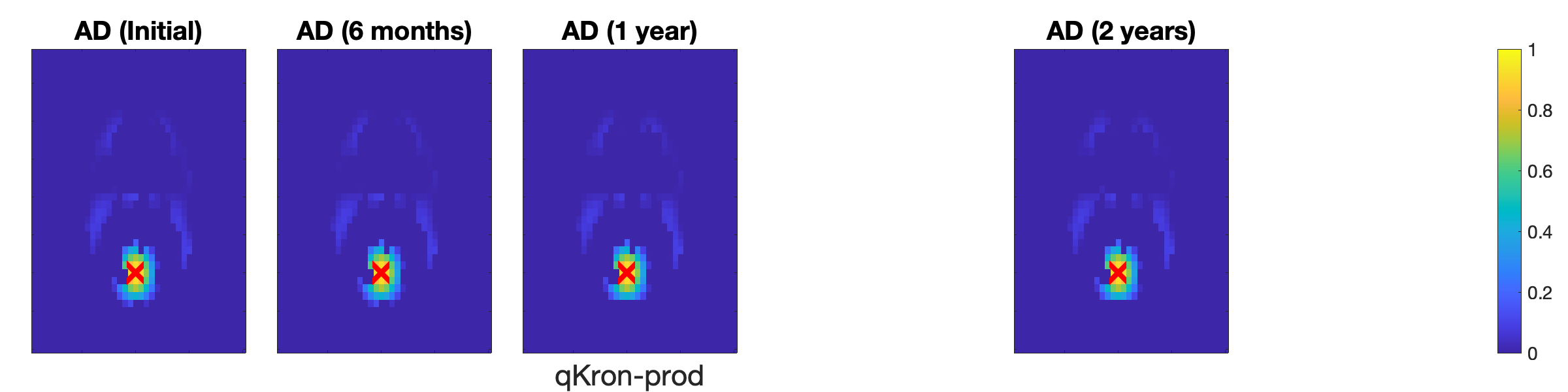}
   \includegraphics[width=1\textwidth,height=.15\textwidth]{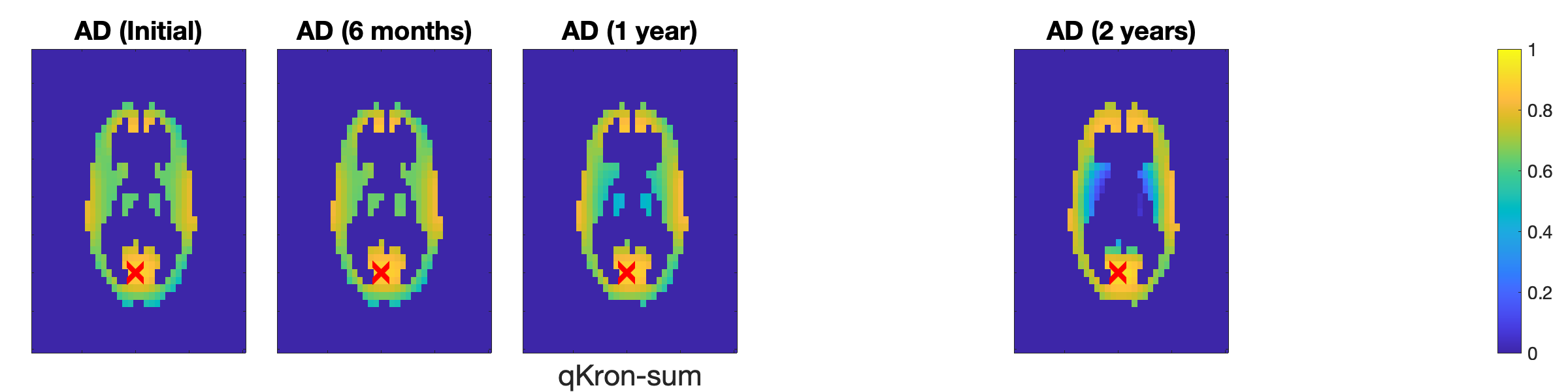} 
   \caption{Estimated correlation between the brain region of interest and a point of interest of the brain images for AD by various models.}
   \label{fig:estcorpoi_multiplemodels}
\end{figure}

In Figure \ref{fig:estcorpoi_multiplemodels}, we compare the estimated ROI-POI correlation on a $40\times 40$ mesh by various spatiotemporal models. All fail to capture the time evolution of such correlations except the proposed qKron-sum model.


\begin{figure}[t] 
   \centering
   \includegraphics[width=1\textwidth,height=.4\textwidth]{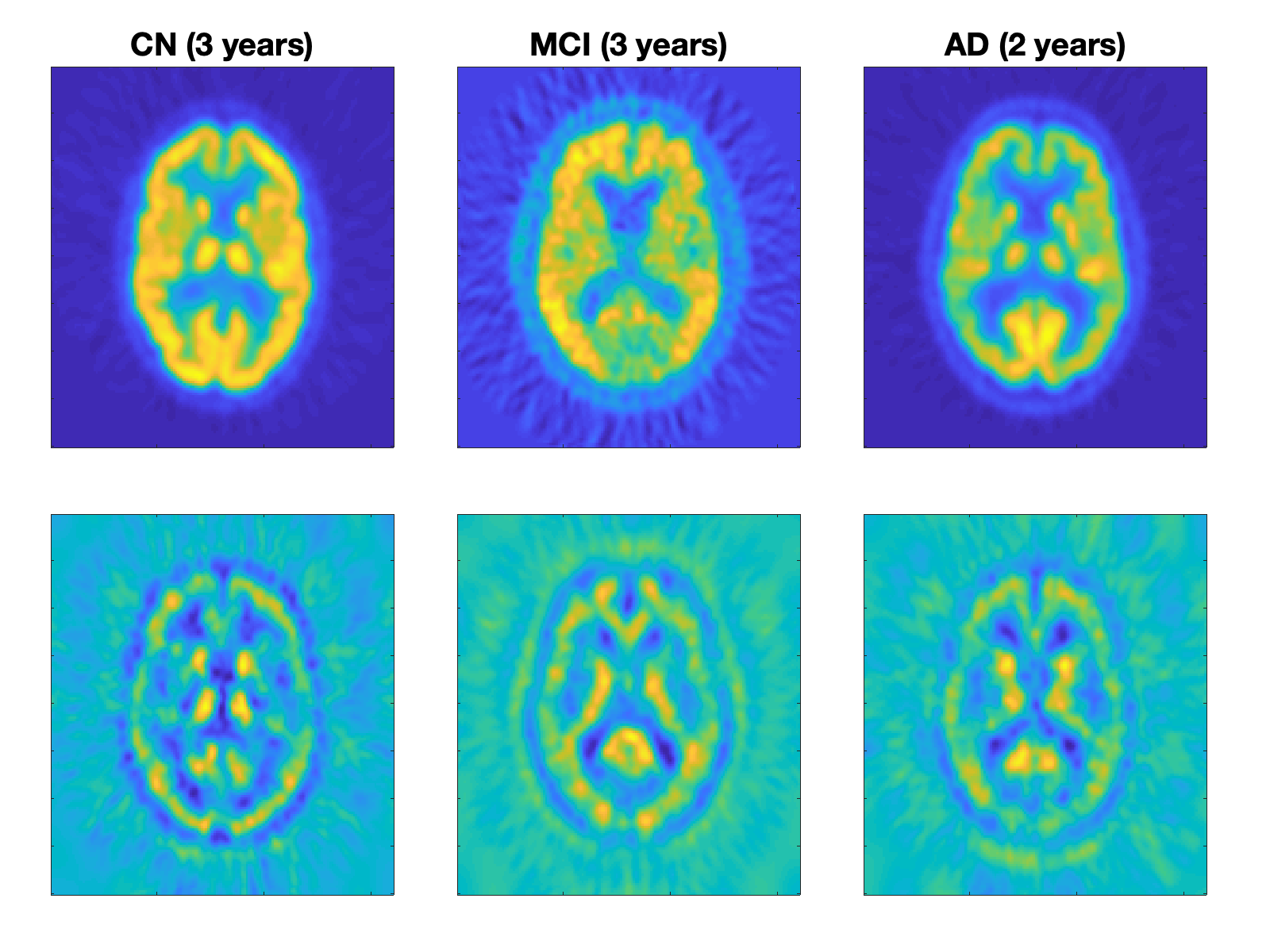}
   \caption{Prediction of the brain images at the last time point for CN (left column), MCI (middle column) and AD (right column) respectively by the proposed model II (qKron-sum). 
   The upper row shows individuals' brain images; the model outputs are displayed in the lower row.}
   \label{fig:predm_PET}
\end{figure}

\begin{figure}[t] 
   \centering
   \includegraphics[width=1\textwidth,height=.2\textwidth]{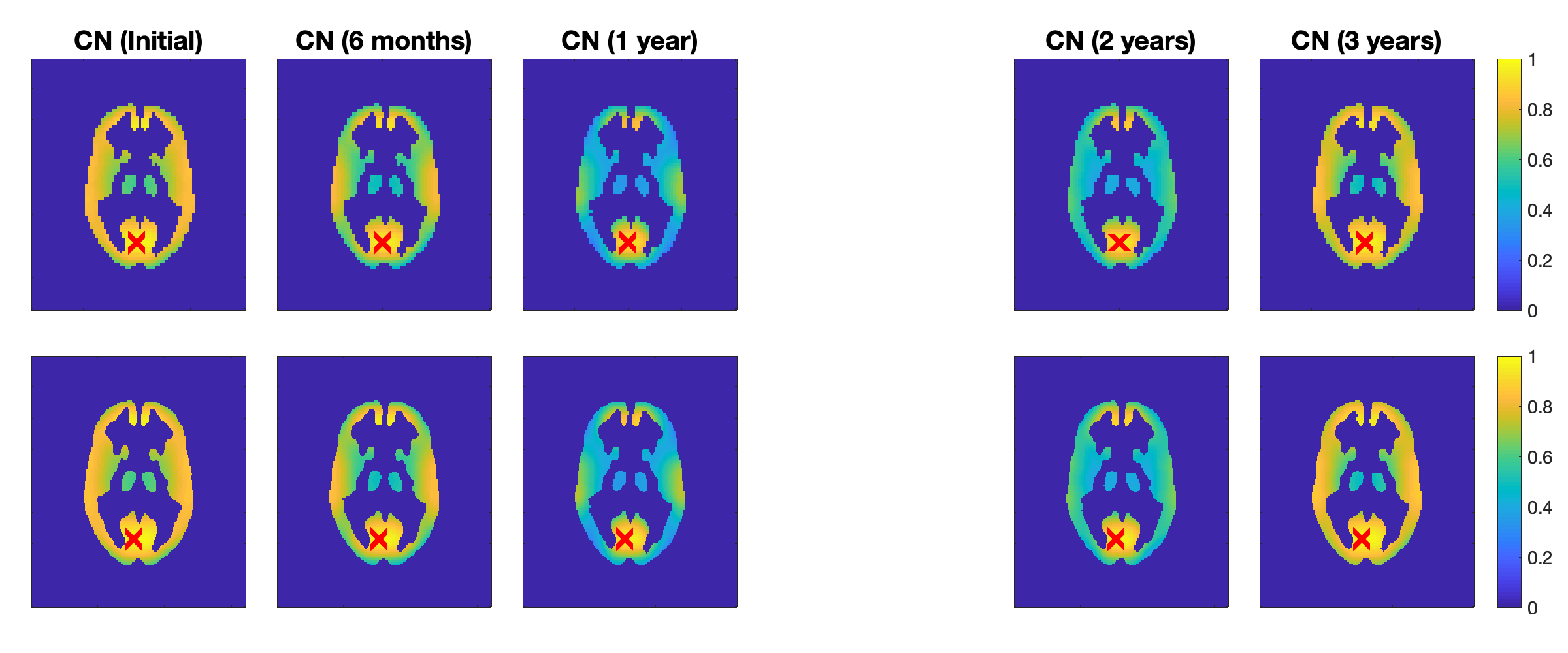} 
   \includegraphics[width=1\textwidth,height=.2\textwidth]{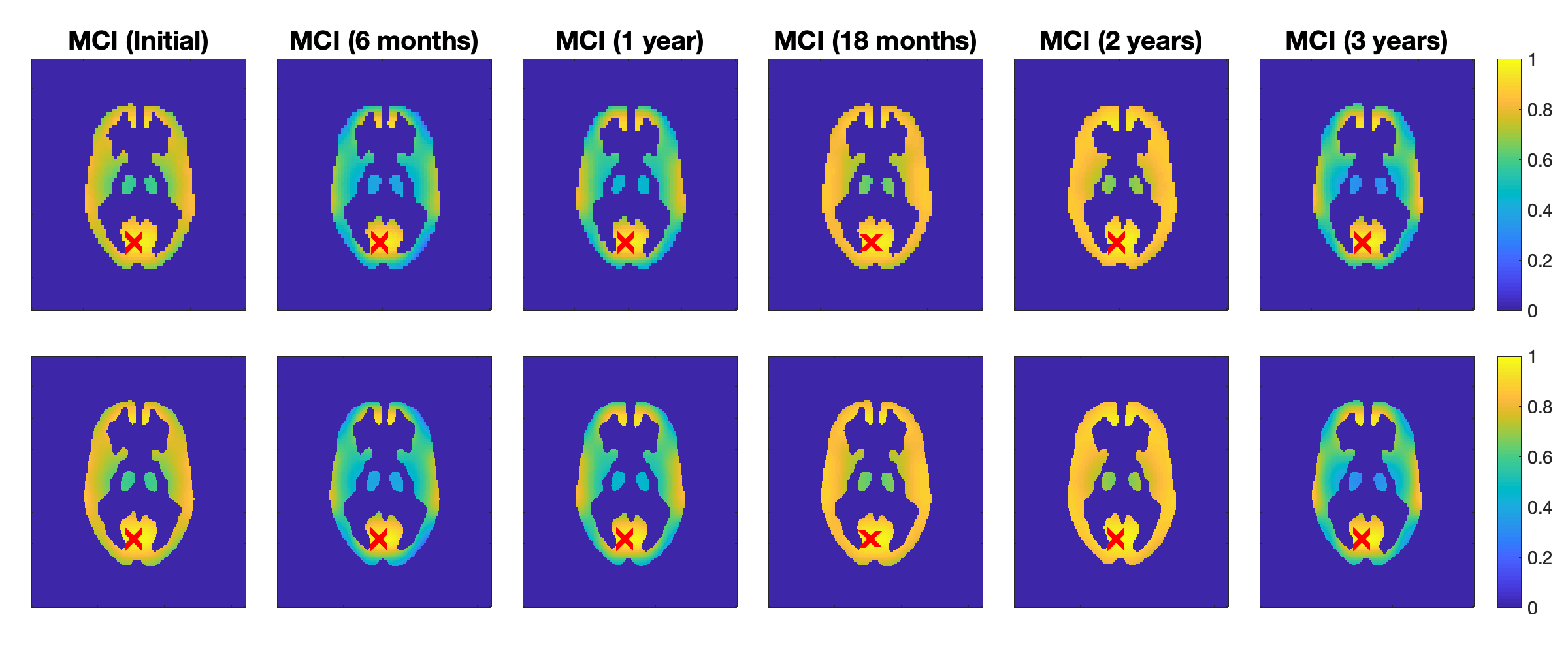} 
   \includegraphics[width=1\textwidth,height=.2\textwidth]{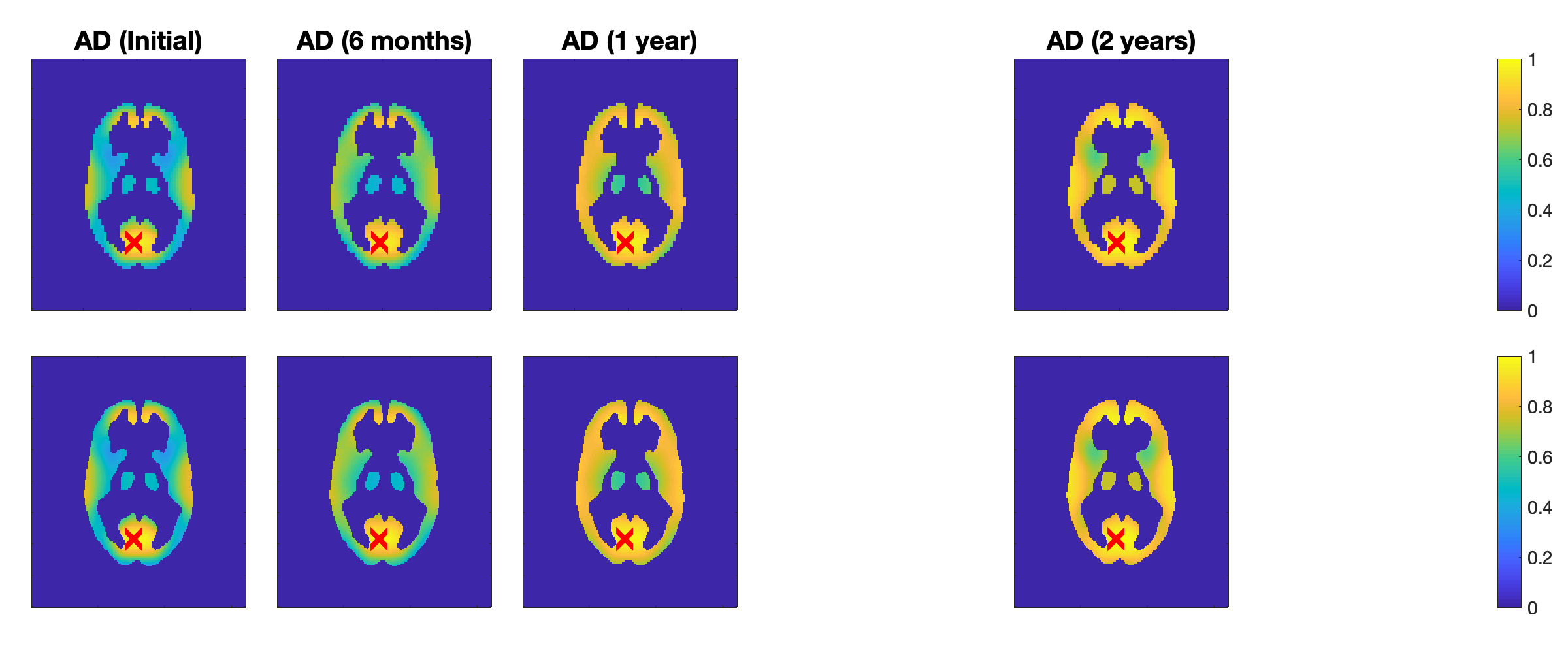} 
   \caption{Extended correlation between the brain region of interest and a point of interest from the coarse mesh (upper) to the fine mesh (lower) in each of CN, MCI and AD groups by the proposed model II (qKron-sum).}
   \label{fig:extdcorpoi1_PET}
\end{figure}

Finally, we consider the problem of extending TESD to new locations, infeasible in dynamic covariance models.
We coarsen the mesh by using every other pixel and build our model based on the resulted $80\times 80$ images.
Figure \ref{fig:extdcorpoi1_PET} compares the estimated ROI-POI correlation on the coarse mesh (upper row) and the prediction to the original $160\times 160$ mesh (lower row) which is consistent with the estimation result in Figure \ref{fig:corpoi_PET}. 
Such extension provides more fine details of TESD at new locations without data.
They all illustrate the benefit of a fully nonparametric approach in modeling TESD in the spatiotemporal data.

\clearpage
\bibliography{references}

\begin{thebibliography}{59}
\providecommand{\natexlab}[1]{#1}
\providecommand{\url}[1]{\texttt{#1}}
\expandafter\ifx\csname urlstyle\endcsname\relax
  \providecommand{\doi}[1]{doi: #1}\else
  \providecommand{\doi}{doi: \begingroup \urlstyle{rm}\Url}\fi

\bibitem[ADNI(2003)]{ADNI}
ADNI.
\newblock The alzheimer's disease neuroimaging initiative.
\newblock http://adni.loni.usc.edu, 2003.

\bibitem[Banerjee(2015)]{Banerjee_2015}
Sudipto~P Banerjee.
\newblock \emph{Hierarchical modeling and analysis for spatial data}.
\newblock Chapman \& Hall/CRC, Boca Raton (Fla.); London; New York [etc], 2nd
  edition. edition, 2015.
\newblock ISBN 9781439819173.

\bibitem[Chung and Graham(1997)]{chung1997}
F.R.K. Chung and F.C. Graham.
\newblock \emph{Spectral Graph Theory}.
\newblock Number no. 92 in CBMS Regional Conference Series. American
  Mathematical Society, 1997.
\newblock ISBN 9780821803158.
\newblock URL \url{https://books.google.com/books?id=4IK8DgAAQBAJ}.

\bibitem[Cressie and Wikle(2011)]{cressie2011}
N.~Cressie and C.K. Wikle.
\newblock \emph{Statistics for Spatio-Temporal Data}.
\newblock CourseSmart Series. Wiley, 2011.
\newblock ISBN 9780471692744.
\newblock URL \url{https://books.google.com/books?id=-kOC6D0DiNYC}.

\bibitem[Cressie and Huang(1999)]{Cressie_1999}
Noel Cressie and Hsin-Cheng Huang.
\newblock Classes of nonseparable, spatio-temporal stationary covariance
  functions.
\newblock \emph{Journal of the American Statistical Association}, 94\penalty0
  (448):\penalty0 1330--1339, 1999.
\newblock \doi{10.1080/01621459.1999.10473885}.
\newblock URL
  \url{https://www.tandfonline.com/doi/abs/10.1080/01621459.1999.10473885}.

\bibitem[Cribben et~al.(2012)Cribben, Haraldsdottir, Atlas, Wager, and
  Lindquist]{cribben12}
Ivor Cribben, Ragnheidur Haraldsdottir, Lauren~Y. Atlas, Tor~D. Wager, and
  Martin~A. Lindquist.
\newblock Dynamic connectivity regression: Determining state-related changes in
  brain connectivity.
\newblock \emph{NeuroImage}, 61\penalty0 (4):\penalty0 907 -- 920, 2012.
\newblock ISSN 1053-8119.
\newblock \doi{https://doi.org/10.1016/j.neuroimage.2012.03.070}.
\newblock URL
  \url{http://www.sciencedirect.com/science/article/pii/S1053811912003515}.

\bibitem[Damianou and Lawrence(2013)]{Damianou_2013}
Andreas Damianou and Neil~D. Lawrence.
\newblock Deep {G}aussian processes.
\newblock In Carlos~M. Carvalho and Pradeep Ravikumar, editors,
  \emph{Proceedings of the Sixteenth International Conference on Artificial
  Intelligence and Statistics}, volume~31 of \emph{Proceedings of Machine
  Learning Research}, pages 207--215, Scottsdale, Arizona, USA, 29 Apr--01 May
  2013. PMLR.
\newblock URL \url{https://proceedings.mlr.press/v31/damianou13a.html}.

\bibitem[Das and Bhattacharya(2020)]{das2020}
Moumita Das and Sourabh Bhattacharya.
\newblock Nonstationary, nonparametric, nonseparable bayesian spatio-temporal
  modeling using kernel convolution of order based dependent dirichlet process,
  2020.

\bibitem[Dashti and Stuart(2017)]{Dashti_2017}
Masoumeh Dashti and Andrew~M. Stuart.
\newblock The bayesian approach to inverse problems.
\newblock \emph{Handbook of Uncertainty Quantification}, pages 311--428, 2017.
\newblock \doi{10.1007/978-3-319-12385-1_7}.
\newblock URL \url{http://dx.doi.org/10.1007/978-3-319-12385-1_7}.

\bibitem[Datta et~al.(2016)Datta, Banerjee, Finley, Hamm, and
  Schaap]{datta2016}
Abhirup Datta, Sudipto Banerjee, Andrew~O. Finley, Nicholas A.~S. Hamm, and
  Martijn Schaap.
\newblock Nonseparable dynamic nearest neighbor gaussian process models for
  large spatio-temporal data with an application to particulate matter
  analysis.
\newblock \emph{The Annals of Applied Statistics}, 10\penalty0 (3):\penalty0
  1286--1316, Sep 2016.
\newblock ISSN 1932-6157.
\newblock \doi{10.1214/16-aoas931}.
\newblock URL \url{http://dx.doi.org/10.1214/16-AOAS931}.

\bibitem[Dong et~al.(2006)Dong, Fang, Bock, Webb, Prawirodirdjo, Kedar, and
  Jamason]{Dong_2006}
D.~Dong, P.~Fang, Y.~Bock, F.~Webb, L.~Prawirodirdjo, S.~Kedar, and P.~Jamason.
\newblock Spatiotemporal filtering using principal component analysis and
  karhunen-loeve expansion approaches for regional gps network analysis.
\newblock \emph{Journal of Geophysical Research: Solid Earth}, 111\penalty0
  (B3), 2006.
\newblock \doi{https://doi.org/10.1029/2005JB003806}.
\newblock URL
  \url{https://agupubs.onlinelibrary.wiley.com/doi/abs/10.1029/2005JB003806}.

\bibitem[Dunlop et~al.(2018)Dunlop, Girolami, Stuart, and
  Teckentrup]{Dunlop_2018}
Matthew~M. Dunlop, Mark~A. Girolami, Andrew~M. Stuart, and Aretha~L.
  Teckentrup.
\newblock How deep are deep gaussian processes?
\newblock \emph{J. Mach. Learn. Res.}, 19\penalty0 (1):\penalty0 2100--2145,
  jan 2018.
\newblock ISSN 1532-4435.

\bibitem[Dunlop et~al.(2020)Dunlop, Slep{\v c}ev, Stuart, and
  Thorpe]{DUNLOP2020}
Matthew~M. Dunlop, Dejan Slep{\v c}ev, Andrew~M. Stuart, and Matthew Thorpe.
\newblock Large data and zero noise limits of graph-based semi-supervised
  learning algorithms.
\newblock \emph{Applied and Computational Harmonic Analysis}, 49\penalty0
  (2):\penalty0 655--697, 2020.
\newblock ISSN 1063-5203.
\newblock \doi{https://doi.org/10.1016/j.acha.2019.03.005}.
\newblock URL
  \url{https://www.sciencedirect.com/science/article/pii/S1063520318301398}.

\bibitem[Fiecas and Ombao(2016)]{fiecas16}
Mark Fiecas and Hernando Ombao.
\newblock Modeling the evolution of dynamic brain processes during an
  associative learning experiment.
\newblock \emph{Journal of the American Statistical Association}, 111\penalty0
  (516):\penalty0 1440--1453, 2016.
\newblock \doi{10.1080/01621459.2016.1165683}.
\newblock URL \url{http://dx.doi.org/10.1080/01621459.2016.1165683}.

\bibitem[Fonseca and Steel(2011)]{Fonseca_2011}
Tha{\'{\i}}s C.~O. Fonseca and Mark F.~J. Steel.
\newblock A general class of nonseparable space-time covariance models.
\newblock \emph{Environmetrics}, 22\penalty0 (2):\penalty0 224--242, mar 2011.
\newblock \doi{10.1002/env.1047}.
\newblock URL \url{https://doi.org/10.1002%2Fenv.1047}.

\bibitem[Fontanella and Ippoliti(2003)]{Fontanella_2003}
Lara Fontanella and Luigi Ippoliti.
\newblock Dynamic models for space-time prediction via karhunen-lo{\'{e}}ve
  expansion.
\newblock \emph{Statistical Methods {\&} Applications}, 12\penalty0
  (1):\penalty0 61--78, feb 2003.
\newblock \doi{10.1007/bf02511584}.

\bibitem[Fox and Dunson(2015)]{fox15}
Emily~B Fox and David~B Dunson.
\newblock Bayesian nonparametric covariance regression.
\newblock \emph{Journal of Machine Learning Research}, 16:\penalty0 2501--2542,
  2015.

\bibitem[Fuentes et~al.(2008)Fuentes, Chen, and Davis]{Fuentes_2008}
Montserrat Fuentes, Li~Chen, and Jerry~M. Davis.
\newblock A class of nonseparable and nonstationary spatial temporal covariance
  functions.
\newblock \emph{Environmetrics}, 19\penalty0 (5):\penalty0 487--507, 2008.
\newblock \doi{https://doi.org/10.1002/env.891}.
\newblock URL \url{https://onlinelibrary.wiley.com/doi/abs/10.1002/env.891}.

\bibitem[Fukunaga(1990)]{FUKUNAGA1990}
Keinosuke Fukunaga.
\newblock \emph{Introduction to Statistical Pattern Recognition}.
\newblock Academic Press, Boston, second edition edition, 1990.
\newblock ISBN 978-0-08-047865-4.
\newblock \doi{https://doi.org/10.1016/B978-0-08-047865-4.50007-7}.
\newblock URL
  \url{https://www.sciencedirect.com/science/article/pii/B9780080478654500077}.

\bibitem[Gelfand et~al.(2005)Gelfand, Kottas, and MacEachern]{Gelfand_2005}
Alan~E Gelfand, Athanasios Kottas, and Steven~N MacEachern.
\newblock Bayesian nonparametric spatial modeling with dirichlet process
  mixing.
\newblock \emph{Journal of the American Statistical Association}, 100\penalty0
  (471):\penalty0 1021--1035, sep 2005.
\newblock \doi{10.1198/016214504000002078}.

\bibitem[Ghosal and {van der Vaart}(2007)]{ghosal2007}
S.~Ghosal and A.W. {van der Vaart}.
\newblock Convergence rates of posterior distributions for non-i.i.d.
  observations.
\newblock \emph{Annals of Statistics}, 35\penalty0 (1):\penalty0 192--223,
  2007.
\newblock ISSN 0090-5364.
\newblock \doi{10.1214/009053606000001172}.
\newblock MR2332274.

\bibitem[Ghosal and van~der Vaart(2017)]{Ghosal_2017}
Subhashis Ghosal and Aad van~der Vaart.
\newblock Fundamentals of nonparametric bayesian inference.
\newblock 2017.
\newblock \doi{10.1017/9781139029834}.
\newblock URL \url{http://dx.doi.org/10.1017/9781139029834}.

\bibitem[Gneiting(2002)]{Gneiting_2002}
Tilmann Gneiting.
\newblock Nonseparable, stationary covariance functions for space--time data.
\newblock \emph{Journal of the American Statistical Association}, 97\penalty0
  (458):\penalty0 590--600, 2002.
\newblock \doi{10.1198/016214502760047113}.
\newblock URL \url{https://doi.org/10.1198/016214502760047113}.

\bibitem[Hairer(2009)]{Hairer2009}
Martin Hairer.
\newblock An introduction to stochastic pdes.
\newblock \emph{arXiv:0907.4178}, 07 2009.

\bibitem[Hartikainen et~al.(2011)Hartikainen, Riihim{\"a}ki, and
  S{\"a}rkk{\"a}]{hartikainen2011}
Jouni Hartikainen, Jaakko Riihim{\"a}ki, and Simo S{\"a}rkk{\"a}.
\newblock Sparse spatio-temporal gaussian processes with general likelihoods.
\newblock In Timo Honkela, W{\l}odzis{\l}aw Duch, Mark Girolami, and Samuel
  Kaski, editors, \emph{Artificial Neural Networks and Machine Learning --
  ICANN 2011}, pages 193--200, Berlin, Heidelberg, 2011. Springer Berlin
  Heidelberg.
\newblock ISBN 978-3-642-21735-7.

\bibitem[Hu et~al.(2015)Hu, Cheng, Sepulcre, Johnson, Fakhri, Lu, and
  Li]{hu2015}
Chenhui Hu, Lin Cheng, Jorge Sepulcre, Keith~A. Johnson, Georges~E. Fakhri,
  Yue~M. Lu, and Quanzheng Li.
\newblock A spectral graph regression model for learning brain connectivity of
  alzheimer's disease.
\newblock \emph{PLOS ONE}, 10\penalty0 (5):\penalty0 e0128136, May 2015.
\newblock ISSN 1932-6203.
\newblock \doi{10.1371/journal.pone.0128136}.
\newblock URL \url{http://dx.doi.org/10.1371/journal.pone.0128136}.

\bibitem[Huang et~al.(2018)Huang, Bolton, Medaglia, Bassett, Ribeiro, and
  Ville]{huang2018}
W.~Huang, T.~A.~W. Bolton, J.~D. Medaglia, D.~S. Bassett, A.~Ribeiro, and
  D.~Van~De Ville.
\newblock A graph signal processing perspective on functional brain imaging.
\newblock \emph{Proceedings of the IEEE}, 106\penalty0 (5):\penalty0 868--885,
  May 2018.
\newblock ISSN 0018-9219.
\newblock \doi{10.1109/JPROC.2018.2798928}.

\bibitem[Hyun et~al.(2016)Hyun, Li, Huang, Styner, Lin, and Zhu]{hyun2016}
Jung~Won Hyun, Yimei Li, Chao Huang, Martin Styner, Weili Lin, and Hongtu Zhu.
\newblock Stgp: Spatio-temporal gaussian process models for longitudinal
  neuroimaging data.
\newblock \emph{NeuroImage}, 134:\penalty0 550--562, Jul 2016.
\newblock ISSN 1053-8119.
\newblock \doi{10.1016/j.neuroimage.2016.04.023}.
\newblock URL \url{http://dx.doi.org/10.1016/j.neuroimage.2016.04.023}.

\bibitem[Kuzin et~al.(2018)Kuzin, Isupova, and Mihaylova]{kuzin2018}
Danil Kuzin, Olga Isupova, and Lyudmila~S. Mihaylova.
\newblock Spatio-temporal structured sparse regression with hierarchical
  gaussian process priors.
\newblock \emph{IEEE Transactions on Signal Processing}, 66:\penalty0
  4598--4611, 2018.

\bibitem[Lan et~al.(2020)Lan, Holbrook, Elias, Fortin, Ombao, and
  Shahbaba]{lan_2019}
Shiwei Lan, Andrew Holbrook, Gabriel~A. Elias, Norbert~J. Fortin, Hernando
  Ombao, and Babak Shahbaba.
\newblock {Flexible Bayesian Dynamic Modeling of Correlation and Covariance
  Matrices}.
\newblock \emph{Bayesian Analysis}, 15\penalty0 (4):\penalty0 1199 -- 1228,
  2020.
\newblock \doi{10.1214/19-BA1173}.
\newblock URL \url{https://doi.org/10.1214/19-BA1173}.

\bibitem[LeCam(1973)]{LeCam_1973}
L.~LeCam.
\newblock Convergence of estimates under dimensionality restrictions.
\newblock \emph{The Annals of Statistics}, 1\penalty0 (1):\penalty0 38--53, Jan
  1973.
\newblock ISSN 0090-5364.
\newblock \doi{10.1214/aos/1193342380}.
\newblock URL \url{http://dx.doi.org/10.1214/aos/1193342380}.

\bibitem[LeCam(1975)]{LeCam_1975}
L.~LeCam.
\newblock On local and global properties in the theory of asymptotic normality
  of experiments.
\newblock \emph{Stochastic Processes and Related Topics}, 1:\penalty0 13--54,
  1975.

\bibitem[Luttinen and Ilin(2012)]{luttinen2012}
Jaakko Luttinen and Alexander Ilin.
\newblock Efficient gaussian process inference for short-scale spatio-temporal
  modeling.
\newblock In Neil~D. Lawrence and Mark Girolami, editors, \emph{Proceedings of
  the Fifteenth International Conference on Artificial Intelligence and
  Statistics}, volume~22 of \emph{Proceedings of Machine Learning Research},
  pages 741--750, La Palma, Canary Islands, 21--23 Apr 2012. PMLR.
\newblock URL \url{http://proceedings.mlr.press/v22/luttinen12.html}.

\bibitem[Marco et~al.(2015)Marco, Ziegler, Alexander, and Ourselin]{marco2015}
Lorenzi Marco, Gabriel Ziegler, Daniel~C. Alexander, and Sebastien Ourselin.
\newblock Modelling non-stationary and non-separable spatio-temporal changes in
  neurodegeneration via gaussian process convolution.
\newblock In Kanwal Bhatia and Herve Lombaert, editors, \emph{Machine Learning
  Meets Medical Imaging}, pages 35--44, Cham, 2015. Springer International
  Publishing.
\newblock ISBN 978-3-319-27929-9.

\bibitem[Mike~West(1997)]{West_1997}
Jeff~Harrison Mike~West.
\newblock \emph{Bayesian Forecasting and Dynamic Models}.
\newblock Springer-Verlag, 2nd edition, 1997.
\newblock \doi{10.1007/b98971}.
\newblock URL \url{https://doi.org/10.1007%2Fb98971}.

\bibitem[Murray et~al.(2010)Murray, Adams, and MacKay]{murray10}
Iain Murray, Ryan~Prescott Adams, and David~J.C. MacKay.
\newblock Elliptical slice sampling.
\newblock \emph{JMLR: W\&CP}, 9:\penalty0 541--548, 2010.

\bibitem[Neal(2003)]{neal03}
Radford~M. Neal.
\newblock Slice sampling.
\newblock \emph{Annals of Statistics}, 31\penalty0 (3):\penalty0 705--767,
  2003.

\bibitem[Ng et~al.(2012)Ng, Siless, Varoquaux, Poline, Thirion, and
  Abugharbieh]{ng2012}
B.~Ng, V.~Siless, G.~Varoquaux, J.~Poline, B.~Thirion, and R.~Abugharbieh.
\newblock Connectivity-informed sparse classifiers for fmri brain decoding.
\newblock In \emph{2012 Second International Workshop on Pattern Recognition in
  NeuroImaging}, pages 101--104, July 2012.
\newblock \doi{10.1109/PRNI.2012.11}.

\bibitem[Niu et~al.(2015)Niu, Dai, Lawrence, and Becker]{niu2016}
Mu~Niu, Zhenwen Dai, Neil Lawrence, and Kolja Becker.
\newblock Spatio-temporal gaussian processes modeling of dynamical systems in
  systems biology.
\newblock In \emph{18th International Conference on Artificial Intelligence and
  Statistics (AISTATS)}, volume~37. JMLR: W\&CP, 10 2015.
\newblock URL \url{https://arxiv.org/pdf/1610.05163}.

\bibitem[Paciorek and Schervish(2003)]{Paciorek_2003}
Christopher Paciorek and Mark Schervish.
\newblock Nonstationary covariance functions for gaussian process regression.
\newblock In S.~Thrun, L.~Saul, and B.~Sch\"{o}lkopf, editors, \emph{Advances
  in Neural Information Processing Systems}, volume~16. MIT Press, 2003.
\newblock URL
  \url{https://proceedings.neurips.cc/paper/2003/file/326a8c055c0d04f5b06544665d8bb3ea-Paper.pdf}.

\bibitem[Paciorek and Schervish(2006)]{Paciorek_2006}
Christopher~J. Paciorek and Mark~J. Schervish.
\newblock Spatial modelling using a new class of nonstationary covariance
  functions.
\newblock \emph{Environmetrics}, 17\penalty0 (5):\penalty0 483--506, 2006.
\newblock \doi{10.1002/env.785}.
\newblock URL \url{https://doi.org/10.1002%2Fenv.785}.

\bibitem[Report(2018)]{adreport2018}
World~Alzheimer Report.
\newblock The state of the art of dementia research: New frontiers.
\newblock https://www.alz.co.uk/research/world-report-2018, 2018.

\bibitem[Salimbeni and Deisenroth(2017)]{Salimbeni_2017}
Hugh Salimbeni and Marc Deisenroth.
\newblock Doubly stochastic variational inference for deep gaussian processes.
\newblock In I.~Guyon, U.~Von Luxburg, S.~Bengio, H.~Wallach, R.~Fergus,
  S.~Vishwanathan, and R.~Garnett, editors, \emph{Advances in Neural
  Information Processing Systems}, volume~30. Curran Associates, Inc., 2017.
\newblock URL
  \url{https://proceedings.neurips.cc/paper/2017/file/8208974663db80265e9bfe7b222dcb18-Paper.pdf}.

\bibitem[Sarkka et~al.(2013)Sarkka, Solin, and Hartikainen]{sarkka2013}
S.~Sarkka, A.~Solin, and J.~Hartikainen.
\newblock Spatiotemporal learning via infinite-dimensional bayesian filtering
  and smoothing: A look at gaussian process regression through kalman
  filtering.
\newblock \emph{IEEE Signal Processing Magazine}, 30\penalty0 (4):\penalty0
  51--61, July 2013.
\newblock ISSN 1053-5888.
\newblock \doi{10.1109/MSP.2013.2246292}.

\bibitem[Sarkka and Hartikainen(2012)]{sarkka2012}
Simo Sarkka and Jouni Hartikainen.
\newblock Infinite-dimensional kalman filtering approach to spatio-temporal
  gaussian process regression.
\newblock In Neil~D. Lawrence and Mark Girolami, editors, \emph{Proceedings of
  the Fifteenth International Conference on Artificial Intelligence and
  Statistics}, volume~22 of \emph{Proceedings of Machine Learning Research},
  pages 993--1001, La Palma, Canary Islands, 21--23 Apr 2012. PMLR.
\newblock URL \url{http://proceedings.mlr.press/v22/sarkka12.html}.

\bibitem[Senanayake et~al.(2016)Senanayake, O'Callaghan, and
  Ramos]{senanayake2016}
Ransalu Senanayake, Simon~Timothy O'Callaghan, and Fabio~Tozeto Ramos.
\newblock Predicting spatio-temporal propagation of seasonal influenza using
  variational gaussian process regression.
\newblock In \emph{AAAI}, 2016.

\bibitem[Shen et~al.(2010)Shen, Papademetris, and Constable]{shen2010}
X.~Shen, X.~Papademetris, and R.T. Constable.
\newblock Graph-theory based parcellation of functional subunits in the brain
  from resting-state fmri data.
\newblock \emph{NeuroImage}, 50\penalty0 (3):\penalty0 1027--1035, Apr 2010.
\newblock ISSN 1053-8119.
\newblock \doi{10.1016/j.neuroimage.2009.12.119}.
\newblock URL \url{http://dx.doi.org/10.1016/j.neuroimage.2009.12.119}.

\bibitem[Singh et~al.(2010)Singh, Ramos, Whyte, and Kaiser]{singh2010}
A.~Singh, F.~Ramos, H.~D. Whyte, and W.~J. Kaiser.
\newblock Modeling and decision making in spatio-temporal processes for
  environmental surveillance.
\newblock In \emph{2010 IEEE International Conference on Robotics and
  Automation}, pages 5490--5497, May 2010.
\newblock \doi{10.1109/ROBOT.2010.5509934}.

\bibitem[Smola and Kondor(2003)]{smola2003}
Alexander~J. Smola and Risi Kondor.
\newblock Kernels and regularization on graphs.
\newblock In Bernhard Sch{\"o}lkopf and Manfred~K. Warmuth, editors,
  \emph{Learning Theory and Kernel Machines}, pages 144--158, Berlin,
  Heidelberg, 2003. Springer Berlin Heidelberg.
\newblock ISBN 978-3-540-45167-9.

\bibitem[Todescato et~al.(2020)Todescato, Carron, Carli, Pillonetto, and
  Schenato]{TODESCATO2020}
Marco Todescato, Andrea Carron, Ruggero Carli, Gianluigi Pillonetto, and Luca
  Schenato.
\newblock Efficient spatio-temporal gaussian regression via kalman filtering.
\newblock \emph{Automatica}, 118:\penalty0 109032, 2020.
\newblock ISSN 0005-1098.
\newblock \doi{https://doi.org/10.1016/j.automatica.2020.109032}.
\newblock URL
  \url{https://www.sciencedirect.com/science/article/pii/S0005109820302302}.

\bibitem[van~der Vaart and van Zanten(2008)]{vanderVaart08}
A.~W. van~der Vaart and J.~H. van Zanten.
\newblock Rates of contraction of posterior distributions based on gaussian
  process priors.
\newblock \emph{The Annals of Statistics}, 36\penalty0 (3):\penalty0
  1435--1463, 2008.
\newblock ISSN 00905364.
\newblock URL \url{http://www.jstor.org/stable/25464673}.

\bibitem[van~der Vaart and van Zanten(2009)]{vandervaart09}
A.~W. van~der Vaart and J.~H. van Zanten.
\newblock Adaptive bayesian estimation using a gaussian random field with
  inverse gamma bandwidth.
\newblock \emph{Ann. Statist.}, 37\penalty0 (5B):\penalty0 2655--2675, 10 2009.
\newblock \doi{10.1214/08-AOS678}.
\newblock URL \url{https://doi.org/10.1214/08-AOS678}.

\bibitem[van~der Vaart and van Zanten(2011)]{vanderVaart11}
Aad van~der Vaart and Harry van Zanten.
\newblock Information rates of nonparametric gaussian process methods.
\newblock \emph{J. Mach. Learn. Res.}, 12:\penalty0 2095--2119, July 2011.
\newblock ISSN 1532-4435.
\newblock URL \url{http://dl.acm.org/citation.cfm?id=1953048.2021067}.

\bibitem[Wang et~al.(2020)Wang, Hamelijnck, Damoulas, and Steel]{Wang_2020}
Kangrui Wang, Oliver Hamelijnck, Theodoros Damoulas, and Mark Steel.
\newblock Non-separable non-stationary random fields.
\newblock In Hal~Daum{\'e} III and Aarti Singh, editors, \emph{Proceedings of
  the 37th International Conference on Machine Learning}, volume 119 of
  \emph{Proceedings of Machine Learning Research}, pages 9887--9897. PMLR,
  13--18 Jul 2020.
\newblock URL \url{http://proceedings.mlr.press/v119/wang20g.html}.

\bibitem[Wikle(2002)]{Wikle_2002}
Christopher~K Wikle.
\newblock A kernel-based spectral model for non-gaussian spatio-temporal
  processes.
\newblock \emph{Statistical Modelling}, 2\penalty0 (4):\penalty0 299--314,
  2002.
\newblock \doi{10.1191/1471082x02st036oa}.
\newblock URL \url{https://doi.org/10.1191/1471082x02st036oa}.

\bibitem[Wikle and Cressie(1999)]{Wikle_1999}
CK~Wikle and N~Cressie.
\newblock {A dimension-reduced approach to space-time Kalman filtering}.
\newblock \emph{Biometrika}, 86\penalty0 (4):\penalty0 815--829, 12 1999.
\newblock ISSN 0006-3444.
\newblock \doi{10.1093/biomet/86.4.815}.
\newblock URL \url{https://doi.org/10.1093/biomet/86.4.815}.

\bibitem[Wilson and Ghahramani(2011)]{wilson11}
Andrew Wilson and Zoubin Ghahramani.
\newblock Generalised wishart processes.
\newblock In \emph{Proceedings of the Twenty-Seventh Conference Annual
  Conference on Uncertainty in Artificial Intelligence (UAI-11)}, pages
  736--744, Corvallis, Oregon, 2011. AUAI Press.

\bibitem[Zhang and Cressie(2020)]{Zhang_2020}
Bohai Zhang and Noel Cressie.
\newblock Bayesian inference of spatio-temporal changes of arctic sea ice.
\newblock \emph{Bayesian Analysis}, 15\penalty0 (2):\penalty0 605--631, jun
  2020.
\newblock \doi{10.1214/20-ba1209}.

\bibitem[Zhao et~al.(2021)Zhao, Emzir, and S{\"a}rkk{\"a}]{Zhao_2021}
Zheng Zhao, Muhammad Emzir, and Simo S{\"a}rkk{\"a}.
\newblock Deep state-space gaussian processes.
\newblock \emph{Statistics and Computing}, 31\penalty0 (6), sep 2021.
\newblock \doi{10.1007/s11222-021-10050-6}.
\newblock URL \url{https://doi.org/10.1007%2Fs11222-021-10050-6}.

\end{thebibliography}

\end{document}